\documentclass[11pt,letterpaper]{article}

\usepackage[margin=1in]{geometry}
\usepackage[T1]{fontenc}
\usepackage{lmodern}
\usepackage{amsmath,amssymb,amsthm,mathtools}
\usepackage{graphicx,xcolor,booktabs,array}
\usepackage{microtype,needspace}
\usepackage{tikz}
\usetikzlibrary{arrows.meta,calc,positioning,fit}
\usepackage[font=small,labelfont=bf]{caption}
\usepackage[colorlinks=true,allcolors=blue!55!black]{hyperref}
\hypersetup{
  pdftitle={The Intrinsic Cost of Quantum Syndrome Extraction},
  pdfauthor={Chia-Tung Chu},
  pdfsubject={Exact stabilizer measurements, intrinsic cut costs, and space-depth tradeoffs}
}
\numberwithin{equation}{section}
\newtheorem{theorem}{Theorem}[section]
\newtheorem{proposition}[theorem]{Proposition}
\newtheorem{corollary}[theorem]{Corollary}
\newtheorem{lemma}[theorem]{Lemma}
\theoremstyle{definition}
\newtheorem{definition}[theorem]{Definition}
\theoremstyle{remark}
\newtheorem{remark}[theorem]{Remark}

\newcommand{\cS}{\mathcal{S}}

\newcommand{\hw}{\mathrm{hw}}
\newcommand{\depth}{\operatorname{depth}}
\newcommand{\rank}{\operatorname{rank}}

\title{\textbf{The Intrinsic Cost of\\Quantum Syndrome Extraction}}
\author{Chia-Tung Chu\\[0.7em]
  \small Pritzker School of Molecular Engineering, The University of Chicago,\\
  \small Chicago, IL 60637, USA\\[0.3em]
  \small Chicago Quantum Institute, Chicago, IL 60637, USA}
\date{}

\begin{document}
\maketitle
\begin{abstract}
How much time is fundamentally required to extract the syndrome of a stabilizer code on hardware with a given connectivity graph?
Previous bounds answer through a chosen list of checks (Delfosse, Beverland, and Tremblay, 2021) or through the global code parameters (Baspin, Fawzi, and Shayeghi, 2023).
Here, we show that neither determines the interaction that a particular hardware cut requires.
We therefore define the intrinsic cut cost, $\chi$, directly from the stabilizer group, independent of the generating set used to describe the code.
It is computable from the check matrix alone, and equals the mutual information across the cut in the maximally mixed code state.
We show that every exact extraction circuit, however adaptive, needs at least $\chi/2$ two-qubit operations across the cut, and alternatively with free local processing exactly $\chi$ shared Bell pairs.
Guided by this intrinsic cut cost, we construct annular surface codes with constant-depth local extraction that give a counterexample to Corollary 1 of Delfosse, Beverland, and Tremblay.
Furthermore, for a fixed quantum Tanner family, $\chi$ grows at least linearly with the size of every sufficiently small region, giving a depth guarantee at each corresponding bottleneck in every layout.
This strengthens the global perspective of Baspin, Fawzi, and Shayeghi and, together with our matching square-grid construction, yields a tight space--depth tradeoff.
Our results connect the interaction required by syndrome extraction to hardware communication capacity and thus provide a way to assess unavoidable circuit costs before committing to a particular layout.
\end{abstract}
\thispagestyle{plain}
\raggedbottom
\clearpage
\tableofcontents
\clearpage

\section{Introduction}
\label{sec:introduction}

Quantum low-density parity-check (qLDPC) codes can encode logical qubits at a constant rate and with growing distance \cite{panteleevAsymptoticallyGoodQuantum2022,leverrierQuantumTannerCodes2022}, which makes them attractive for low-overhead fault-tolerant quantum computation \cite{gottesmanFaulttolerantQuantumComputation2014,fawziConstantOverheadQuantum2018}.
Realizing such a code on two-dimensional hardware with bounded qubit density, however, requires stabilizer generators whose spatial range grows with the code size.
Otherwise, the bound $kd^2=O(n)$ of Bravyi, Poulin, and Terhal for an $[[n,k,d]]$ code \cite{bravyiTradeoffsReliableQuantum2010} would force constant distance at constant rate.
Therefore, during syndrome extraction, we must implement nonlocal checks using the local operations that the hardware provides.
A circuit can route data with SWAP gates \cite[Section~VII]{baconSparseQuantumCodes2017b}, use networks of additional ancillas \cite{delfosseBoundsStabilizerMeasurement2021}, or consume Bell pairs shared between modules \cite{chandraDistributedQuantumError2026}.
Regardless of the approach, implementing these nonlocal checks incurs a cost in the \emph{depth} of the syndrome extraction circuit.
Bounding this depth is therefore crucial, because greater circuit depth means that errors can be corrected less frequently.
How much depth is needed depends on the code, on the connectivity of the hardware, and on the circuit chosen.

Before this work, two ways of bounding this depth had been developed.
Delfosse, Beverland, and Tremblay (DBT) count the stabilizer checks within a chosen generating list whose support crosses a cut of the hardware graph \cite{delfosseBoundsStabilizerMeasurement2021}.
We call the Pauli operators of such a list the \emph{listed checks}, and we call the list itself a \emph{stabilizer presentation}.
A list may include redundant checks, and each listed check corresponds to a row of its check matrix.
The \emph{weight} of a check counts its nonidentity Pauli factors, and the \emph{incidence} of a qubit counts the listed checks acting nontrivially on it.
A qLDPC family keeps both bounded independently of the code size.
From their count, DBT conclude that qLDPC families whose listed checks have uniformly expanding contracted Tanner graphs, defined before Theorem~\ref{thm:intro-annulus} below, require extraction depth $\Omega(n/\sqrt N)$ on an $N$-qubit square grid.
By contrast, Baspin, Fawzi, and Shayeghi (BFS) track the entanglement that any circuit must create across correctable regions \cite{baspinLowerBoundOverhead2023}.
They derive depth bounds in terms of the code parameters $[[n,k,d]]$ and the device size $N$.
Although both bounds hold for every circuit in their respective measurement models, they use different descriptions of the code.
We show in this work that neither description determines the \emph{interaction} required across a particular hardware cut, which is the quantity that ultimately bounds the depth.

We illustrate the importance of this quantity through the toy example in Figure~\ref{fig:basis-paradox}.
Suppose that a hardware device contains five data qubits $\ell_1,\ell_2,\ell_3,\ell_4,r$ together with several ancillary qubits.
Suppose also that syndrome extraction requires us to measure the four stabilizer checks
\[
\{S_1,S_2,S_3,S_4\}=\{Z_{\ell_1}Z_r,\,Z_{\ell_2}Z_r,\,Z_{\ell_3}Z_r,\,Z_{\ell_4}Z_r\}
\]
during each round.
Consider the cut that divides the data qubits into the set $U=\{\ell_1,\ell_2,\ell_3,\ell_4\}$ and its complement $U^c=\{r\}$.
Following the method of DBT, one would bound the interaction required across this cut by the number of independent listed checks that cross it, which is four.
However, a circuit is free to measure any other generating set of the same stabilizer group, such as
\[
\{T_1,T_2,T_3,T_4\}=\{Z_{\ell_1}Z_r,\,Z_{\ell_1}Z_{\ell_2},\,Z_{\ell_2}Z_{\ell_3},\,Z_{\ell_3}Z_{\ell_4}\}.
\]
The new outcomes are parities of the old ones, so adding up the measured parities recovers the requested syndrome by classical post-processing at no quantum cost.
Such an invertible recombination of the generators changes the list but not the measurement, since the syndrome sectors are unchanged.
At the same time, it reduces the number of cut-crossing checks from four to one, because only $T_1$ crosses.

\begin{figure}[t]
\centering
\resizebox{\textwidth}{!}{
\begin{tikzpicture}[x=1cm,y=1cm,font=\small,
  data/.style={circle,draw=blue!65!black,fill=blue!9,minimum size=4.6mm,inner sep=0pt},
  anc/.style={rectangle,draw=black!65,fill=white,minimum size=4.1mm,inner sep=0pt},
  cross/.style={red!70!black,densely dashed,line width=.9pt},
  localcheck/.style={green!40!black,line width=1.4pt},
  hardware/.style={black!55,line width=.85pt},
  cut/.style={black!60,densely dotted,line width=.9pt},
  title/.style={font=\normalsize\bfseries,align=center},
  note/.style={align=center}]
\node[title] at (2.5,3.3) {(a) Listed checks};
\foreach \i/\y in {1/2.45,2/1.85,3/1.25,4/.65}{
  \coordinate (a\i) at (.55,\y);
  \coordinate (b\i) at (6.35,\y);
  \draw[cross] (a\i) -- (4.45,1.55);
}
\draw[cut] (2.35,.35) -- (2.35,2.75);
\node[fill=white,inner sep=1pt] at (2.35,2.9) {$U\mid U^c$};
\foreach \i in {1,2,3,4}{\node[data] at (a\i) {$\ell_\i$};}
\node[data] at (4.45,1.55) {$r$};
\node[note] at (2.5,-.05) {$S_i=Z_{\ell_i}Z_r$\\four crossing checks};
\node at (5.4,1.55) {$\Longrightarrow$};
\node[title] at (8.3,3.3) {(b) Equivalent generators};
\draw[localcheck] (b1) -- (b2) -- (b3) -- (b4);
\draw[cross] (b1) -- (10.25,1.55);
\draw[cut] (8.15,.35) -- (8.15,2.75);
\node[fill=white,inner sep=1pt] at (8.15,2.9) {$U\mid U^c$};
\foreach \i in {1,2,3,4}{\node[data] at (b\i) {$\ell_\i$};}
\node[data] at (10.25,1.55) {$r$};
\node[note] at (8.3,-.05) {$T_1=S_1,\ \ T_i=S_{i-1}S_i$\\one crossing check: $\chi_U=1$};
\draw[black!20] (11.0,-.3) -- (11.0,3.45);
\node[title] at (14.2,3.3) {(c) A hardware cut};
\filldraw[fill=blue!5,draw=blue!45,rounded corners=4pt] (11.5,.8) rectangle (14.55,2.8);
\node[blue!65!black,anchor=north west,inner sep=1.5pt] at (11.5,2.8) {$L$};
\draw[hardware] (12.0,2.3) -- (14.0,2.3) -- (16.5,2.3);
\draw[hardware] (12.0,1.3) -- (14.0,1.3) -- (16.5,1.3);
\foreach \x in {12.0,13.0,14.0,16.5}{\draw[hardware] (\x,2.3) -- (\x,1.3);}
\node[data] at (12.0,2.3) {$\ell_1$};
\node[data] at (13.0,2.3) {$\ell_2$};
\node[data] at (12.0,1.3) {$\ell_3$};
\node[data] at (13.0,1.3) {$\ell_4$};
\node[anc] at (14.0,2.3) {};
\node[anc] at (14.0,1.3) {};
\node[data] at (16.5,2.3) {$r$};
\node[anc] at (16.5,1.3) {};
\draw[cut] (15.3,.8) -- (15.3,2.8);
\node[fill=white,inner sep=1.5pt] at (15.3,1.8) {$b_L=2$};
\node[note] at (13.0,.3) {$U=Q\cap L=\{\ell_1,\ldots,\ell_4\}$};
\node[data,minimum size=3.6mm] at (15.55,.35) {};
\node[anchor=west,inner sep=1pt] at (15.75,.35) {data};
\node[anc,minimum size=3.3mm] at (15.55,-.15) {};
\node[anchor=west,inner sep=1pt] at (15.75,-.15) {ancilla};
\end{tikzpicture}}
\caption{Equivalent checks and a physical boundary.
The four independent checks $S_i=Z_{\ell_i}Z_r$ in \textbf{(a)} are replaced by $T_1=S_1$ and $T_i=S_{i-1}S_i$ for $i=2,3,4$ in \textbf{(b)}.
Three of the new checks lie entirely in $U=\{\ell_1,\ldots,\ell_4\}$, leaving the intrinsic cut cost $\chi_U=1$.
The lines in these two panels describe check support.
Panel \textbf{(c)} instead shows a hardware realization of the toy example, where circles represent the set $Q$ of data-qubit sites and squares represent ancillary sites.
The vertex region $L$ contains data and ancillas, while $U=Q\cap L$ contains only its data, and $b_L=2$ hardware edges join $L$ to its complement.
Only quantum operations across this physical boundary contribute to the cut cost, and changing generators changes neither that boundary nor the measurement task.}
\label{fig:basis-paradox}
\end{figure}

Global code parameters \cite{baspinLowerBoundOverhead2023}, in turn, have a different limitation, because they do not specify \emph{where} interaction is needed.
For example, append $n$ qubits, each fixed in $|0\rangle$ by an individual $Z$ check, to a good $[[n,k,d]]$ code.
The result is a $[[2n,k,d]]$ code that still has constant rate, linear distance, and bounded check weight and incidence.
Yet the added checks can all be measured independently, so a hardware cut separating the added qubits from the original code requires no crossing operations.
A bound in terms of $n$, $k$, $d$, and $N$ cannot see this difference, since it constrains the device as a whole without distinguishing the data qubits on the two sides of a cut.
Thus, even good global parameters do not guarantee that any interaction is required across a particular cut.

Together, these examples motivate our central question.
\begin{center}
\emph{How much depth does exact syndrome extraction of an arbitrary stabilizer code require on quantum hardware with an arbitrary connectivity graph?}
\end{center}
Here, exact extraction implements $\rho\mapsto\Pi_s\rho\Pi_s$ for each syndrome $s$ and every input density operator $\rho$ of the data register, where $\Pi_s$ projects onto the subspace of states with syndrome $s$.
This requirement holds for every input state, including states entangled with an external reference system on which the extraction circuit does not act.
We model the hardware by a graph whose $N$ vertices are physical qubit sites and whose edges specify the pairs of sites on which a two-qubit operation may act.
A placement assigns the $n$ labeled data qubits to a fixed set $Q$ of sites, with the remaining sites available for ancillas.
A cut in this hardware graph separates a vertex region $L$ from its complement, and its boundary consists of the edges joining the two sides.
To quantify the depth required for this task, we count layers of disjoint one-site and edge-supported two-qubit instruments, while allowing arbitrary classical processing and feedforward at no depth cost.
These instruments include counted gates, measurements, and resets, and later layers may depend on earlier outcomes.
The ancillas start in product states, and every data register is returned to its original physical site at the end.
Section~\ref{sec:models} gives the precise model.

\subsection{Results}

\begin{table}[!b]
\centering
\caption{Summary of results.
All statements concern exact, noiseless extraction.
The depth statements use the adaptive model of Section~\ref{sec:models}, and the Bell-pair count uses the bipartite model with free local processing.
The fixed family is the quantum Tanner family of Theorem~\ref{thm:intro-uniform}, and the family statements apply to sufficiently large members.}
\label{tab:intro-results}
\small
\begin{tabular}{@{}>{\raggedright\arraybackslash}p{0.17\textwidth}>{\raggedright\arraybackslash}p{0.55\textwidth}>{\raggedright\arraybackslash}p{0.2\textwidth}@{}}
\toprule
Result & Statement & Where \\
\midrule
Exact cut cost & On every branch that can occur, every exact extraction circuit performs at least half the intrinsic cut cost in crossing two-qubit operations. With free local processing, the intrinsic cut cost is also the exact shared Bell-pair cost. & Theorems~\ref{thm:adaptive-instrument-cut-bound} and~\ref{thm:single-cut-resource} \\
\addlinespace
Uniform small-region cost & For the fixed family, every data subset up to a fixed fraction of the code size has intrinsic cut cost at least linear in its size. On every placement, this gives a depth guarantee at each bottleneck, proportional to the number of data qubits behind it divided by its nonzero number of boundary edges. & Theorem~\ref{thm:quantum-tanner-target-window}, Corollary~\ref{cor:quantum-tanner-bottleneck} \\
\addlinespace
Matching grid depth & For $n$ data qubits on a square grid of $N\ge n$ sites, any bounded-weight, bounded-incidence check list admits depth $O(\max\{1,n/\sqrt N\})$ on a suitable fixed data placement. For the fixed family, this scaling is tight up to any fixed multiple of $n^2$ hardware sites. & Theorem~\ref{thm:adaptive-grid-width-upper}, Corollary~\ref{cor:fixed-family-grid-width-tight} \\
\addlinespace
Counterexample to listed-crossing bounds & Annular surface codes with expanding bounded-weight, bounded-incidence generating lists admit depth-$13$ local extraction, which gives a counterexample to Corollary~1 of DBT. & Theorem~\ref{thm:annular-counterexample}, Corollary~\ref{cor:dbt-local-expander} \\
\bottomrule
\end{tabular}
\end{table}

We answer the above question using a single quantity that depends only on the stabilizer group, regardless of the list of checks chosen to generate it.
We call this quantity the \emph{intrinsic cut cost} $\chi_U$ and define it as the minimum number of crossing generators over all generating sets of the stabilizer group.
Here, a generator is \emph{crossing} if it acts nontrivially on both the data set $U$ and its complement $U^c$, so $\chi_U=1$ in Figure~\ref{fig:basis-paradox}.
We show three things about this quantity.
First, it determines the interaction that a cut requires.
Every exact extraction circuit, however adaptive, needs at least $\chi_U/2$ two-qubit operations across the cut, and with free local processing exactly $\chi_U$ shared Bell pairs.
Second, for one fixed family of quantum Tanner codes, $\chi_U$ grows at least linearly with $|U|$ on every sufficiently small region.
For every placement of such a code on any hardware graph, this gives a depth guarantee at each bottleneck.
The guarantee is proportional to the number of data qubits behind the bottleneck divided by the number of hardware edges leaving it.
Third, an adaptive compiler attains the resulting depth bound on square grids of every width, which yields a tight space--depth tradeoff for that family.
Guided by $\chi_U$, we also construct annular surface codes with constant-depth local extraction that give a counterexample to Corollary~1 of DBT.
Table~\ref{tab:intro-results} summarizes these results.

\Needspace{12\baselineskip}
To state the results precisely, represent each stabilizer of an $[[n,k,d]]$ code modulo phase by its $X$ and $Z$ components on the $n$ data qubits.
The resulting binary stabilizer space is $\cS\subseteq\mathbb F_2^{2n}$ and has dimension $n-k$.
For a hardware region $L$, which may contain both data and ancillary qubit sites, let $U=Q\cap L$ be its data subset and $U^c=Q\setminus U$ its data complement.
The interaction available to $L$ in one layer is bounded by the number of hardware edges leaving it.
Let $\cS_U$ and $\cS_{U^c}$ be the subspaces of stabilizers supported entirely in $U$ and $U^c$, respectively.
We further give a formula that evaluates the intrinsic cut cost without a search over generating sets,
\begin{equation}
\boxed{\chi_U=\dim\cS-\dim\cS_U-\dim\cS_{U^c}=\rank H_U+\rank H_{U^c}-\rank H},
\label{eq:intro-chi-definition}
\end{equation}
where $H$ is any binary symplectic check matrix of $\cS$, and $H_U$ contains both the $X$ and $Z$ column of each qubit in $U$.
The matrix $H_{U^c}$ contains the remaining columns, and all ranks and dimensions are over $\mathbb F_2$.
The first expression is the intrinsic characterization, which counts the constraints that remain once every constraint measurable within one side has been removed.
The second expression evaluates it with three binary ranks of the stored check matrix, redundant rows included.
Lemma~\ref{lem:minimum-crossing-generators} shows that this dimension is the minimum crossing-generator count defined above.
Interestingly, this quantity is also the mutual information between $U$ and $U^c$ in the maximally mixed code state \cite{fattalEntanglementStabilizerFormalism2004,audenaertEntanglementMixedStabilizer2005} (Lemma~\ref{lem:correctable-cut-entropy}).
Our first result shows that $\chi_U$ is exactly the cost of a cut.

\Needspace{24\baselineskip}
\begin{theorem}[Exact interaction cost of a cut, Theorems~\ref{thm:adaptive-instrument-cut-bound} and~\ref{thm:single-cut-resource} and Remark~\ref{rem:adaptive-sharp-constant}]
\label{thm:intro-cut-cost}
Let $\cS$ be the binary stabilizer space of data qubits occupying fixed sites $Q$ of a hardware graph.
Let $L$ be a region of hardware vertices, let $b$ count the hardware edges with exactly one endpoint in $L$, and put $U=Q\cap L$ and $U^c=Q\setminus U$.
Define $\chi_U=\dim\cS-\dim\cS_U-\dim\cS_{U^c}$, where $\cS_U$ and $\cS_{U^c}$ are the subspaces of stabilizers supported entirely on the respective sides and all dimensions are binary.
A crossing operation acts on an edge leaving $L$, and a nonzero branch has a data Kraus operator that is not identically zero.
\begin{enumerate}
\item\label{it:intro-cut-lower}
Every exact syndrome-extraction circuit that starts with product ancillas and returns the data registers to their original sites uses at least $\chi_U/2$ crossing two-qubit operations on every nonzero branch.
If every crossing operation is a CNOT, it uses at least $\chi_U$.
Its depth is therefore at least $\chi_U/(2b)$ when $b>0$.
When $b=0<\chi_U$, exact extraction is impossible.
\item\label{it:intro-cut-upper}
In the bipartite resource model, suppose all quantum operations within either side and all classical communication are free.
Each side has sufficiently many ancillary qubits to store Bell-pair halves and other intermediate quantum states.
Then exact extraction is possible with $\chi_U$ initially shared Bell pairs and no crossing quantum operations.
If a crossing link is available, starting instead with product ancillas, it is possible with $\chi_U$ crossing CNOTs or with $\lceil\chi_U/2\rceil$ arbitrary crossing two-qubit gates.
It is impossible with fewer of each.
\end{enumerate}
\end{theorem}

Part~\ref{it:intro-cut-lower} requires no assumption on the distance, the rate, or the gate set.
Its depth bound improves to $\chi_U/b$ when all crossing operations are CNOTs, CZs, or single two-qubit Pauli measurements (Remark~\ref{rem:adaptive-sharp-constant}).
Part~\ref{it:intro-cut-upper} shows that the bound is tight for the crossing resource, so $\chi_U$ is exactly the Bell-pair cost of the cut in the stated resource model.
It does not bound the elapsed depth of a real device, where moving states to a port and processing them locally take additional layers.
We return to this distinction in the discussion below.

To see why these bounds hold, recall that exact extraction must succeed even when the data are entangled with an external reference system.
We may therefore prepare each data qubit in a Bell state with a reference qubit assigned to the same side of the cut, so that initially no entanglement crosses the cut.
The references are untouched by the circuit and are not working ancillas.
The ideal output for syndrome $s$ is then a pure state, proportional to the vectorized projector $|\Pi_s\rangle\!\rangle$, with exactly $\chi_U$ ebits across the cut.
More precisely, its Schmidt spectrum across the cut is flat and of rank $2^{\chi_U}$ (Lemma~\ref{lem:projector-choi-schmidt}), matching the mutual-information value noted above.
This state must appear on every outcome branch of nonzero probability.
On a complete branch, the observed outcomes fix all feedforward choices, and each instrument and final ancilla discard is resolved into a single Kraus operator for the proof.
A branch is \emph{nonzero} if its resulting operator on the data is nonzero, equivalently if it has positive probability on this reference input.
The extra Kraus indices used in this refinement are not available to the controller.
To track how the entanglement is built up along such a branch, we follow the Schmidt rank across the cut, which starts at one and must reach $2^{\chi_U}$.
Operations within one side cannot increase this rank.
A crossing two-qubit operation multiplies it by at most four, while a crossing CNOT, CZ, or two-qubit Pauli measurement multiplies it by at most two.
These factors follow from the operator Schmidt rank, the minimum number of product operators needed to express a Kraus operator across the cut \cite{nielsenQuantumDynamicsPhysical2003}.
Each pre-shared Bell pair likewise supplies a factor of two, which yields the three lower bounds.
Conversely, a generating set with only $\chi_U$ crossing generators always exists (Lemma~\ref{lem:minimum-crossing-generators}).
Each crossing generator can be measured with one shared Bell pair by the standard protocol for measuring a Pauli product nonlocally \cite{beckmanCausalLocalizableQuantum2001,limTradeoffInformationGain2025}, so $\chi_U$ Bell pairs suffice.
The argument also survives approximation.
We use squared fidelity, defined by $F(\rho,\sigma)=(\operatorname{Tr}\sqrt{\sqrt\sigma\rho\sqrt\sigma})^2$.
Let $F>0$ be this fidelity between the actual and ideal outputs on the reference input, including their classical syndrome registers.
For $b>0$, the instrument still requires depth at least $(\chi_U-\log_2(1/F))/(2b)$ (Corollary~\ref{cor:approximate-instrument-cut-bound}).

This use of reference entanglement connects our argument to earlier lower bounds on the entanglement cost of nonlocal measurements \cite{bandyopadhyayEntanglementCostNonlocal2009}.
The Bell-pair statement of Part~\ref{it:intro-cut-upper} also follows from the general entanglement-cost theorems of Lim, Hhan, and Kwon \cite{limTradeoffInformationGain2025} and of Akibue, Miyazaki, and Osaka \cite{akibueOptimizingEntanglementManipulation2026}, as Appendix~\ref{subsec:single-cut-prior-work} shows.
The present argument adds three things.
It identifies the stabilizer invariant that these results share when the task is syndrome extraction, and it computes that invariant without searching over generating sets.
It also counts each crossing primitive on every nonzero branch by its own operator Schmidt rank.

Since Eq.~\eqref{eq:intro-chi-definition} requires only three rank computations, the cost of a cut can be computed from the check matrix before any circuit is designed.
For codes of current interest, these costs differ from what the listed checks suggest.
Divide the $[[144,12,12]]$ gross code \cite{bravyiHighthresholdLowoverheadFaulttolerant2024} into two halves of its torus.
The stabilizer space has dimension $132$, and each half supports $36$ independent stabilizers, so $\chi_U=132-36-36=60$.
A remote CNOT, which consumes one Bell pair, per crossing check--data incidence would use $144$ Bell pairs, while a parity measurement per crossing listed check would use $72$.
The optimum of $60$ is attained simply by deleting twelve dependent checks and reconstructing the omitted syndrome bits classically, so the optimal circuit remains sparse (Proposition~\ref{prop:gross-sparse-single-cut-optimum} and Table~\ref{tab:gross-resource-comparison}).
By contrast, the mitten codes of Bhardwaj et al.\ \cite{bhardwajHighrateQLDPCProcessors2026} are built for modular processors from four ``finger'' blocks and one ``thumb'' block of $g$ qubits each.
Every one of their $4g$ listed checks crosses the cut between thumb and fingers, yet $\chi_U=2g$ for all eight members of the mitten family, from $[[150,30,10]]$ to $[[975,195,{\le}24]]$ (Section~\ref{subsec:mitten-thumb}).
In other words, a modular layout along this cut needs only half the Bell pairs that the listed checks would indicate.
On the $150$-qubit member, however, a generating set with only $2g$ crossing generators must contain finger generators of weight at least $28$.
The cost of a cut and the sparsity of the attaining circuit can therefore pull in different directions.
Since $\chi_U$ bounds every list at once, it also complements compilers for a fixed check list \cite{zhangOptimalCompilationSyndrome2026} and distributed implementations of a chosen list \cite{chandraDistributedQuantumError2026,shawNetworkedRealizationQuantum2026}, which optimize one list at a time.

Having settled one cut, we next ask when Theorem~\ref{thm:intro-cut-cost} constrains an entire device.
Its strength depends on how large $\chi_U$ is on the regions that a bottleneck can isolate.
The appended-qubit example above shows that the code parameters, which are all that the bound of BFS uses, do not guarantee this.
We therefore look for a family in which every small region costs at least a fixed fraction of its size.
We find one among the quantum Tanner codes of Leverrier and Z\'emor \cite{leverrierQuantumTannerCodes2022}, under the full hypotheses of their construction.
A nontrivial point is that $\cS_U$ contains not only the checks inside $U$ but also products of checks whose factors outside $U$ cancel.
Using the reduction theorem of that construction together with expansion, we give every such stabilizer a unique short description in terms of local check components.
Bounding the dimension of all such descriptions yields the following volume law.

\begin{theorem}[Uniform small-region cost, Theorem~\ref{thm:quantum-tanner-target-window}]
\label{thm:intro-uniform}
Fix a quantum Tanner family of Leverrier and Z\'emor with sufficiently large fixed local degree, under the full hypotheses of their construction.
There are constants $\eta,\tau>0$ with the following property.
Every sufficiently large member and every set $U$ of at most $\eta n/2$ data qubits satisfy $\chi_U\ge2\tau|U|$.
Every pure code state has entanglement entropy at least $\tau|U|$ across such a $U$.
\end{theorem}

\begin{corollary}[Quantum Tanner bottleneck bound, Corollary~\ref{cor:quantum-tanner-bottleneck}]
\label{cor:intro-bottleneck}
Place the data qubits of such a code arbitrarily on any hardware graph.
A hardware region holding $0<u\le\eta n/2$ data qubits behind $b>0$ hardware edges requires extraction depth at least $\tau u/b$.
Behind no edges, exact extraction is impossible.
\end{corollary}

The corollary concerns each individual small region of each placement, whatever the geometry inside the region or the degree of the hardware.
It thus strengthens the global perspective of BFS to a guarantee at each bottleneck of every layout.
Their framework, combined with data-weighted separators, already gives the aggregate scaling $\Omega(n/\sqrt N)$ for good codes on bounded-degree planar hardware (Appendix~\ref{subsec:bfs-weighted-partitions}), but it does not control a specific region.
Earlier entanglement results bound averages over correctable partitions at positive rate \cite{kimConditionalIndependenceQuantum2013} or over selected small subsets of CSS codes \cite{zhaoGraphbasedApproachEntanglement2026}.
Theorem~\ref{thm:intro-uniform} instead gives a volume law for every region in its range and every code state.

On a square grid with $N$ sites and $n/\sqrt N$ large, a vertical line isolates $\Theta(n)$ data qubits behind $O(\sqrt N)$ hardware edges.
Corollary~\ref{cor:intro-bottleneck} therefore forces depth $\Omega(n/\sqrt N)$ on every placement of the fixed family, region by region rather than only in aggregate.
When $n/\sqrt N$ is bounded, the same order follows from the need for at least one quantum operation (Corollary~\ref{cor:intrinsic-patch}).
Until now, this scaling was only a lower bound, with matching constructions known just at the two endpoints.
With linear space, mesh permutations implement a constant number of gate layers in depth $O(\sqrt n)$ \cite{baconSparseQuantumCodes2017b}.
With quadratic space, the construction of DBT reaches constant depth for bounded-degree CSS presentations \cite{delfosseBoundsStabilizerMeasurement2021}.
Building on Corollary~\ref{cor:intro-bottleneck}, we fill the gap between these endpoints with an adaptive compiler that measures any bounded-weight, bounded-incidence presentation in this depth on a suitably chosen fixed placement.
The compiler measures each check in three steps.
It concentrates the parity of the check onto one data qubit with CNOTs, measures that qubit, and undoes the concentration.
The branch operator is then exactly the check projector.
On a grid wide enough for one $4\times4$ tile per data qubit, the compiler implements the required remote CNOTs with Bell pairs.
Each tile prepares its Bell pairs locally, and one layer of simultaneous Bell measurements fuses them.
In this way $O(1+n/\sqrt N)$ batches of constant depth suffice.
On a narrower grid, ordinary mesh permutations bring the gate pairs together in depth $O(\sqrt N)$.

\begin{theorem}[Adaptive extraction at every grid width, Theorem~\ref{thm:adaptive-grid-width-upper} and Corollary~\ref{cor:fixed-family-grid-width-tight}]
\label{thm:intro-grid}
Let a stabilizer presentation on $n$ data qubits have checks of weight at most $w$, with at most $\delta$ checks incident on each qubit.
For every full square grid with $N\ge n$ sites, there is a fixed placement of the data and an exact extraction circuit of depth $O_{w,\delta}(\max\{1,n/\sqrt N\})$.
Every quantum operation is counted, including preparation, routing, measurement, reset, correction, and the return of every data qubit to its site.
For the family of Theorem~\ref{thm:intro-uniform}, the minimum depth over circuits and fixed placements is $\Theta(\max\{1,n/\sqrt N\})$ throughout $n\le N\le cn^2$, for any fixed $c>0$ and with constants depending on $c$.
Equivalently, $ND_{\min}^2=\Theta(n^2)$ in that range.
\end{theorem}

The theorem interpolates between the two endpoints, including the case $N=n$ in which no site is spare.
Together with Corollary~\ref{cor:intro-bottleneck}, it shows that the planar scaling, which the framework of BFS also yields, is \emph{tight} for this family.
We note that the tradeoff describes the best placement rather than every placement.
A prescribed placement that concentrates $u\le\eta n/2$ data qubits behind $b>0$ edges still pays depth $\tau u/b$ by Corollary~\ref{cor:intro-bottleneck}, however large the rest of the grid.

Finally, we revisit the listed-crossing lower bound of DBT.
Their Theorem~1 bounds the depth of a local Clifford extraction circuit by the number of independent listed checks crossing a hardware cut.
Corollary~1 of their work then states that qLDPC families whose fixed generating lists have uniformly expanding contracted Tanner graphs require local Clifford extraction depth $\Omega(n/\sqrt N)$ on an $N$-qubit square grid.
Here, the contracted Tanner graph of a list joins two data qubits whenever a listed check acts nontrivially on both.
Expansion means that every sufficiently small set of qubits has at least a fixed positive number of graph edges leaving it per qubit.
The following theorem provides a counterexample to that corollary.

\begin{theorem}[Expanding presentations with constant-depth extraction, Theorem~\ref{thm:annular-counterexample} and Corollary~\ref{cor:dbt-local-expander}]
\label{thm:intro-annulus}
For every $L\ge1$ there is a surface code on a square annulus with parameters $[[n_L,1,L+1]]$, where $n_L=16L^2+8L$, together with an independent generating list of its stabilizer group with the following properties.
Every generator has weight at most $132$, and every qubit lies in at most $36$ generators.
The contracted Tanner graph of the list has at least $|A|/18$ edges leaving every set $A$ of at most $n_L/2$ qubits.
On the $(6L+1)\times(6L+1)$ square grid, with $N_L=(6L+1)^2$ sites, a nearest-neighbor Clifford circuit of depth at most $13$ implements the exact syndrome instrument of this list in the measurement model of DBT.
The family satisfies the hypotheses of Corollary~1 of \href{https://arxiv.org/abs/2109.14599v1}{DBT} while $n_L/\sqrt{N_L}\to\infty$.
The conclusion $\depth=\Omega(n_L/\sqrt{N_L})$ of that corollary therefore does not hold.
\end{theorem}

The construction simply applies the recombination of Figure~\ref{fig:basis-paradox} at scale to the square annulus of Section~\ref{sec:presentation-expansion} and Appendix~\ref{subsec:annulus-code}.
Color the vertices of the annulus like a checkerboard.
The vertex checks at the black vertices then have disjoint supports.
Multiply each of them by the black checks that precede it and are adjacent to it in a bounded-degree expander.
This keeps weight and incidence bounded and makes the contracted Tanner graph expand, while the stabilizer group is unchanged.
The original constant-depth circuit therefore still measures the new list after classical post-processing of the measurement outcomes.
Across a vertical line that splits the data in half, $\chi_U$ stays $O(L)$, whereas the number of listed checks crossing that line grows linearly with $n_L$.
This gap separates the listed crossing count that enters Theorem~1 of DBT from the interaction that the measurement requires.
Several annuli side by side extend the family to parameters $[[m(16L^2+8L),m,L+1]]$ with $kd^2=\Theta(n)$ (Corollary~\ref{prop:many-annuli}), the most that two-dimensional locality allows \cite{bravyiTradeoffsReliableQuantum2010}.
This counterexample shows that expansion of a chosen generating list need not imply a large intrinsic cut cost, so it cannot substitute for a lower bound on $\chi_U$.

Taken together, these results establish the intrinsic cut cost $\chi_U$ as the quantity that governs syndrome extraction across a cut.
It is a stabilizer invariant computable from three ranks of the check matrix, and it also quantifies the entanglement that any exact syndrome extraction must create across the cut.
Our examples show that counting the listed checks crossing a cut can overestimate the required interaction, while the global code parameters alone can miss this cut-dependent requirement entirely.
Our results thus connect the interaction required by syndrome extraction to the communication capacity of a device.
They provide a common quantitative basis for comparing stabilizer code constructions, evaluating syndrome extraction compilers, and assessing candidate hardware layouts before committing to a particular design.

\subsection{Discussion}

We first note that the single-cut results extend in two directions.
When the data are distributed over processors connected by a tree, the edge budgets $b_e\ge\chi_{U_e}$ can all be met at once.
The resource that meets them is the vectorized syndrome projector, prepared with the pure-state tree preparation theorem of Yamasaki, Soeda, and Murao \cite[Theorem~5]{yamasakiGraphAssociatedEntanglement2017} (Proposition~\ref{prop:network-tree-resource}).
Appendix~\ref{app:certified-circuits} moreover extends the cut bound to weighted certificates and compares them with verified extraction circuits for twelve published codes on fifteen hardware graphs.

We emphasize that all of our results concern ideal, noiseless circuits.
The lower bounds are stated for the most permissive such model.
Adaptivity, intermediate measurements of data, and unlimited classical communication are allowed, and Theorem~\ref{thm:intro-cut-cost} holds on every nonzero branch.
The constructions, on the other hand, are ideal-circuit upper bounds.
Theorem~\ref{thm:intro-grid} relies on long chains of Bell measurements whose errors would accumulate under noise.
Part~\ref{it:intro-cut-upper} of Theorem~\ref{thm:intro-cut-cost} optimizes the crossing resource rather than the elapsed depth of local processing.
Neither establishes fault tolerance, which we leave open together with the questions in Section~\ref{sec:discussion}.
Under noise, moreover, the redundant checks that our optimization discards can supply consistency information.
Its value depends on the noise model and the decoder, which raises the question of single-shot correction \cite{bombinSingleshotFaulttolerantQuantum2015,guSingleshotDecodingGood2024}.

To place our bounds among the earlier ones, consider first the work of BFS.
Their entanglement-growth method also applies to our reference input, and their partition framework with data-weighted separators yields the aggregate planar bound $\Omega(n/\sqrt N)$ for good codes without any uniform theorem.
Neither the reference input nor the planar scaling by itself therefore distinguishes the present argument.
Section~\ref{sec:static-2d-barrier} together with Appendix~\ref{subsec:bfs-weighted-partitions} separates the two conclusions and their assumptions.
What the present approach adds is the intrinsic quantity itself, the counts for each crossing primitive on each branch, the uniform small-region estimate, and the matching constructions.
Our comparison with DBT concerns Theorem~1 and Corollary~1 of their paper.
We use that paper's measurement model, which permits classical parity relabeling.
Appendix~\ref{subsec:basis-change-counterexample} gives simpler examples and identifies the corresponding steps in the proof.
The identification of $\chi_U$ with mutual information descends from the stabilizer entropy calculations of Fattal et al.\ and of Audenaert and Plenio \cite{fattalEntanglementStabilizerFormalism2004,audenaertEntanglementMixedStabilizer2005}.
The conversions between crossing CNOT or SWAP gates and shared Bell pairs used in the upper bounds are those of Eisert et al.\ \cite{eisertOptimalLocalImplementation2000}.

The rest of the paper is organized as follows.
Section~\ref{sec:models} defines the measurement task and the circuit model.
Section~\ref{sec:stabilizer-cut-rank} introduces $\chi_U$, proves the cut bound and its attainability, and treats approximate instruments.
Section~\ref{sec:worked-example-tree} works out the gross and mitten codes and tree networks, and Section~\ref{sec:presentation-expansion} constructs the annular examples.
Section~\ref{sec:separator-barriers} states and explains the uniform small-region theorem.
Section~\ref{sec:static-2d-barrier} gives the grid lower bound, the compiler, and the tradeoff, and Section~\ref{sec:discussion} collects open questions.
The appendices supply the supporting algebra, the attaining protocols and tree resources, the annular constructions, the uniform Tanner theorem with its geometric bounds, the grid compiler, and weighted-cut certificates for published extraction circuits.

\section{Preliminaries and the measurement task}
\label{sec:models}

\subsection{Checks, syndromes, and preserved information}

Consider measuring $Z_1Z_2$ on two data qubits.
The two projectors are $(I+Z_1Z_2)/2$ and $(I-Z_1Z_2)/2$, onto the even and odd parity subspaces.
An input $\alpha|00\rangle+\beta|11\rangle$ therefore has even parity with certainty and is left unchanged.
Measuring $Z_1$ and $Z_2$ separately and then multiplying their outcomes reports the same parity.
The separate measurements, however, also reveal which of $|00\rangle$ and $|11\rangle$ occurred, and so destroy the coherence between them.
An ancilla-mediated parity measurement avoids this extra information.
Its ancilla is prepared in $|0\rangle$ and made the target of a CNOT from each data qubit.
The ancilla is then measured in the $Z$ basis.
This circuit records only parity and preserves the corresponding data subspace.

A stabilizer code extends this idea to a commuting collection of Hermitian Pauli operators on $n$ qubits.
A Pauli operator is a tensor product of $I,X,Y,Z$, with a sign, and its weight is the number of nonidentity factors.
We fix a commuting Pauli group that does not contain $-I$, and take its common $+1$ eigenspace as the code space.
If the group has $r$ independent generators, the code space has dimension $2^{n-r}$ and encodes $k=n-r$ logical qubits.
The distance $d$ is the minimum weight of a Pauli that preserves the code space but does not act as a scalar within it.
We use $[[n,k,d]]$ for these parameters.
A Calderbank--Shor--Steane (CSS) code admits generators that separate into $X$-only and $Z$-only checks.

To describe changes of generators, we work modulo phase, where Pauli multiplication is addition in $\mathbb F_2^{2n}$.
The two binary coordinates at a qubit specify its $X$ and $Z$ factors, and their sum represents $Y$ up to phase.
The resulting stabilizer space is denoted by $\cS$, with $\dim\cS=r$.
Commutation is expressed by the symplectic product $\langle(x,z),(x',z')\rangle=x\cdot z'+z\cdot x'$ over $\mathbb F_2$.
For a physical measurement, however, we keep the signed Pauli representative, even though signs do not enter the binary dimensions and ranks.
All vector-space dimensions and matrix ranks below are binary, and all entropies and logarithms use base two.

Choose independent signed generators $P_1,\ldots,P_r$.
The syndrome $s\in\mathbb F_2^r$ records whether each eigenvalue is $+1$ or $-1$, and its eigenspace projector is
\begin{equation}
\Pi_s=\prod_{j=1}^r\frac{I+(-1)^{s_j}P_j}{2}.
\label{eq:syndrome-projector}
\end{equation}
Each sector has dimension $2^k$ and can contain arbitrary logical quantum information.
The ideal measurement sends an input density operator $\rho$ to the unnormalized conditional state $\Pi_s\rho\Pi_s$ and reports $s$.
The trace of this conditional state is the outcome probability, and dividing by that trace gives the post-measurement state.
The collection of these completely positive maps is the joint L\"uders instrument \cite{knillTheoryQuantumErrorcorrecting1997}.
An exact extraction must implement the full instrument on every input, including inputs entangled with an external reference.

Changing generators does not change the sector projectors.
For two independent bases, the syndrome bits are related by an invertible binary transformation.
The relation may include a fixed sign offset determined by the chosen Pauli representatives.
A redundant list adds only determined parity functions of these bits.
Circuits may therefore choose a convenient generating set and reconstruct the requested syndrome classically, provided they also preserve the conditional quantum state.
This freedom concerns the ideal instrument, whereas adding redundant checks for protection against measurement noise is a different design problem \cite{campbellTheorySingleshotError2019}.

\subsection{The hardware and its cost}

We represent the device by a graph whose vertices are physical qubit sites and whose edges permit two-qubit operations.
A hardware region of this graph can contain both data and ancillas.
Only its data qubits enter the intrinsic cut cost, whereas all of its vertices constrain the circuit.

\begin{definition}[Hardware graph and boundary]
\label{def:hardware-boundary}
Let $G_{\hw}=(V_{\hw},E_{\hw})$ be a hardware graph of maximum degree $\Delta_{\hw}$, with fixed data sites $Q\subseteq V_{\hw}$.
Write $n=|Q|$ and $N=|V_{\hw}|$.
For a vertex region $L\subseteq V_{\hw}$, its data subset is $U=Q\cap L$ and its edge boundary is
\begin{equation}
\partial_{G_{\hw}}L
=
\{uv\in E_{\hw}:u\in L,\ v\notin L\}.
\end{equation}
The two-qubit operations in one layer have disjoint supports, so at most $|\partial_{G_{\hw}}L|$ cross this boundary.
\end{definition}

To cover measurements as well as gates, we describe each primitive as an instrument.
An instrument may have several classical outcomes and one or more Kraus operators for each outcome.
This language includes a unitary gate, a measurement, and a reset as special cases.
We allow this generality because a lower bound should not depend on a particular gate synthesis or ancilla readout scheme.

\begin{definition}[Exact adaptive hardware-local extraction]
\label{def:adaptive-local-extraction}
An exact extraction circuit for $\cS$ starts with product ancillas uncorrelated with the data.
It is a finite circuit in layers of disjoint quantum primitives, each a finite-outcome one-qubit instrument or a finite-outcome two-qubit instrument on a hardware edge.
Preparations, measurements, resets, and discards are allowed, with at most one qubit at each vertex.
Previous observed outcomes may determine later primitives and their locations.
Classical processing, randomness, memory, and communication on finite registers are free, but supply no uncounted quantum operation or shared quantum state.
The circuit reports a complete syndrome $s=f(O)$ as a function of its observed record $O$ and implements $\rho\mapsto\Pi_s\rho\Pi_s$ for each $s$.
Every data label returns to its original physical site at the output, and routing and restoration use counted operations.
The depth $\depth(C)$ is the largest number of quantum layers on a nonzero branch, including preparation, measurement, reset, and quantum correction.
\end{definition}

The product-ancilla assumption excludes a supply of free entanglement across a cut.
The fixed input and output sites ensure that moving a data state around a boundary cannot change the task without paying for the movement.
Classical communication is unrestricted even between distant vertices.
The model therefore measures quantum depth separately from decoder runtime and classical signal delay.

For comparison with earlier results, we also use two restricted classes of circuits, which do not change the main lower bound.

\Needspace{9\baselineskip}
\begin{definition}[Clifford subclass for comparison]
\label{def:local-clifford-stabilizer-measurement}
The Clifford subclass uses $|0\rangle$ and $|+\rangle$ preparations, Clifford unitaries, Pauli measurements, and classically controlled Pauli corrections.
A Clifford unitary maps Pauli operators to Pauli operators under conjugation.
Each two-qubit Pauli measurement is a counted primitive on a hardware edge.
The DBT v1 model further excludes single-qubit data measurements and imposes its stated ancilla input and output conditions, and its classical parity controls and output parities may overlap \cite[Sections~2.1--2.2]{delfosseBoundsStabilizerMeasurement2021}.
\end{definition}

\begin{definition}[Unitary extraction with final readout]
\label{def:swap-only-compilation}
This subclass applies only hardware-local one- and two-qubit unitaries before a final layer of local measurements, with no intermediate measurement, reset, or adaptive quantum correction.
Product ancillas and the fixed data interface are as in Definition~\ref{def:adaptive-local-extraction}, and all routing gates are counted.
\end{definition}

The grid upper bound will use the broader adaptive model, including intermediate data measurements.
The annular counterexample, in contrast, can be implemented within the DBT subclass just specified.

\subsection{How much correctness must be tested?}

Although the definition of exact extraction asks for equality of quantum instruments, exactness can equivalently be tested on pure inputs in individual syndrome sectors.
It suffices to require the correct syndrome and preservation of every pure state in every sector.
Testing a smaller promised set, in contrast, can be insufficient.

\begin{remark}[Sectorwise exactness and Pauli frames]
\label{rem:sectorwise-exactness}
Suppose the same physical circuit reports the correct syndrome with certainty and preserves every pure state in every syndrome sector.
Then each refined Kraus operator reporting $s$ is $K_\beta=c_\beta\Pi_s$, and the circuit implements the full ideal instrument, including on superpositions of sectors and reference-entangled inputs.
If correctness instead holds after removing a physical Pauli frame $P_o$ determined by the observed record $o$, then $P_o^\dagger K_\beta=c_\beta\Pi_s$ for every refinement of that record.
Such a correction factors into single-qubit operators across every hardware cut, so it does not change any crossing-resource lower bound.
A correction selected using a hidden Kraus index, or an arbitrary entangling correction, is not part of this allowance.
\end{remark}

For the first assertion, correct labeling forces a branch to vanish on all wrong sectors.
Preservation of every pure state in the correct sector then forces its restriction there to be a scalar multiple of the identity.
The reason is that applying the restriction to two basis states and their superposition makes the scalars equal.
Proposition~\ref{prop:sectorwise-instrument-rigidity} gives the full argument, including observed-record corrections.
Requiring only one known sector is weaker, however, because the identity circuit can report that fixed syndrome without any interaction.
Preserving selected basis states likewise does not ensure preservation of their superpositions.

\section{The intrinsic cost of a cut}
\label{sec:stabilizer-cut-rank}

Split the data into two sets, $U$ and $U^c$.
Some stabilizers can then be measured using only the qubits in $U$, and others using only $U^c$.
The remaining independent constraints require the two sides to participate in a common measurement.
We first count these constraints without choosing a generating set, and then show that the count determines exactly the crossing resource needed when local processing is free.

\subsection{A generator-independent count}

Let $\cS_U$ be the subspace of stabilizers supported entirely in $U$, and define $\cS_{U^c}$ in the same way.
Their intersection is zero because their supports are disjoint.
We define
\begin{equation}
\chi_U(\cS)=\dim\cS-\dim\cS_U-\dim\cS_{U^c},
\label{eq:intro-chi}
\end{equation}
which is the dimension of the quotient $\cS/(\cS_U+\cS_{U^c})$.
The quotient identifies stabilizers that differ only by constraints measurable within the two sides.
Hence $\chi_U$ counts the independent constraints that remain after all local ones have been removed.

\begin{lemma}[Minimum crossing-generator count]
\label{lem:minimum-crossing-generators}
For every data partition $Q=U\sqcup U^c$, $\chi_U(\cS)$ is the minimum number of generators with support on both sides, over all generating sets of $\cS$.
The same minimum is attained among bases.
\end{lemma}

\begin{proof}
The noncrossing generators lie in $\cS_U\oplus\cS_{U^c}$.
The crossing generators of any generating set must therefore span the quotient $\cS/(\cS_U+\cS_{U^c})$.
Spanning it requires at least $\chi_U$ of them.
Conversely, concatenate bases of $\cS_U$ and $\cS_{U^c}$ and extend them to a basis of $\cS$.
Exactly $\chi_U$ vectors must be added.
Every added vector crosses the cut, because a vector supported on one side would already be in the span.
\end{proof}

In Figure~\ref{fig:basis-paradox}, the full space has dimension four, the subspace supported on the left has dimension three, and that on the right has dimension zero.
Thus $\chi_U=1$, although all four checks of the first presentation cross.
The minimum places no restriction on the weight or incidence of the new generators, and the basis that attains it may depend on the cut.
What the minimum supplies is an intrinsic lower-bound quantity that does not assume that one sparse basis minimizes every cut simultaneously.

To evaluate this minimum, it suffices to use the listed checks, including redundant ones.
Restricting a stabilizer to $U$ forgets its action on $U^c$, and the kernel of this restriction is precisely $\cS_{U^c}$.
Rank-nullity therefore gives the required local dimensions and then the intrinsic cut cost.

\begin{lemma}[Rank formula]
\label{lem:cross-cut-rank-formula}
Let $H$ be any binary symplectic check matrix with row space $\cS$, allowing dependent rows.
Let $H_U$ contain both the $X$ and $Z$ column of each qubit in $U$, and let $H_{U^c}$ contain the remaining columns.
Then
\begin{equation}
\chi_U(\cS)
=
\rank(H_U)+\rank(H_{U^c})-\rank(H).
\label{eq:cross-cut-rank-formula}
\end{equation}
\end{lemma}

Indeed, $\dim\cS_U=\rank(H)-\rank(H_{U^c})$ and $\dim\cS_{U^c}=\rank(H)-\rank(H_U)$, so substitution into Eq.~\eqref{eq:intro-chi} proves the formula.
Three binary Gaussian eliminations therefore evaluate a specified cut, and no search over generating sets is needed.
Searching for the strongest cut of a large device is a separate optimization problem, but any evaluated cut gives a valid bound.

\subsection{Correlations and the measurement projector}

The same quantity measures correlations in the code space.
Let $\omega=\Pi_0/2^k$ be the maximally mixed code state, which treats every encoded state on equal footing.
For a density operator $\rho$, write $S(\rho)=-\operatorname{Tr}(\rho\log_2\rho)$.
The mutual information $I(U:U^c)_\rho=S(\rho_U)+S(\rho_{U^c})-S(\rho)$ then measures the total correlations between the two sides.

\Needspace{13\baselineskip}
\begin{lemma}[Code-state mutual information and correctable regions]
\label{lem:correctable-cut-entropy}
For the maximally mixed state $\omega$ of an $[[n,k,d]]$ stabilizer code,
\begin{equation}
I(U:U^c)_\omega=\chi_U(\cS)
\label{eq:cut-rank-code-state-mi}
\end{equation}
for every data region $U$.
If erasure of $U$ is correctable, in particular if $|U|<d$, then
\begin{equation}
\chi_U(\cS)
=2S(\omega_U)
=2\bigl(|U|-\dim\cS_U\bigr).
\label{eq:correctable-cut-entropy}
\end{equation}
Every pure code state $|\psi\rangle$ then has $\psi_U=\omega_U$, so its entanglement entropy across the cut is $S(\psi_U)=\chi_U(\cS)/2$.
\end{lemma}

The first identity follows from the usual stabilizer entropy calculation \cite{fattalEntanglementStabilizerFormalism2004,audenaertEntanglementMixedStabilizer2005}.
Tracing out $U^c$ removes every Pauli term whose support reaches outside $U$.
The remaining state is therefore a normalized projector with $2^{|U|-\dim\cS_U}$ equal nonzero eigenvalues, and its entropy is $|U|-\dim\cS_U$.
Adding the two marginal entropies and subtracting $k$ gives $\chi_U$.
Correctability means that the erased region carries no information about the encoded state, that is, its reduced state is the same for every code state.
Appendix~\ref{subsec:correctable-cut-entropy} proves the resulting factor-of-two identity.
For a general, noncorrectable region, the mutual information can include classical correlations and need not equal twice a pure code state's entanglement entropy.

Classical communication, however, can create correlations for free in our model.
The mutual information of $\omega$ is therefore not by itself the circuit lower-bound witness.
To obtain a circuit lower bound, we instead examine the entanglement created by the measurement on a particular input.
Pair every data qubit with a private reference qubit in a Bell state, and keep each reference on the same side of the cut as its data partner.
The input then has no entanglement across the cut.
Conditioned on syndrome $s$, the ideal measurement produces a pure state proportional to the vectorized projector $\Pi_s$.
Here, for an operator $M$, the unnormalized vectorization is $|M\rangle\!\rangle=(M\otimes I)\sum_z|z\rangle_Q|z\rangle_{Q'}$ in the computational basis.

A bipartite pure state has Schmidt rank $R$ if its Schmidt decomposition has $R$ nonzero terms, or equivalently if its reduced state has rank $R$.
When the Schmidt coefficients are equal, the state is a maximally entangled state on its two $R$-dimensional supports.
The syndrome projectors always give such a flat spectrum.

\begin{lemma}[Schmidt spectrum of a syndrome projector]
\label{lem:projector-choi-schmidt}
Let $\Pi_s$ be any syndrome projector and define $|\widehat\Pi_s\rangle=|\Pi_s\rangle\!\rangle/\sqrt{\operatorname{Tr}\Pi_s}$.
Across $(U,U')|(U^c,(U^c)')$, this vector has $2^{\chi_U(\cS)}$ nonzero Schmidt coefficients, all equal to $2^{-\chi_U(\cS)/2}$.
Its Schmidt rank is $2^{\chi_U(\cS)}$ and its entanglement entropy is $\chi_U(\cS)$ bits.
\end{lemma}

To obtain the decomposition, expand $\Pi_s$ as a stabilizer average and group its terms by cosets of $\cS_U+\cS_{U^c}$.
Each of the $2^{\chi_U}$ cosets supplies one product of local operators, and distinct cosets give orthogonal factors of equal Hilbert--Schmidt norm on both sides.
After normalization these factors therefore form the Schmidt decomposition \cite{nielsenQuantumDynamicsPhysical2003}.
Appendix~\ref{subsec:adaptive-branch-proof} verifies the orthogonality and constants.
Unlike the pure-code-state entropy identity, this projector statement requires no correctability assumption and holds at every cut.

\subsection{Crossing operations and circuit depth}
\label{sec:submodular-cd}

The reference input turns exactness into a constraint on each branch of the circuit.
A branch must produce the projector state from Lemma~\ref{lem:projector-choi-schmidt}.
Each crossing operation can increase the Schmidt rank of that state by only a bounded factor.
The same argument also handles adaptive circuits, because once all outcomes on a branch are fixed, each quantum operation is represented by a fixed Kraus operator.
Unobserved Kraus indices can be refined for this purpose within the proof, without giving the controller access to them.

For a hardware region $L$, write $g_L(\beta)$ for the number of crossing two-qubit operations on a complete Kraus trajectory $\beta$.
This count includes measurements and channels as well as gates.
To collect the resulting depth bounds, define
\begin{equation}
\kappa(\cS,G_{\hw})
=\max_{\substack{L\subseteq V_{\hw}\\|\partial_{G_{\hw}}L|>0}}
\frac{\chi_{Q\cap L}(\cS)}{|\partial_{G_{\hw}}L|},
\label{eq:main-kappa}
\end{equation}
with an empty maximum set to zero.
The numerator describes the data measurement, whereas the denominator describes the available hardware boundary.

\begin{theorem}[Adaptive extraction across a cut]
\label{thm:adaptive-instrument-cut-bound}
\label{thm:swap-only-submodular-cd}
\label{thm:stabilizer-cut-rank-functional}
Every exact circuit $C$ in Definition~\ref{def:adaptive-local-extraction} satisfies, for every hardware cut $L$ and every nonzero complete Kraus trajectory $\beta$,
\begin{equation}
\chi_{Q\cap L}(\cS)\le2g_L(\beta).
\label{eq:adaptive-trajectory-bound}
\end{equation}
Consequently,
\begin{equation}
\depth(C)\ge\frac12\kappa(\cS,G_{\hw}).
\label{eq:adaptive-cut-depth}
\end{equation}
If a cut has no boundary edge but positive $\chi_{Q\cap L}(\cS)$, exact extraction is impossible.
\end{theorem}

\begin{proof}
Prepare the local data--reference Bell pairs described above.
For syndrome $s$, exactness then requires the pure conditional output $|\widehat\Pi_s\rangle$, of Schmidt rank $2^{\chi_{Q\cap L}}$.
To compare individual branches, refine every instrument and final ancilla discard into individual Kraus outcomes.
The positive contributions of all trajectories reporting $s$ sum to a rank-one output.
Every nonzero contribution is therefore proportional to the same vector.

Along a fixed trajectory, operations within either side cannot increase Schmidt rank.
A crossing two-qubit Kraus operator is a sum of at most four products of local operators, since the one-qubit operator space has dimension four.
Such an operator therefore multiplies Schmidt rank by at most four.
Since the input has rank one across the cut, it follows that
\begin{equation}
2^{\chi_{Q\cap L}(\cS)}\le4^{g_L(\beta)}.
\end{equation}
At most $|\partial_{G_{\hw}}L|$ crossing operations occur in one layer, which gives the depth statement and the impossibility at an empty boundary.
Appendix~\ref{subsec:adaptive-branch-proof} includes the treatment of ancillary refinements and the projector decomposition.
\end{proof}

The cut bound does not assume code distance, rate, correctability, or Clifford gates.
Its constant, however, depends on which crossing primitive is counted as one operation.
The relevant algebraic quantity is the operator Schmidt rank, namely the least number of product operators needed to express that primitive's Kraus operator.
If those ranks are $r_1,\ldots,r_g$ on a trajectory, the same proof gives $\chi_U\le\sum_j\log_2 r_j$.

\begin{remark}[Primitive cost and sharp constant]
\label{rem:adaptive-sharp-constant}
A CNOT has operator Schmidt rank two, as shown by
\begin{equation}
\operatorname{CNOT}
=|0\rangle\langle0|\otimes I+|1\rangle\langle1|\otimes X.
\label{eq:cnot-schmidt-decomposition}
\end{equation}
If all crossing operations are CNOTs, then on every nonzero branch
\begin{equation}
g_L^{\rm CNOT}(\beta)\ge\chi_{Q\cap L}(\cS),
\qquad
T_{\rm CNOT}\ge\frac{\chi_{Q\cap L}(\cS)}{b_L},
\label{eq:cnot-cut-bound}
\end{equation}
where $b_L=|\partial_{G_{\hw}}L|>0$ and $T_{\rm CNOT}$ is the worst-case number of layers containing CNOTs.
Crossing CZ gates and individual projective two-qubit Pauli measurements also have Kraus operators of operator Schmidt rank at most two and obey the same crossing-count bound.
For arbitrary crossing unitaries mixed with individual Pauli measurements, one obtains
\begin{equation}
\chi_{Q\cap L}(\cS)
\le2g_L^{\rm unitary}(\beta)+g_L^{\rm Pauli\ measurement}(\beta).
\label{eq:weighted-crossing-primitives}
\end{equation}
A joint Bell-basis L\"uders measurement is a different primitive, since it measures $XX$ and $ZZ$ together and has Kraus operators of operator Schmidt rank four across the cut.
For two data qubits on opposite sides of one edge, this depth-one primitive attains $\chi_U/2=1$, so the general factor $1/2$ cannot be improved in this model.
\end{remark}

The boundary edge count may overestimate the operations available per layer when edges share a site.
If $\nu_L$ is the largest matching of boundary edges, then $g_L(\beta)\le\depth(C)\nu_L$.
For $\nu_L>0$, the depth lower bound then strengthens to $\chi_{Q\cap L}/(2\nu_L)$.
The simpler edge count is sufficient for the geometric results below.

\subsection{Protocols that attain the cut cost}
\label{subsec:single-cut-exact}

To determine whether the crossing bound is attainable, we now allow arbitrary operations within each side and charge only the interaction across the cut.
Free local processing separates the intrinsic resource needed by the measurement from the time needed to route and process qubits inside a device.

\begin{definition}[Bipartite extraction with free local processing]
\label{def:bipartite-extraction}
Alice holds $U$ and Bob holds $U^c$.
Each may apply arbitrary finite local instruments using unlimited local ancillas and memory, and classical communication, randomness, and processing are free.
The resource state is independent of the input data.
For each reported syndrome $s$, the protocol must implement $\rho\mapsto\Pi_s\rho\Pi_s$ and return every data qubit to its original owner's original register.
The crossing resources are either Bell pairs shared at the start, with no later crossing quantum operation, or counted two-qubit operations across the cut, with product ancillas initially.
When crossing links are specified, simultaneous crossing operations must use disjoint link endpoints.
\end{definition}

\begin{theorem}[Exact interaction cost of a cut]
\label{thm:single-cut-resource}
Let $\chi=\chi_U(\cS)$ in Definition~\ref{def:bipartite-extraction}.
\begin{enumerate}
\item\label{it:single-cut-bell}
With $m$ shared Bell pairs and no crossing operation, exact extraction is possible if and only if $m\ge\chi$.
\item\label{it:single-cut-cnot}
If the only crossing operations are CNOTs, every exact protocol uses at least $\chi$ on each nonzero branch, and some protocol uses exactly $\chi$.
\item\label{it:single-cut-general}
With arbitrary crossing two-qubit instruments, every exact protocol uses at least $\lceil\chi/2\rceil$ on each nonzero branch, and some protocol attains this count using unitary gates.
\item\label{it:single-cut-layers}
Suppose the available crossing links have maximum matching size $\nu\ge1$, and each party can move states freely between link qubits and sufficiently many local storage qubits.
The minimum number of layers containing crossing operations is $\lceil\chi/\nu\rceil$ for CNOTs and $\lceil\chi/(2\nu)\rceil$ for arbitrary two-qubit instruments.
\end{enumerate}
\end{theorem}

\begin{proof}
The reference-input proof requires final Schmidt rank $2^\chi$ on every nonzero branch.
Local processing cannot increase this rank, and $m$ shared Bell pairs supply rank $2^m$.
A crossing CNOT can at most double it, and a general crossing primitive can at most quadruple it.
These facts give all lower bounds, including the layer counts.

For the upper bounds, first take a basis with exactly $\chi$ crossing generators, as in Lemma~\ref{lem:minimum-crossing-generators}.
Measure its local generators locally.
For each crossing generator $P=P_A\otimes P_B$, share a Bell pair.
Apply controlled-$P_A$ and controlled-$P_B$ from its respective halves.
Then measure both halves in the $X$ basis.
The raw outcome pair $(u,v)$ applies
\begin{equation}
K_{u,v}=\frac{I+(-1)^{u+v}P}{2\sqrt2}
\label{eq:main-remote-parity-kraus}
\end{equation}
to the data.
The parity $u\oplus v$ is the check outcome, and summing its two raw branches gives the ideal parity instrument.
Since the generators commute, composing these gadgets yields $2^{-\chi/2}\Pi_s$ on each raw branch and the full syndrome instrument after grouping records.
Generator signs and redundant output labels are determined by the original signed group.

To realize the same protocol using crossing gates, note that one crossing CNOT prepares one shared Bell pair from $|+\rangle|0\rangle$.
A SWAP between halves of two locally prepared Bell pairs creates two shared Bell pairs.
Thus $\chi$ CNOTs suffice, or $\lfloor\chi/2\rfloor$ SWAPs and one CNOT when $\chi$ is odd.
Preparing these resources in batches along a maximum matching and storing them locally attains the crossing-layer counts.
Appendix~\ref{app:single-cut-protocols} gives the full calculation and storage argument.
\end{proof}

The parity gadget is the familiar entanglement-assisted nondestructive measurement \cite{beckmanCausalLocalizableQuantum2001,limTradeoffInformationGain2025}, and the CNOT and SWAP resource conversions are established \cite{eisertOptimalLocalImplementation2000}.
Here, the cut-adapted basis specifies exactly how many copies are needed.
For example, a single $ZZ$ measurement has $\chi_U=1$, although its classical outcome could be obtained destructively without shared entanglement.
A complete Bell-basis instrument has $\chi_U=2$, whereas a measurement that need only report its Bell label can consume one Bell pair and destroy the input \cite{bandyopadhyayEntanglementCostNonlocal2009}.
Discarding the syndrome changes the task again, because shared randomness implements the resulting stabilizer dephasing channel with no crossing operation \cite{beckmanCausalLocalizableQuantum2001}.

The Bell-pair lower bound can also be recovered from prior results, and it does not require a new general entanglement-cost principle.
The bipartite stabilizer normal form, for example, produces $2^{\chi_U}$ orthogonal test states, maximally entangled on their active registers, with distinct syndromes.
Theorem~2 of Lim, Hhan, and Kwon bounds their perfect nondestructive discrimination cost by $\chi_U$ ebits \cite{limTradeoffInformationGain2025}.
The instrument theorem of Akibue, Miyazaki, and Osaka gives the same Schmidt-rank obstruction from a branch Choi operator \cite[Theorem~3(i)]{akibueOptimizingEntanglementManipulation2026}.
Appendix~\ref{subsec:single-cut-prior-work} gives both reductions.

Free local processing is essential to interpreting the upper bounds.
Moving stored states to a port and carrying out controlled Pauli operations take quantum time on a real device.
Theorem~\ref{thm:single-cut-resource} optimizes the crossing resource and its crossing layers, not the total elapsed depth.
A second limitation is that different cuts can require different adapted bases.
Section~\ref{sec:worked-example-tree} gives a sparse attaining example and explains the additional structure that permits simultaneous edge optimality on a tree.

\subsection{Approximate instruments}

The flat projector spectrum also gives a quantitative bound when the instrument is only approximate.
The bound uses the normalized reference input
\begin{equation}
|\Phi\rangle=2^{-n/2}\sum_z|z\rangle_Q|z\rangle_{Q'},
\qquad \Phi=|\Phi\rangle\langle\Phi|,
\label{eq:reference-bell-input}
\end{equation}
and includes the syndrome register in the output.
Let $\mathcal E$ be the implemented instrument channel and $\mathcal M_{\cS}(\rho)=\sum_s|s\rangle\langle s|\otimes\Pi_s\rho\Pi_s$ the ideal one.
Write $\rho=(\mathcal E\otimes\operatorname{id})(\Phi)$ and $\sigma=(\mathcal M_{\cS}\otimes\operatorname{id})(\Phi)$, and define their squared fidelity by
\begin{equation}
F=F(\rho,\sigma)
=\left(\operatorname{Tr}\sqrt{\sqrt\sigma\rho\sqrt\sigma}\right)^2.
\label{eq:approximate-test-fidelity}
\end{equation}

\Needspace{14\baselineskip}
\begin{corollary}[Approximate extraction]
\label{cor:approximate-instrument-cut-bound}
Let $C$ satisfy the physical resource assumptions of Definition~\ref{def:adaptive-local-extraction}, without requiring exactness, and implement a trace-preserving instrument on the same data and syndrome output spaces.
For $U=Q\cap L$, $b_L=|\partial_{G_{\hw}}L|>0$, and $F>0$,
\begin{equation}
\depth(C)\ge
\frac{\chi_U(\cS)-\log_2(1/F)}{2b_L}.
\label{eq:approximate-cut-depth}
\end{equation}
On the reference input, the crossing count also satisfies
\begin{equation}
\Pr_{\Phi}\!\left[g_L(\beta)\ge a_F\right]\ge F/2,
\qquad
 a_F=\frac{\chi_U(\cS)-\log_2(2/F)}2.
\label{eq:approximate-cut-tail}
\end{equation}
If $b_L=0$, necessarily $F\le2^{-\chi_U(\cS)}$.
\end{corollary}

A branch with $g$ crossing operations has Schmidt rank at most $4^g$.
Its squared overlap with the flat rank-$2^{\chi_U}$ target is therefore at most $\min\{1,2^{2g-\chi_U}\}$ \cite{terhalSchmidtNumberDensity2000}.
Averaging this overlap bound over branches gives the corollary.
Appendix~\ref{subsec:approximate-branch-proof} proves the fidelity and probability steps.
At fixed positive fidelity, approximation thus subtracts only a constant from the cut numerator.
The test is the same for every cut because only the grouping of reference qubits changes.
If instead $\tfrac12\|\mathcal E-\mathcal M_{\cS}\|_\diamond\le\varepsilon<1$, the bounds hold with $1-\varepsilon$ replacing $F$, because the ideal-output projector then has expectation at least $1-\varepsilon$, which the proof uses directly.
In both forms, this is a bound on the full quantum instrument, not on syndrome probabilities alone.

\section{Worked examples and tree networks}
\label{sec:worked-example-tree}
\label{subsec:finite-code-examples}

The single-cut theorem allows arbitrary local processing, so its attaining basis need not retain the sparse checks used in an experimental circuit.
In the first example below, however, the optimum can be reached by deleting dependent checks from the original weight-six list.
In the second example, every listed check crosses the cut, the listed count is twice the intrinsic cut cost, and the optimum requires products of listed checks.
We then compare three circuits for four codes.
Finally, we extend the resource question to several processors connected by a tree, where the individual edge optima can all be attained by one protocol.

\subsection{Sixty Bell pairs for a gross-code half-torus}

The gross code is the $[[144,12,12]]$ bivariate bicycle code of Bravyi et al.~\cite{bravyiHighthresholdLowoverheadFaulttolerant2024}.
Its 144 data qubits come in two species on a $12\times6$ torus.
Its original presentation has 72 $X$ checks and 72 $Z$ checks, each of weight six.
We divide the code into two processors.
The set $U$ contains both species in the six columns $0\le i<6$, and $U^c$ contains both species in the remaining six columns.
Appendix~\ref{subsec:finite-code-calculations} fixes the matrix convention behind these coordinates and gives the exact row certificate.

The stabilizer space has dimension 132, and restriction to either half has rank 96.
There are therefore 36 independent stabilizers supported entirely in each half, and
\begin{equation}
\chi_U=132-36-36=96+96-132=60.
\label{eq:gross-half-torus-cut}
\end{equation}
With this value, the cut bound requires at least 60 crossing CNOTs if every crossing operation is a CNOT.
It requires at least 60 shared pairs if the only nonlocal resource is supplied Bell pairs and all subsequent operations are local to the processors.
Theorem~\ref{thm:single-cut-resource} also gives an optimum of 30 arbitrary crossing two-qubit gates under free local processing.
These crossing-resource counts are for this specified half-torus, and they do not minimize over balanced data cuts.

\begin{proposition}[A sparse optimum for the gross-code half-torus]
\label{prop:gross-sparse-single-cut-optimum}
For the gross code and cut just specified, the 36 original checks supported in each half form a basis of its full supported stabilizer space.
A basis of the full stabilizer space consists of 72 local and 60 crossing original checks, all of weight six and with data-qubit incidence at most six.
With arbitrary local processing, sufficient local storage, and free classical communication, 60 shared Bell pairs implement the exact syndrome instrument, including all 144 original outcome bits and the data output at its original sites.
No protocol supplied only Bell pairs and product local work registers uses fewer pairs across this cut.
\end{proposition}

\begin{proof}
In each half, 18 original checks of each Pauli type are independent and exhaust the supported space.
Retain these 72 local checks.
Delete six crossing $X$ checks and six crossing $Z$ checks, as specified in Appendix~\ref{subsec:finite-code-calculations}.
After this deletion, the remaining 132 checks are independent.
To attain the resource bound, measure each local check locally.
Measure each crossing check with the one-pair remote-parity gadget.
The omitted bits are parities of retained bits of the same Pauli type, with no sign offsets.
The circuit therefore recovers the full original syndrome.
Commutativity of the checks ensures that their joint measurement preserves every state within each syndrome sector, including its correlations with a reference.
Optimality follows from Equation~\eqref{eq:gross-half-torus-cut} and Theorem~\ref{thm:single-cut-resource}.
\end{proof}

\begin{table}[t]
\centering
\caption{Bell pairs used by three circuits for the same ideal gross-code syndrome instrument at the named half-torus cut.
The first circuit assigns each original check ancilla to the side of its row coordinate $i$ and replaces every crossing check--data CNOT by a remote CNOT.
The second measures the original crossing checks as whole parities.
The third measures an independent subset and reconstructs the omitted bits.
All three allow arbitrary local processing and sufficient local storage.
The counts exclude local depth and do not compare noisy performance.}
\label{tab:gross-resource-comparison}
\begin{tabular}{@{}p{0.76\textwidth}r@{}}
\toprule
Circuit & Bell pairs \\
\midrule
Remote CNOT for each crossing check--data incidence & 144 \\
One remote parity gadget for each original crossing check & 72 \\
Sparse independent basis of Proposition~\ref{prop:gross-sparse-single-cut-optimum} & 60 \\
\bottomrule
\end{tabular}
\end{table}

Table~\ref{tab:gross-resource-comparison} shows two ways to remove avoidable communication.
The first reduction measures a crossing check as one parity, avoiding transmission of its separate check--data interactions.
The second removes dependent crossing checks and thereby reaches the intrinsic lower bound while retaining the original sparse checks.
Both changes produce different circuits with the same ideal measurement instrument.
Under noise, however, directly measuring redundant checks can provide additional consistency information, whose usefulness depends on the noise model and decoder.
Elapsed depth also depends on the time required for local gates, pair preparation, and storage.
The optimal crossing-resource count therefore does not by itself establish optimal elapsed depth.

\subsection{Mitten codes and the thumb cut}
\label{subsec:mitten-thumb}

The mitten codes of Bhardwaj et al.\ are CSS qLDPC codes designed for modular processors \cite{bhardwajHighrateQLDPCProcessors2026}.
A mitten code on $n=5g$ data qubits encodes $k=g$ logical qubits.
Its published check matrices arrange the data in five consecutive blocks $D_1,\ldots,D_5$ of $g$ qubits each.
Every listed check has weight nine, with exactly three qubits in the thumb block $D_5$ and six in the four finger blocks.
We take $U=D_5$, so that the cut separates the thumb from the fingers.
Every listed check then crosses the cut, and the listed crossing count equals the total number $4g$ of listed checks.
Table~\ref{tab:mitten-thumb} gives the intrinsic cut cost of this cut for the eight instances published with the code.
We computed it from the stored matrices with exact arithmetic over $\mathbb F_2$, as Appendix~\ref{subsec:further-code-details} records.
In every instance the $4g$ listed checks are independent, no stabilizer is supported on the thumb, and the fingers support exactly $g$ stabilizers of each Pauli type.
Equation~\eqref{eq:intro-chi-definition} then gives
\begin{equation}
\chi_{D_5}=4g-0-2g=2g .
\label{eq:mitten-thumb-cut}
\end{equation}
A cut-adapted generating set therefore has $2g$ crossing generators, half of the listed count.
By Theorem~\ref{thm:single-cut-resource}, $2g$ shared Bell pairs are the exact cost of this cut under free local processing, and $2g$ crossing CNOTs or $g$ arbitrary crossing gates are the exact counts from product ancillas.

\begin{table}[t]
\centering
\caption{The thumb cut $U=D_5$ of the eight published mitten codes.
Every listed check crosses this cut, so the listed crossing count is the number $4g$ of listed checks.
The third column places each check ancilla on the finger side and counts one remote CNOT for each crossing check--data incidence.
The last column is the intrinsic cut cost, which equals the exact Bell-pair cost of the cut under free local processing.
The parameters are those reported by the source, and the distance of the largest instance is reported there as an upper bound.}
\label{tab:mitten-thumb}
\begin{tabular}{@{}lrrrr@{}}
\toprule
Code & $g$ & Listed checks crossing & Crossing CNOTs & $\chi_{D_5}$ \\
\midrule
$[[150,30,10]]$ & 30 & 120 & 360 & 60 \\
$[[200,40,12]]$ & 40 & 160 & 480 & 80 \\
$[[300,60,14]]$ & 60 & 240 & 720 & 120 \\
$[[500,100,16]]$ & 100 & 400 & 1200 & 200 \\
$[[540,108,18]]$ & 108 & 432 & 1296 & 216 \\
$[[630,126,20]]$ & 126 & 504 & 1512 & 252 \\
$[[780,156,22]]$ & 156 & 624 & 1872 & 312 \\
$[[975,195,{\le}24]]$ & 195 & 780 & 2340 & 390 \\
\bottomrule
\end{tabular}
\end{table}

The $2g$ finger stabilizers that complete a cut-adapted generating set are not listed checks, because every listed check meets the thumb.
They are products of listed checks, and such products can be heavy.
For the $[[150,30,10]]$ instance we enumerated all $2^{30}-1$ nonzero stabilizers supported on the fingers in each Pauli sector.
The lightest has weight $28$, whereas the listed checks have weight nine.
A generating set with only $60$ crossing generators must contain at least $60$ generators supported on the fingers, so the attaining protocol of Theorem~\ref{thm:single-cut-resource} measures finger generators of weight at least $28$.
This separates the exact crossing cost of a cut from the sparsity of the circuit that attains it.
The gross code of the previous subsection attains its optimum with original checks, and the mitten code does not.

\subsection{Three circuits for four codes}
\label{subsec:three-circuits-four-codes}

The comparison in Table~\ref{tab:gross-resource-comparison} extends to any code and cut.
The first circuit places each check ancilla on the side that holds the larger part of its support and uses one remote CNOT for each crossing check--data incidence.
The second measures each listed crossing check as one parity with the remote parity gadget.
The third uses a cut-adapted generating set with exactly $\chi_U$ crossing generators, which Theorem~\ref{thm:single-cut-resource} shows to be optimal.
Table~\ref{tab:three-circuits-four-codes} gives the three counts for four codes at one specified cut each.
The rotated surface code of distance $12$ is cut along a straight line between its sixth and seventh rows.
The gross code is cut as in the first subsection, and the two-gross $[[288,12,18]]$ bivariate bicycle code of Bravyi et al.\ \cite{bravyiHighthresholdLowoverheadFaulttolerant2024} is cut by the same rule on its $12\times12$ torus.
The mitten code is cut at the thumb.
Appendix~\ref{subsec:further-code-details} records the matrix conventions, the restricted ranks, and the provenance of every matrix.

\begin{table}[t]
\centering
\caption{Bell pairs used by three circuits for the same ideal syndrome instrument at one specified cut per code.
The first two circuits measure the listed checks, and the third measures a cut-adapted generating set.
All three allow arbitrary local processing and sufficient local storage, report every listed outcome, and preserve the state within each syndrome sector.
The surface distance was certified for the constructed matrix, and the other distances are the values reported by their sources.
The counts are for the specified cuts and are not minimized over cuts.}
\label{tab:three-circuits-four-codes}
\small
\begin{tabular}{@{}llrrrr@{}}
\toprule
Code & Cut & $|U|,|U^c|$ & Remote CNOTs & Listed parities & $\chi_U$ \\
\midrule
Surface $[[144,1,12]]$ & straight line & 72, 72 & 22 & 11 & 11 \\
Gross $[[144,12,12]]$ & half-torus & 72, 72 & 144 & 72 & 60 \\
Two-gross $[[288,12,18]]$ & half-torus & 144, 144 & 288 & 144 & 132 \\
Mitten $[[150,30,10]]$ & thumb & 30, 120 & 360 & 120 & 60 \\
\bottomrule
\end{tabular}
\end{table}

For the surface code, the eleven listed checks crossing the line are independent and already form the crossing part of a cut-adapted generating set, so the second circuit is optimal.
For the three qLDPC codes, the listed count exceeds $\chi_U$ by twelve, twelve, and sixty.
The first circuit is the one that a compiler for the listed checks would route, and its count exceeds the exact cost by factors of $2.4$, $2.2$, and $6$ for these three codes.
The excess has two sources, which the two reductions separate.
The remote CNOTs transmit the separate incidences of a crossing check, whereas its parity is one bit.
The listed checks also include crossings that an invertible recombination removes.
Only the second source depends on the stabilizer group rather than on the circuit, and $\chi_U$ measures exactly what remains after both reductions.

\subsection{Attaining all edge costs on a tree}

Suppose the data are distributed among several processors, and Bell pairs can be supplied only along the edges of a tree.
Deleting an edge $e$ divides the processors into two groups, and we let $U_e$ be the data held by either group.
The single-cut theorem requires at least $\chi_{U_e}$ pairs on that edge.
For a general network, these separate cut bounds need not describe an attainable resource allocation.
For a tree, however, each edge separates one subtree from the rest, and this structure permits an exact answer.

\begin{proposition}[Exact tree resources]
\label{prop:network-tree-resource}
Partition the data of an $[[n,k]]$ stabilizer code among the vertices of a tree.
Allow arbitrary quantum processing at each vertex, sufficient local memory, and free classical communication.
The initial nonlocal resources consist only of $b_e$ Bell pairs on each tree edge $e$, and all other work registers start in product states and independently of the input.
Exact deterministic syndrome extraction, with the quantum output returned to the original data sites, is possible if and only if
\begin{equation}
b_e\ge\chi_{U_e}(\cS)\qquad\text{for every edge }e.
\label{eq:tree-edge-optimum}
\end{equation}
In particular, all edge minima are simultaneously attained, and the total Bell-pair optimum is $\sum_e\chi_{U_e}(\cS)$.
\end{proposition}

The attaining resource is the normalized vectorized zero-syndrome projector,
\begin{equation}
|C\rangle=\frac{|\Pi_0\rangle\!\rangle}{\sqrt{2^k}},
\label{eq:tree-projector-resource}
\end{equation}
with both copies of a data qubit held by its processor.
Once this resource is available, local Bell measurements inject the actual input.
The measurement record then specifies local Pauli corrections and the syndrome.
Conversely, running exact extraction on local data--reference Bell pairs prepares a vectorized syndrome projector, which becomes $|C\rangle$ after a known Pauli correction on both copies.
Appendix~\ref{app:network-tree-resource} verifies the normalization and every outcome branch.

Across edge $e$, Lemma~\ref{lem:projector-choi-schmidt} gives $|C\rangle$ Schmidt rank $2^{\chi_{U_e}}$.
The pure-state tree preparation theorem of Yamasaki, Soeda, and Murao then prepares this state with exactly $\chi_{U_e}$ Bell pairs on every edge simultaneously \cite[Theorem~5, arXiv:1705.00006v2]{yamasakiGraphAssociatedEntanglement2017}.
Proposition~\ref{prop:network-tree-resource} applies that established preparation theorem to syndrome extraction.
It permits data-free helper vertices with unrestricted local quantum processing, and it counts pairs supplied on the actual tree edges.
The proposition makes no optimality claim, however, for networks with cycles or for elapsed circuit depth within the processors.

\section{Why presentation expansion is insufficient}
\label{sec:presentation-expansion}

The intrinsic cut cost is unchanged when we replace generators by invertible products of generators.
The graph of a chosen check list, however, can change substantially under the same operation.
The relevant graph here is the \emph{contracted Tanner graph}.
Its vertices are data qubits, and two vertices are joined whenever a listed check acts on both.
Expansion of this graph means that every sufficiently small set of qubits has many graph edges leaving it.
Such edges suggest communication only if they represent constraints that cannot be eliminated by changing generators.

To exhibit this distinction, we construct a counterexample from an annular surface code.
Its usual checks admit a constant-depth nearest-neighbor measurement circuit.
We recombine these checks into a bounded-weight, bounded-incidence list whose contracted Tanner graph expands.
The same circuit still measures the new list after classical relabeling.
The construction explains why presentation expansion cannot replace an intrinsic lower bound on $\chi_U$, and it contradicts the local-expander conclusion in the precise version of DBT identified below.

\subsection{An expanding annular presentation}

The construction starts from the surface code on a square annulus.
Take the unit faces of the $3L\times3L$ square grid and remove the central $L\times L$ block of faces.
Place data qubits on the edges of the remaining faces.
Each vertex check applies $X$ to its incident data qubits, while each face check applies $Z$ around its boundary.
Color the vertices like a checkerboard.
Remove one white vertex check.
The remaining checks form an independent generating list of a code with parameters $[[16L^2+8L,1,L+1]]$.
In this coloring every edge has exactly one black endpoint.
The black vertex checks therefore have pairwise disjoint supports that together cover every data qubit.
Order the black vertices, and choose a bounded-degree expander on them (Lemma~\ref{lem:expander-supply}).
Replace each black vertex check by its product with the black checks that precede it and are adjacent to it in the expander.
Because the black checks have disjoint supports, no Pauli factor cancels in these products.
The change of basis is unit triangular and hence invertible.

\begin{theorem}[Expanding presentations with constant-depth extraction]
\label{thm:annular-counterexample}
For every integer $L\ge1$, the construction above gives a CSS stabilizer code with parameters $[[n_L,1,L+1]]$, where $n_L=16L^2+8L$, together with an independent generating list $\mathcal G_L$ of its stabilizer group, with the following properties.
\begin{enumerate}
\item\label{it:annulus-ldpc}
Every generator in $\mathcal G_L$ has weight at most $132$, and every qubit lies in the support of at most $36$ generators.
\item\label{it:annulus-expansion}
The contracted Tanner graph of $\mathcal G_L$, which joins two qubits whenever some generator in $\mathcal G_L$ acts on both, has at least $|A|/18$ edges leaving every nonempty set $A$ of at most $n_L/2$ qubits.
\item\label{it:annulus-circuit}
On the $(6L+1)\times(6L+1)$ square grid, with each data qubit at a fixed site and two-qubit gates only between nearest neighbors, a Clifford circuit of quantum depth at most $13$ implements the exact Lüders instrument of $\mathcal G_L$.
The circuit prepares ancillas in $|0\rangle$ or $|+\rangle$ and returns them to these states, never measures a data qubit individually, and reports each syndrome bit of $\mathcal G_L$ as a parity of measurement outcomes, as in the model of Sections~2.1--2.2 of Ref.~\cite{delfosseBoundsStabilizerMeasurement2021}.
\end{enumerate}
\end{theorem}

\begin{proof}[Proof sketch]
Appendix~\ref{app:annular-construction} verifies each property directly.
Each listed black generator is a product of at most $33$ disjoint vertex checks of weight at most four, which gives weight at most $132$.
A data qubit lies in at most $33$ such products, one white vertex check, and two face checks.
Every product contains the full support of each of its factors.
The contracted Tanner graph therefore contains a clique on the support of every black vertex check.
It also contains all edges between the supports of two black vertex checks that are adjacent in the expander.
Lemma~\ref{lem:clique-blowup} transfers the expansion of the expander to this graph for every set of data qubits, including sets that are not unions of supports.
The circuit measures the original vertex and face checks with one ancilla each, using four nearest-neighbor CNOT layers for each check type.
It reports the syndrome of $\mathcal G_L$ through the invertible parity map inherited from the change of basis.
The joint projectors of the two lists coincide under this relabeling.
The circuit therefore implements the exact instrument of $\mathcal G_L$ on every input, including inputs entangled with references.
\end{proof}

The same construction extends to several disjoint annuli by applying one expander recombination to the black vertex checks of all components.
The extension gives positive-rate examples at any fixed annulus size and growing-distance examples when the annulus size increases, with the distinction between these regimes explicit.

\begin{corollary}[Several annuli]
\label{prop:many-annuli}
For every pair of integers $m,L\ge1$, there is a CSS code with parameters
\begin{equation}
[[n,k,d]]=[[m(16L^2+8L),m,L+1]]
\label{eq:many-annuli-parameters}
\end{equation}
and an independent generating list with weights at most $132$ and qubit incidences at most $36$ whose contracted Tanner graph has at least $|A|/18$ edges leaving every set $A$ with $0<|A|\le n/2$.
Writing $q=\lceil\sqrt m\rceil$, its exact L\"uders instrument has a nearest-neighbor Clifford implementation on a square patch of
\begin{equation}
N=q^2(6L+1)^2<\frac92n
\label{eq:many-annuli-hardware}
\end{equation}
qubits with quantum depth at most $13$, under the ancilla and parity-output conventions of Theorem~\ref{thm:annular-counterexample}.
The parameters satisfy
\begin{equation}
\frac{k}{n}=\frac1{16L^2+8L},
\qquad
\frac1{16}<\frac{kd^2}{n}=\frac{(L+1)^2}{16L^2+8L}\le\frac16.
\label{eq:many-annuli-rate-distance}
\end{equation}
\end{corollary}

The proof packs the annuli into disjoint square tiles, runs their local circuits in parallel, and then applies one global invertible parity map to the outcomes.
The construction is given in Appendix~\ref{subsec:many-annuli-proof}.

The intrinsic cut cost explains why the depth remains constant.
The stabilizer space of $\mathcal G_L$ is that of the surface code.
The original checks of that code are local, so $\chi_U$ never exceeds the number of original checks that cross a cut.
For a vertical line through the grid, only checks centered within distance one of the line cross it.
Hence $\chi_U=O(L)$ while $\Theta(L)$ hardware edges cross the line.
The depth lower bound from this cut therefore remains bounded as $L$ grows.
The number of listed generators that cross a vertical line splitting the data in half, in contrast, grows linearly with $n_L$, because the contracted Tanner graph of $\mathcal G_L$ expands.

\subsection{Consequences for local-expander bounds}

In the notation of DBT, a family of local-expander codes is a family of quantum LDPC codes, each with a fixed generating list \cite{delfosseBoundsStabilizerMeasurement2021}.
The contracted Tanner graphs of these lists are required to have local Cheeger constant $h_\varepsilon\ge\alpha$ for fixed $\alpha,\varepsilon>0$.
Here $h_\varepsilon$ minimizes $|\partial A|/|A|$ over nonempty sets of at most $\varepsilon n/2$ qubits.
Part~\ref{it:annulus-expansion} of Theorem~\ref{thm:annular-counterexample} gives this property with $\alpha=1/18$ and $\varepsilon=1$.

\begin{corollary}[Counterexample to the local-expander corollary of DBT]
\label{cor:dbt-local-expander}
The codes and circuits of Theorem~\ref{thm:annular-counterexample} form a family of two-dimensional $1$-local Clifford syndrome-extraction circuits for a family of local-expander quantum LDPC codes, with $N_L=(6L+1)^2$ qubits in a square patch, as in Corollary~1 of DBT in arXiv:2109.14599v1 \cite{delfosseBoundsStabilizerMeasurement2021}.
Their depth is at most $13$, whereas
\begin{equation}
\frac{n_L}{\sqrt{N_L}}=\frac{16L^2+8L}{6L+1}\longrightarrow\infty.
\label{eq:annulus-ratio}
\end{equation}
The conclusion $\depth(C_L)\ge\Omega(n_L/\sqrt{N_L})$ of that corollary therefore does not hold for this family.
\end{corollary}

If only the qubits that the circuit acts on are counted, then $N_L=32L^2+16L-1$.
These qubits still fill a constant fraction of the patch, and $n_L/\sqrt{N_L}$ still diverges.

Theorem~1 of the same version counts independent listed generators crossing a hardware cut.
For the annular list, expansion makes that count proportional to $n_L$ at a near-balanced vertical cut.
The physical boundary there has only $O(L)$ edges, however, and the extraction depth stays bounded.
Appendix~\ref{subsec:basis-change-counterexample} gives simpler examples and identifies the corresponding failures in the proof.
These failures concern the stated lower bounds, whereas the extraction constructions of DBT are unaffected.

Corollary~\ref{prop:many-annuli} extends the counterexample to a two-parameter family satisfying $kd^2=\Theta(n)$, including positive-rate subfamilies of bounded distance.
Indeed, keeping $L$ fixed and taking $m\to\infty$ gives rate $1/(16L^2+8L)>0$ and distance $L+1$.
Taking $L\to\infty$ instead gives growing distance and vanishing rate, even if $m$ also grows.
For every sequence with $n\to\infty$, the bound $N<9n/2$ makes $n/\sqrt N$ diverge while extraction depth stays bounded.
The same contradiction to Corollary~1 of DBT v1 therefore persists \cite{delfosseBoundsStabilizerMeasurement2021}.
The family is consistent with the tradeoff of Bravyi, Poulin, and Terhal \cite{bravyiTradeoffsReliableQuantum2010}.
It does not address the simultaneous positive-rate and growing-distance regime, because any uniform bound $kd^2\le Cn$ together with $k/n\ge r>0$ implies $d^2\le C/r$.
For a union of whole annuli the intrinsic cut cost is zero, even though the chosen presentation has crossing generators.
This contrast makes the distinction between stabilizer-space and presentation expansion explicit.
The expander supply is specified by a terminating finite search rather than an efficient construction.
The uniform statements follow from the proof in Appendix~\ref{app:annular-construction}.

The distinction between stabilizer-space and presentation expansion leads us to impose a bound on the stabilizer space itself, namely $\chi_U\ge\gamma|U|$ throughout a stated range of data regions.
Section~\ref{sec:separator-barriers} proves such a condition for one fixed quantum Tanner family, and Section~\ref{sec:static-2d-barrier} derives its geometric consequences.

\section{Uniform costs for small regions}
\label{sec:separator-barriers}

The intrinsic cut cost can be evaluated once a hardware region and its data qubits have been chosen.
A guarantee for every placement, however, needs a code whose intrinsic cut cost is large on every data subset that a hardware bottleneck might select.
We now prove this property for one fixed quantum Tanner family.
Its shared constraints grow with the number of qubits in every sufficiently small region, not with the boundary of a particular drawing of the checks.

\subsection{Why rate and distance are not enough}

A good quantum LDPC code has both a constant rate and a distance proportional to its block length.
These global parameters, however, do not require every small region to share many constraints with its complement.
To see this, tensor a good $[[m,k,d]]$ code with $p=\Theta(m)$ qubits fixed in $|0\rangle$ by weight-one $Z$ checks.
The enlarged code still has constant rate, linear distance, and bounded check weight and incidence.
Every subset of the added qubits, however, has $\chi_U=0$.
A uniform lower bound on $\chi_U$ therefore needs information beyond rate and distance.

The same distinction separates a uniform statement from an average one.
Entropy sums over correctable partitions constrain average entanglement at positive code rate~\cite{kimConditionalIndependenceQuantum2013}, and selected small-subset averages have been studied for CSS codes~\cite{zhaoGraphbasedApproachEntanglement2026}.
Our theorem instead bounds each subset in its stated size range and every pure code state.
For a correctable region, Lemma~\ref{lem:correctable-cut-entropy} reduces this task to counting the stabilizers supported entirely inside the region.
The count must cover their full vector space, including products of checks whose factors outside the region cancel.

\subsection{The fixed-family volume law}

Quantum Tanner codes place qubits on the squares of a combinatorial square complex~\cite{leverrierQuantumTannerCodes2022} and define $X$ and $Z$ checks from small classical codes at the vertices of the complex.
A square complex specifies which vertices, edges, and squares meet, and it need not describe a physical grid.
A local check is a binary matrix of coefficients whose rows and columns obey the two component codes.
The construction combines distance properties of these small codes with expansion of the graphs formed by the square corners.

We use the full construction hypotheses of Leverrier and Z\'emor, with a sufficiently large degree fixed before the block length grows.
To state them precisely, fix $\rho\in(0,1/2)$, $\varepsilon\in(0,1/2)$, $\gamma\in(1/2+\varepsilon,1)$, and $\delta>0$.
The component codes $C_A,C_B$ have length $\Delta$, satisfy $0<\dim C_A\le\rho\Delta$ and $\dim C_B=\Delta-\dim C_A$, and both codes and both duals have distance at least $\delta\Delta$.
The two dual tensor codes $(C_A\otimes C_B)^\perp$ and $(C_A^\perp\otimes C_B^\perp)^\perp$ are $\Delta^{3/2-\varepsilon/2}$-robust with $\Delta^\gamma$-resistance to puncturing.
We take both properties as in Definitions~5 and~8 and Theorem~17 of \href{https://arxiv.org/abs/2202.13641v3}{arXiv:2202.13641v3}.
Robustness here bounds the number of rows and columns needed to cover a short dual tensor codeword.
Resistance to puncturing requires that property to survive the specified deletions of rows and columns.

The square-complex family obeys the Ramanujan and total no-conjugacy (TNC) assumptions of the same construction, which supply mixing estimates and nondegenerate square incidences.
For the existence construction, choose $\delta$ small enough that $-\delta\log_2\delta-(1-\delta)\log_2(1-\delta)<\rho$.
Theorem~18 of that version then supplies component codes with the required properties for sufficiently large admissible degree.
All numbered references to this construction below use that version, while the bibliography also links its FOCS publication.

\Needspace{16\baselineskip}
\begin{theorem}[Uniform small-region cost and entanglement]
\label{thm:quantum-tanner-target-window}
Fix a construction of Theorem~17 of Leverrier and Z\'emor (arXiv:2202.13641v3) with sufficiently large fixed local degree, component codes satisfying all three hypotheses of that theorem, and the TNC/Ramanujan square-complex family used in their Theorem~1.
For the resulting CSS stabilizer spaces $\cS_j$ on $n_j$ squares, there are constants $\eta,\tau>0$, depending only on the fixed local data, such that every sufficiently large $j$ and every data set $U$ satisfy
\begin{equation}
\begin{gathered}
|U|\le\frac{\eta n_j}{2}\quad\Longrightarrow\\
\dim(\cS_j)_U\le(1-\tau)|U|,
\qquad
\chi_U(\cS_j)\ge2\tau|U|.
\end{gathered}
\label{eq:quantum-tanner-target-window}
\end{equation}
Every such $U$ is correctable, and every pure code state $|\psi\rangle$ obeys
\begin{equation}
S(\psi_U)=|U|-\dim(\cS_j)_U\ge\tau|U|.
\label{eq:quantum-tanner-pure-entanglement}
\end{equation}
\end{theorem}

The constants are independent of $j$, of the chosen region, and of the logical state, although their values depend on the fixed construction.
The theorem does not claim the same conclusion for arbitrary good quantum LDPC families.
For comparison, local surface-code checks give $\chi_U=O(|\partial U|)$ on regular lattice regions, consistent with constant-depth local extraction.
The theorem instead forces a volume law, in which the entropy of every eligible region is proportional to its number of qubits.

\begin{proof}[Proof idea]
The obstacle is cancellation outside $U$, because counting the original checks inside $U$ would miss stabilizers formed by multiplying checks that extend beyond it.
To account for these products, we give every supported stabilizer a short description in terms of local check components.
We then bound the dimension of all such descriptions.

The reduction theorem of the construction writes a short stabilizer as a sum of few local components.
Expansion forces a nonzero relation among those components to use linearly many vertices.
When $U$ is small enough, two short descriptions cannot differ by such a relation, and three cannot form one either.
The short description is therefore unique and respects addition, so it converts the entire supported stabilizer space into a linear space of short coefficient tuples.

To bound the dimension of this linear space, consider every vertex used by any of its tuples.
At a vertex used at all, at least half the tuples have a nonzero component.
Averaging therefore bounds the number of such vertices by twice the maximum support size of one tuple.
A coefficient on a square with only one selected endpoint cannot cancel.
If that square lies outside $U$, the coefficient must therefore vanish.
The remaining allowed coordinates, together with local distance and the Singleton bound, limit the dimension at each selected vertex.

The separate $X$ and $Z$ counts, however, can both include the same boundary square.
To control this double counting, we associate each overlap with an edge of the square complex and apply the incidence and mixing estimates for that complex.
For sufficiently large fixed degree and sufficiently small fixed $\eta$, the resulting deficit is a positive fraction of $|U|$.
Correctability and Lemma~\ref{lem:correctable-cut-entropy} then give the cut and entropy statements.
Appendix~\ref{subsec:quantum-tanner-target-window} proves these steps, including the strict size conditions and uniform constants.
\end{proof}

\subsection{Individual hardware bottlenecks}

A hardware region $L$ may contain both data and ancillas.
The theorem applies to its data subset $U=Q\cap L$, whereas the available interactions are counted by the hardware edges leaving $L$.
The code estimate holds for every eligible $U$.
Applying the theorem to a hardware region therefore requires no assumption on the placement, the geometry inside the region, or the hardware degree.

\begin{corollary}[Quantum Tanner bottleneck bound]
\label{cor:quantum-tanner-bottleneck}
For every sufficiently large member of the fixed family in Theorem~\ref{thm:quantum-tanner-target-window}, place the data arbitrarily on a hardware graph.
Let $L$ be a hardware region, put $U=Q\cap L$ and $b_L=|\partial_{G_{\hw}}L|$, and suppose $0<|U|\le\eta n/2$.
For $b_L>0$, exact extraction in the adaptive model of Definition~\ref{def:adaptive-local-extraction} requires
\begin{equation}
\depth(C)\ge\frac{\tau|U|}{b_L}.
\label{eq:quantum-tanner-bottleneck}
\end{equation}
For $b_L=0$, exact extraction is infeasible.
\end{corollary}

\begin{proof}
The family theorem gives $\chi_U\ge2\tau|U|$.
Theorem~\ref{thm:adaptive-instrument-cut-bound} requires at least $\chi_U/2$ crossing operations, of which at most $b_L$ can occur in one layer.
If $b_L=0$, no operation can supply the positive crossing cost.
\end{proof}

A module holding $u$ data qubits in the theorem's range behind $b>0$ hardware edges thus requires depth at least $\tau u/b$.
Equivalently, depth at most $T>0$ requires at least $\lceil\tau u/T\rceil$ edges leaving the module.
Internal routing can increase the required depth, but cannot remove this boundary cost.
All small bottlenecks of a device are summarized by the data-weighted edge expansion
\begin{equation}
h_{Q,\eta}(G_{\hw})=
\min_{L:\,0<|Q\cap L|\le\eta n/2}
\frac{|\partial_{G_{\hw}}L|}{|Q\cap L|}.
\label{eq:small-region-hardware-expansion}
\end{equation}
For a positive minimum, the corollary gives depth at least $\tau/h_{Q,\eta}$, whereas a zero minimum rules out exact extraction.

\begin{remark}[Connection with presentation expansion]
\label{rem:intrinsic-family-gap}
For any generating presentation with qubit incidence at most $q$, the uniform bound also implies local expansion of the contracted Tanner graph of DBT.
Each crossing generator contains a pair of data qubits across the cut, and each pair belongs to at most $q$ generators.
Consequently its edge boundary has size at least $\chi_U/q\ge(2\tau/q)|U|$ throughout the theorem's range.
The converse fails by Section~\ref{sec:presentation-expansion}, because presentation expansion does not control cancellations among generators.
\end{remark}

\section{Space and depth on grids}
\label{sec:static-2d-barrier}

Additional ancillas can spread the work of syndrome extraction over a larger device.
On a square grid, the intrinsic cut bound quantifies the depth that remains necessary.
An adaptive construction achieves the same order for the fixed family.
The lower bound applies to every placement, whereas the construction chooses one placement in advance and restores every data label to its original site after each round.

\subsection{From a cut to a geometric lower bound}

A square grid with $N$ sites has width $\sqrt N$.
Sweep a vertical line across it until the data on one side approach the upper endpoint of the small-region theorem.
One column changes the count by at most $\sqrt N$, while only $O(\sqrt N)$ local hardware edges cross the line.
If $n/\sqrt N$ is large, the sweep therefore selects $\Theta(n)$ data qubits behind $O(\sqrt N)$ edges, and such a bottleneck requires depth $\Omega(n/\sqrt N)$.
When the ratio is bounded, the same order follows from the need for at least one quantum operation.
The following statement includes the constants and permits a fixed interaction range larger than nearest-neighbor range.

\begin{corollary}[Intrinsic bound on square patches]
\label{cor:intrinsic-patch}
Let $\cS$ be a stabilizer space on $n$ data qubits with $\chi_U(\cS)\ge\gamma|U|$ for every $U$ with $0<|U|\le\eta n/2$, where $\gamma>0$, $0<\eta\le1$, and $\eta n\ge2$.
Let $C$ be an exact extraction circuit in the model of Definition~\ref{def:adaptive-local-extraction} on an $N$-site square patch whose hardware edges join sites at Chebyshev distance at most a positive integer $b$.
Then
\begin{equation}
\depth(C)\ge
\min\left\{\frac{\gamma\eta}{8b^2(2b+1)},\frac\eta4\right\}\frac{n}{\sqrt N}.
\label{eq:intrinsic-patch-depth}
\end{equation}
\end{corollary}

The proof in Appendix~\ref{subsec:intrinsic-patch-proof} treats both regimes and permits missing hardware edges.
For the fixed family, Theorem~\ref{thm:quantum-tanner-target-window} supplies $\gamma=2\tau$, so the bound holds for every placement of that family's data on the patch.

The same reasoning extends to planar devices that are not square patches.
A weighted separator removes few hardware vertices while leaving outside sides with controlled data counts.
To use it here, assign weight one to data sites and zero to ancillas.
Then apply the separator recursively to select a region with between $\eta n/8$ and $\eta n/2$ data qubits.
Suppose the hardware has maximum degree $\Delta_{\hw}$ and a hereditary weighted-separator bound $s_{\hw}(N)$.
Theorem~\ref{thm:weighted-separator-depth} then gives depth $\Omega(n/[\Delta_{\hw}s_{\hw}(N)])$ whenever $s_{\hw}(N)$ is a sufficiently small fraction of $n$.
The weighted separator theorem of Alon, Seymour, and Thomas supplies $s_{\hw}(N)=O(\sqrt N)$ for a fixed excluded minor~\cite{alonSeparatorTheoremGraphs1990}.
Appendix~\ref{subsec:weighted-separators} gives the selection argument and its explicit constants.

\paragraph{Relation to existing geometric bounds.}
The rate and distance of a good code also constrain its aggregate extraction cost.
Theorem~24 of Baspin, Fawzi, and Shayeghi (BFS) gives $\Omega(k\sqrt d/N)$ in two dimensions~\cite{baspinLowerBoundOverhead2023}.
Their correctable-partition framework, combined with data-weighted separators, further yields $\Omega(n/\sqrt N)$ for good stabilizer families on bounded-degree hardware excluding a fixed minor.
Appendix~\ref{subsec:bfs-weighted-partitions} derives this consequence from their partition method, with a direct entanglement-growth argument for product inputs.
The aggregate planar scaling therefore does not require our uniform theorem.
What Section~\ref{sec:separator-barriers} adds is control of each individual small bottleneck, which rate and distance alone do not imply.
Both statements concern exact adaptive extraction with all quantum systems included in $N$ and free classical processing.

\subsection{An adaptive construction at every width}

The two endpoints of the width range suggest how additional space can reduce depth.
With linear space, routing a constant number of disjoint gate layers across a square mesh costs $O(\sqrt n)$ depth.
This ingredient is supplied by the mesh-permutation bound underlying Lemma~25 of Bacon, Flammia, Harrow, and Shi~\cite{baconSparseQuantumCodes2017b}.
With quadratic space, the construction of DBT uses intermediate measurements and classical feedforward to obtain constant depth for bounded-degree CSS presentations~\cite{delfosseBoundsStabilizerMeasurement2021}.
We interpolate between these regimes for general commuting Pauli presentations, including the case $N=n$ where no separate check ancillas fit.

\begin{theorem}[Adaptive extraction at every grid width]
\label{thm:adaptive-grid-width-upper}
Let a commuting signed Pauli presentation on $n\ge1$ data qubits have nonempty checks of weight at most $w$, with at most $\delta$ checks incident on each data qubit, where $w,\delta$ are fixed positive integers.
For every full nearest-neighbor $L\times L$ grid with $N=L^2\ge n$, there is a fixed placement of the data and an exact extraction circuit in the model of Definition~\ref{def:adaptive-local-extraction} with
\begin{equation}
\depth(C)=O_{w,\delta}\!\left(\max\left\{1,\frac{n}{\sqrt N}\right\}\right).
\label{eq:adaptive-grid-width-upper}
\end{equation}
The circuit uses product ancillas, nearest-neighbor Clifford gates, one-qubit measurements and resets, and free classical feedforward, and each data label has the same physical site at input and output.
The same conclusion holds on full rectangular grids of bounded aspect ratio, with constants also depending on that ratio.
The placement may be chosen before execution, and the statement does not assert this bound for every prescribed data placement.
\end{theorem}

The circuit has two parts, namely an unrestricted circuit that measures the checks and a grid implementation of its two-qubit gates.
Both parts preserve the full measurement instrument, including logical coherence and entanglement with a reference.

For the first part, color the checks so that checks of the same color have disjoint supports.
Bounded weight and incidence ensure that only a constant number of colors is needed.
Change each Pauli factor of one check to $Z$.
Then concentrate the parity of that check onto a data pivot with CNOTs.
Measure that pivot, and undo the concentration.
The resulting branch operator is precisely the check projector.
Measuring the pivot therefore does not measure an additional logical observable.
Commutativity of the check projectors then ensures that completing all colors gives their joint measurement.
The unrestricted circuit as a whole has only a constant number of layers of disjoint CNOTs.

On a narrow grid, ordinary mesh permutations bring each such matching to adjacent pairs.
They then return all states to their starting sites.
This route uses no extra check register, so it also works when every grid site initially holds data.
On a wider grid, divide the hardware into $4\times4$ tiles.
Place one data qubit in each occupied tile, balancing the number of occupied tiles across the rows.
Pairs of terminals can then be divided into $O(1+n/\sqrt N)$ batches whose endpoint rows are disjoint.
Within one batch, assign each pair in different rows a distinct routing column.
Connect the two rows of that pair through its column.
A same-row pair uses its row directly.

These routes are implemented by entanglement swapping, as shown in Fig.~\ref{fig:grid-compiler}.
Each tile prepares Bell pairs between its selected ports using only its local ancillas.
Bell measurements on facing ports then join all pieces of a route simultaneously.
The measurements leave an endpoint Bell pair after classical parity corrections.
With this pair, the two endpoint data qubits can perform a corrected remote CNOT.
Different routes can pass through one tile using disjoint ports.
Their Bell pairs are prepared serially within the tile, so no physical overpass is assumed.
The quantum depth of a batch is constant, including preparation, measurements, resets, and corrections.
Appendix~\ref{app:grid-width-compiler} specifies the ports, proves every branch identity, and joins the two width regimes.

\begin{figure}[tbp]
\centering
\begin{tikzpicture}[
  x=1cm,y=1cm,
  every node/.style={font=\small},
  panel/.style={font=\small\bfseries,anchor=west},
  note/.style={font=\scriptsize,align=left},
  qdot/.style={circle,draw=black,fill=white,inner sep=0pt,minimum size=3.3mm},
  data/.style={circle,draw=blue!65!black,fill=blue!65!black,inner sep=0pt,minimum size=3.3mm},
  bell/.style={blue!70!black,line width=1pt},
  routeA/.style={blue!70!black,line width=1.5pt},
  routeB/.style={orange!80!black,line width=1.5pt},
  arrow/.style={-{Latex[length=1.8mm]},line width=.7pt}
]
\begin{scope}
  \node[panel] at (0,4.2) {(a) A batch on the tiled grid};
  \foreach \r in {0,...,4}{
    \foreach \c in {0,...,4}{
      \draw[black!25,fill=black!2] (.68*\c,3.5-.68*\r)
        rectangle ++(.61,-.61);
    }
  }
  \draw[routeA] (.305,3.195)--(1.665,3.195)--(1.665,.475)--(3.025,.475);
  \draw[routeB] (.305,1.835)--(3.025,1.835);
  \draw[black,line width=.8pt] (1.36,2.14) rectangle (1.97,1.53);
  \foreach \xx/\yy/\lab in {.305/3.195/a,3.025/.475/b,.305/1.835/c,3.025/1.835/d}{
    \node[data] at (\xx,\yy) {};
    \node[font=\scriptsize,fill=white,inner sep=.7pt] at (\xx+.22,\yy+.2) {$\lab$};
  }
  \node[note,text width=3.5cm] at (5.35,2.9) {One data site per occupied tile.\par Different pairs use disjoint endpoint rows.};
  \draw[arrow] (3.45,1.55)--(2.01,1.81);
  \node[note,text width=3.5cm] at (5.35,1.4) {A shared tile uses distinct ports for the two routes.};
  \node[note,anchor=west] at (0,-.15) {Blue: row--column--row. Orange: same-row pair.};
\end{scope}
\begin{scope}[xshift=8.05cm]
  \node[panel] at (0,4.2) {(b) Physical tile and local preparation};
  \foreach \r in {0,...,3}{
    \draw[black!35] (.15,3.4-.8*\r)--(2.55,3.4-.8*\r);
  }
  \foreach \c in {0,...,3}{
    \draw[black!35] (.15+.8*\c,3.4)--(.15+.8*\c,1.0);
  }
  \foreach \r in {0,...,3}{
    \foreach \c in {0,...,3}{
      \node[qdot,minimum size=2.3mm,draw=black!50] at (.15+.8*\c,3.4-.8*\r) {};
    }
  }
  \node[data,label={[font=\scriptsize]above:$D$}] at (.15,3.4) {};
  \node[qdot,label={[font=\scriptsize]above:$T$}] at (.95,3.4) {};
  \node[qdot,label={[font=\scriptsize]above:$N$}] at (1.75,3.4) {};
  \node[qdot,label={[font=\scriptsize]right:$E$}] at (2.55,1.8) {};
  \node[qdot,label={[font=\scriptsize]below:$S$}] at (1.75,1.0) {};
  \node[qdot,label={[font=\scriptsize]left:$W$}] at (.15,1.8) {};
  \node[note,text width=3.45cm] at (5.05,2.95) {Filled: data.\par Open: ancillas.\par Lines: physical grid edges.};
  \node[note,text width=3.45cm] at (5.05,1.5) {Prepare up to two disjoint port pairs serially.\par Restore every other tile register.};
  \node[note,anchor=west] at (0,-.15) {Local coordinates: row downward, column rightward.};
\end{scope}
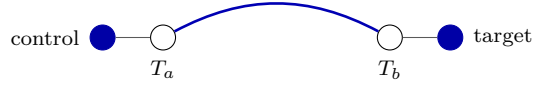
\begin{scope}[yshift=-5.0cm]
  \node[panel] at (0,4.2) {(c) Simultaneous fusion};
  \foreach \x in {0,2.35,4.7}{
    \draw[black!25,rounded corners=2pt] (\x,2.2) rectangle ++(1.9,1.25);
    \node[qdot] at (\x+.27,2.6) {};
    \node[qdot] at (\x+1.63,2.6) {};
    \draw[bell] (\x+.27,2.6) to[bend left=45] (\x+1.63,2.6);
  }
  \foreach \x in {1.63,3.98}{
    \draw[orange!80!black,line width=1pt,rounded corners=3pt]
      (\x-.19,2.35) rectangle (\x+1.18,2.86);
    \draw[black!60] (\x+.17,2.6)--(\x+.82,2.6);
  }
  \node[note,anchor=west] at (0,1.83) {Orange boxes: Bell measurements on facing ports.};
  \draw[arrow] (3.3,1.5)--(3.3,1.03);
  \node[note,anchor=west] at (3.55,1.28) {classical parities};
  \node[qdot,label={[font=\scriptsize]below:$T_a$}] at (.27,.5) {};
  \node[qdot,label={[font=\scriptsize]below:$T_b$}] at (6.33,.5) {};
  \draw[bell] (.27,.5) to[bend left=13] (6.33,.5);
\end{scope}
\begin{scope}[xshift=8.05cm,yshift=-5.0cm]
  \node[panel] at (0,4.2) {(d) Remote CNOT and reuse};
  \node[data,label={[font=\scriptsize]left:control}] (d1) at (1.0,2.65) {};
  \node[qdot,label={[font=\scriptsize]below:$T_a$}] (ta) at (1.8,2.65) {};
  \node[qdot,label={[font=\scriptsize]below:$T_b$}] (tb) at (4.8,2.65) {};
  \node[data,label={[font=\scriptsize]right:target}] (d2) at (5.6,2.65) {};
  \draw[black!65] (d1)--(ta);
  \draw[black!65] (tb)--(d2);
  \draw[bell] (ta) to[bend left=28] (tb);
  \node[note,text width=6.7cm,anchor=north west] at (0,1.85) {Local CNOTs, port measurements, and conditional Pauli corrections implement the remote CNOT.};
  \node[note,text width=6.7cm,anchor=north west] at (0,.72) {Data remain at their fixed sites.\par Reset and reuse the ancillas for the next batch.};
\end{scope}
\end{tikzpicture}
\caption{The adaptive grid compiler.
(a) Data sites are chosen in $4\times4$ tiles and balanced over tile rows.
A batch of terminal pairs uses disjoint endpoint rows, and pairs in different rows use distinct routing columns.
(b) The tile has data at $(0,0)$, terminal port $T=(0,1)$, and directional ports $N=(0,2)$, $E=(2,3)$, $S=(3,2)$, $W=(2,0)$.
Coordinates are ordered as row and column.
The tile's ancilla grid prepares each requested port pair locally.
In a shared tile, it serializes up to two pair preparations while restoring all spectator states.
(c) Bell measurements on disjoint facing ports fuse the locally prepared pairs into endpoint pairs.
Classical parities determine their Pauli corrections.
(d) Each corrected endpoint pair implements a remote CNOT, preserving the data sites.
Bell-pair curves depict entanglement, not extra hardware edges.
All quantum steps are counted, and classical feedforward is free.
The mesh-permutation fallback covers grids too narrow for one tile per data qubit.}
\label{fig:grid-compiler}
\end{figure}

\Needspace{15\baselineskip}
\subsection{The matching fixed-family tradeoff}

We can now combine the construction with the placement-independent lower bound to determine the optimum asymptotic tradeoff for the fixed family.
In this tradeoff, the minimization includes the choice of data placement, made once before the circuit runs.

\begin{corollary}[Optimal adaptive tradeoff for the fixed family]
\label{cor:fixed-family-grid-width-tight}
Retain the fixed quantum Tanner family and all hypotheses of Theorem~\ref{thm:quantum-tanner-target-window}.
On full square grids with $N$ sites, let $D_{\min}(n,N)$ minimize exact adaptive extraction depth over circuits and fixed data placements.
For the sufficiently large block lengths of that family and any fixed $c>0$, throughout $n\le N\le c n^2$,
\begin{equation}
\begin{gathered}
D_{\min}(n,N)=\Theta\!\left(\max\left\{1,\frac{n}{\sqrt N}\right\}\right),\\
N D_{\min}(n,N)^2=\Theta(n^2).
\end{gathered}
\label{eq:fixed-family-grid-width-tight}
\end{equation}
The constants may depend on the fixed family and on $c$.
\end{corollary}

\begin{proof}
The bounded-weight, bounded-incidence presentation gives the upper bound by Theorem~\ref{thm:adaptive-grid-width-upper}.
Corollary~\ref{cor:intrinsic-patch} gives the matching lower bound uniformly over placements under the same adaptive circuit assumptions.
The two bounds match because, for fixed $c$ and $N\le c n^2$, the quantities $n/\sqrt N$ and $\max\{1,n/\sqrt N\}$ differ by at most a constant factor.
\end{proof}

At $N=\Theta(n)$, the optimum depth is $\Theta(\sqrt n)$, and at $N=\Theta(n^2)$ it is $\Theta(1)$.
Intermediate widths follow the same law, provided $N$ is a full square-grid size in the stated range.

The upper construction depends on the freedom to choose a fixed placement.
A prescribed placement, by contrast, can concentrate data behind a narrow boundary even when the surrounding grid is large.
Corollary~\ref{cor:quantum-tanner-bottleneck} still bounds that region's depth.
The construction also uses intermediate measurements, including data-pivot measurements, and free classical feedforward.
It establishes an exact noiseless tradeoff but does not establish fault tolerance or the same tradeoff for circuits restricted to unitary gates followed by final readout.

\section{Discussion}
\label{sec:discussion}

An ideal syndrome measurement has an intrinsic cut cost that can be read from the stabilizer space before a circuit is chosen.
The intrinsic cut cost $\chi_U$ counts the constraints that remain after stabilizers confined to either side have been removed.
Its operational meaning depends on preserving the quantum state within each syndrome sector.
Under this requirement and with free local processing, the intrinsic cut cost is exactly the Bell-pair cost when all operations are local, and exactly the crossing-gate cost when only CNOTs cross.
Arbitrary crossing two-qubit operations instead reduce the optimum to $\lceil\chi_U/2\rceil$.
The fixed-family small-region theorem and the adaptive grid compiler then connect this local resource count to a space--depth tradeoff.
Several questions remain beyond the models settled here.

\paragraph{Prescribed layouts and elapsed depth.}
The grid compiler chooses a fixed data placement suited to the available space.
For a prescribed placement, however, a small occupied region may retain a narrow required boundary even if many distant ancillas are available.
A useful next step is to characterize which placements admit the same scaling and which incur additional routing costs.
On modular devices, the exact cut budget likewise leaves the depth of local processing and storage transfers to be determined.
The gross-code construction in particular attains a crossing resource optimum, not a minimum elapsed extraction time.
Restricting all measurements to final readout removes another ingredient of the adaptive compiler and requires a separate upper-bound analysis.

\paragraph{Finite gaps.}
On the verified circuits of Appendix~\ref{app:certified-circuits}, the certified rank-two bounds and the two-qubit depths differ by factors between about $1.3$ and $18$.
On a Tanner incidence graph the hardware boundary of every region is wide, and the verified depths are explained by the incidence of the listed checks rather than by any cut.
Whether a circuit that changes generators or measures during the computation can beat the incidence bound of a listed check list on its own Tanner graph is open.
On the bus grids, closing the gap would require either deeper certificates or a proof that the routing overhead of these sizes is unavoidable.

\paragraph{Networks beyond trees.}
On a tree, the projector resource and the pure-state preparation theorem allow all edge budgets to be attained together.
Cycles provide alternative routes for distributing entanglement, but the individual cut inequalities need not characterize a feasible global resource allocation.
Determining when they do, or how large the gap can be, would extend the single-cut characterization to more general architectures.
Any such statement must specify whether it charges delivered entangled pairs, elementary links, storage, or the quantum depth needed to prepare the resource.

\paragraph{Balanced cuts and other code families.}
The uniform Tanner theorem controls every region within a fixed small fraction of the block, but it does not establish a linear lower bound on all balanced cuts.
Beyond the correctable range, a region can also support logical observables, and there the relation between its stabilizer dimension and pure-state entanglement changes.
Even within small regions, extending the theorem to another qLDPC family requires control of products of checks that cancel outside the region.
Good rate and distance constrain aggregate geometry, but they do not by themselves exclude locally fixed qubits or other weakly correlated subsets.
Nor can expansion of a selected presentation supply this missing information, as the annular examples demonstrate.

\paragraph{Noisy extraction.}
Exactness and fault tolerance impose different requirements.
The compiler counts all its quantum operations but allows perfect gates, measurements, and unrestricted classical feedforward.
In particular, errors in long resource chains can accumulate even when simultaneous fusion has constant ideal quantum depth.
A noisy construction would therefore have to account for those errors, the reliability of its measurement records, and repeated correction rounds.
The approximate-instrument bound gives a constraint on a specified approximation to one complete measurement, but it does not prove a noise threshold or an efficient decoder.
Relating the intrinsic cut cost to the overhead of reliable repeated extraction remains open.

\clearpage
\appendix
\addtocontents{toc}{\protect\newpage}
\section{Stabilizer algebra and measurement instruments}
\label{app:proof-details}

This appendix proves the algebraic and instrument identities used in Sections~\ref{sec:models} and~\ref{sec:stabilizer-cut-rank}.
We first relate the stabilizers confined to each side of a cut to matrix ranks and to the correlations of the code state.
We then compute the Schmidt decomposition of a syndrome projector and follow it through each complete measurement branch.
The same calculation gives the sectorwise correctness criterion and the approximate-instrument bound.
All entropies and logarithms are to base two.

\subsection{Rank and local support}
\label{subsec:stabilizer-cut-rank-proof-details}

A stabilizer confined to one side can be measured without crossing the cut.
To count the constraints that remain, we therefore subtract the dimensions of the two confined subspaces from the full stabilizer dimension.
The following definition restates the intrinsic cut cost in notation suited to the proofs.

\begin{definition}[Stabilizer subspaces]
\label{def:cross-cut-stabilizer-rank}
Let $\cS\subseteq\mathbb F_2^{2Q}$ be a commuting binary stabilizer space and let $U\sqcup U^c=Q$.
Write $\cS_U$ and $\cS_{U^c}$ for its subspaces supported entirely on the indicated qubits.
Define
\begin{equation}
\chi_U(\cS)
=
\dim\cS-\dim\cS_U-\dim\cS_{U^c}.
\label{eq:cross-cut-stabilizer-rank-definition}
\end{equation}
\end{definition}

To prove Lemma~\ref{lem:cross-cut-rank-formula}, observe that the stabilizers confined to $U$ are exactly those that vanish upon restriction to its complement.
If $H$ has row space $\cS$, the restriction map $\pi_{U^c}:\cS\to\mathbb F_2^{2|U^c|}$ therefore has kernel $\cS_U$ and image equal to the row space of $H_{U^c}$.
Applying rank-nullity on each side gives
\begin{align}
\dim\cS_U&=\rank(H)-\rank(H_{U^c}),\\
\dim\cS_{U^c}&=\rank(H)-\rank(H_U).
\end{align}
Subtracting these local dimensions from $\dim\cS$ proves Eq.~\eqref{eq:cross-cut-rank-formula}, independently of the check matrix.

Ancilla coordinates add zero $X$ and $Z$ columns, because the measured observables act on the data.
These columns change neither side of Eq.~\eqref{eq:cross-cut-rank-formula}.
We specify the initialization of product ancillas as part of the circuit, while the measured stabilizer space remains a property of the data.

\subsection{Correlations of the code state}
\label{subsec:correctable-cut-entropy}

The entropy of the maximally mixed code state relates $\chi_U(\cS)$ to the local stabilizer dimensions.
Writing $P_s$ for the signed stabilizer operator with binary label $s$, we have
\begin{equation}
\omega
=
2^{-n}\sum_{s\in\cS}P_s.
\end{equation}
The partial trace over $U^c$ removes the terms with $s\notin\cS_U$.
The reduced state $\omega_U$ is then a normalized stabilizer projector with $2^{|U|-\dim\cS_U}$ equal nonzero eigenvalues.
Its entropy is therefore
\begin{equation}
S(\omega_U)=|U|-\dim\cS_U.
\end{equation}
Since the whole code state has entropy $S(\omega)=k=n-\dim\cS$, subtracting this joint entropy from the two reduced entropies gives
\begin{align}
I(U:U^c)_\omega
&=|U|-\dim\cS_U+|U^c|-\dim\cS_{U^c}-k\nonumber\\
&=\chi_U(\cS).
\end{align}

To express $\chi_U(\cS)$ through the stabilizers supported on $U$ alone, we now assume that erasure of $U$ is correctable.
A Pauli supported on a correctable data region cannot act as a nontrivial logical operator.
If such a Pauli commutes with every stabilizer, it must therefore be a stabilizer itself.
In terms of the symplectic projection onto $U$, which we denote by $p_U$, this says that the supported centralizer is
\begin{equation}
(p_U\cS)^{\perp_U}
=
\cS^\perp\cap\mathbb F_2^{2U}
=
\cS_U.
\end{equation}
Since a subspace and its symplectic orthogonal complement have dimensions summing to $2|U|$, the centralizer identity gives
\begin{equation}
\dim p_U\cS=2|U|-\dim\cS_U.
\end{equation}
Because the projection discards exactly $\cS_{U^c}$, we can use $\ker(p_U|_{\cS})=\cS_{U^c}$ in the definition of $\chi_U(\cS)$ to obtain
\begin{align}
\chi_U(\cS)
&=\dim p_U\cS-\dim\cS_U\nonumber\\
&=2\bigl(|U|-\dim\cS_U\bigr)
=2S(\omega_U).
\end{align}

Every set with $|U|<d$ is correctable, so this identity holds throughout the distance regime in Lemma~\ref{lem:correctable-cut-entropy}.
An upper bound on $\dim\cS_U$ for the chosen code family therefore lower-bounds the mutual information between that data region and its complement, taken in the maximally mixed code state.
For a general region, this mutual information can include classical correlations.
In the correctable case, the identity for pure code states below identifies half of it with entanglement entropy.

Correctability also fixes the reduced state on $U$ for every code state.
A Pauli supported in $U$ either anticommutes with a stabilizer or commutes with all stabilizers.
In the first case, its expectation vanishes in every code state.
In the second case, it is a stabilizer up to its sign, and its expectation is fixed by that sign.
These Pauli expectations determine the reduced density matrix.
The equality $\rho_U=\omega_U$ therefore holds for every density operator $\rho$ supported in the code space.
In particular, every pure code state $|\psi\rangle$ satisfies
\begin{equation}
S(\psi_U)
=|U|-\dim\cS_U
=\frac12\chi_U(\cS).
\label{eq:pure-code-correctable-entropy}
\end{equation}
For a pure state, this quantity is the entanglement entropy across $U|U^c$.
The supported-stabilizer bound of Theorem~\ref{thm:quantum-tanner-target-window} therefore gives $S(\psi_U)\ge\tau|U|$ for every pure code state and every region in its stated correctable range.

\subsection{Syndrome projectors and adaptive trajectories}
\label{subsec:adaptive-branch-proof}
\label{subsec:raw-record-proof}

We prove the operational bound by resolving the circuit into complete Kraus trajectories.
Exact extraction forces each nonzero trajectory to implement, up to a scalar, the projector of its reported syndrome.
The operator Schmidt rank of that projector then bounds the number of crossing operations on the trajectory.
The argument permits intermediate measurements and arbitrary classical adaptation without using a Clifford normal form.
We first establish the decomposition of the projectors and then analyze the trajectories.

\paragraph{Projector decomposition.}
Fix a data partition $U|U^c$ and a syndrome $s$.
Let $\mathsf S^{(s)}$ be the signed stabilizer group whose common $+1$ eigenspace is this syndrome sector.
Its binary space is $\cS$, independent of $s$, and we write
\begin{equation}
\begin{aligned}
r&=\dim\cS, & a&=\dim\cS_U,\\
b&=\dim\cS_{U^c}, & c&=r-a-b.
\end{aligned}
\end{equation}
Let $\Pi_{U,s}$ and $\Pi_{U^c,s}$ be the projectors of the two local signed stabilizer subgroups.
Choose representatives $h_t=P_t\otimes Q_t$ of the quotient of $\mathsf S^{(s)}$ by those subgroups, absorbing any phase into one tensor factor.
Grouping the group average for the syndrome projector by these cosets gives
\begin{equation}
\Pi_s
=2^{-c}\sum_{t=1}^{2^c}
(P_t\Pi_{U,s})\otimes(Q_t\Pi_{U^c,s}).
\label{eq:projector-coset-decomposition}
\end{equation}
Every representative commutes with the full stabilizer group, so each local restriction commutes with the corresponding local projector.

To identify this expansion as a Schmidt decomposition, it suffices to show that the factors on each side are mutually orthogonal and of equal norm.
For two different representatives, the Hilbert--Schmidt inner product of their left factors is
\begin{equation}
\operatorname{Tr}\!\left(
\Pi_{U,s}P_t^\dagger P_{t'}\Pi_{U,s}
\right)
=\operatorname{Tr}\!\left(\Pi_{U,s}P_t^\dagger P_{t'}\right).
\end{equation}
Pauli orthogonality makes this trace zero unless $P_t^\dagger P_{t'}$ belongs, up to phase, to the local stabilizer group on $U$.
In that case, multiplying the full stabilizer $h_t^\dagger h_{t'}$ by the corresponding local stabilizer would leave a stabilizer supported entirely on $U^c$.
The two representatives would then lie in the same quotient coset, which is a contradiction.
Distinct left factors are therefore orthogonal, and the same argument applies to the right factors.
With $\|\cdot\|_2$ denoting the Hilbert--Schmidt norm, their squared norms are independent of $t$,
\begin{equation}
\|P_t\Pi_{U,s}\|_2^2=2^{|U|-a},
\qquad
\|Q_t\Pi_{U^c,s}\|_2^2=2^{|U^c|-b}.
\end{equation}
Since $\operatorname{Tr}\Pi_s=2^{n-r}$, normalizing the vectorization gives $2^c$ orthogonal Schmidt terms, each with coefficient
\begin{equation}
\frac{2^{-c}2^{(|U|-a)/2}2^{(|U^c|-b)/2}}
{2^{(n-r)/2}}
=2^{-c/2}.
\end{equation}
This proves Lemma~\ref{lem:projector-choi-schmidt}, including both its rank and entropy assertions, for every syndrome and every cut.
No correctability assumption enters the calculation.

\paragraph{Complete Kraus trajectories.}
For each instrument outcome $o$, choose a finite Kraus representation of its completely positive map.
Keep $o$ as the outcome used by the controller of the circuit, and introduce a separate index for its Kraus refinement.
Later choices depend on the actual observed record, not on these unobserved refinement indices.
Once a trajectory is fixed, all choices are fixed as well.
The circuit then becomes a sequence of one- and two-qubit linear operators, including nonunitary operations, resets, and discarded local environments through their Kraus descriptions.

Product ancillas that start in mixed states can be purified individually with private purifiers on their original sides of the cut.
At the end, refine the trace over all remaining ancillas and purifiers into local computational-basis bra operators.
Any earlier discarded wire can equivalently be contracted when it is discarded.
These local contractions add no crossing operations.

A complete trajectory $\beta$ then gives a single operator $K_\beta$ from the input data space to the fixed output data space.
If $m(\beta)$ is its reported syndrome, exactness states that
\begin{equation}
\sum_{\beta:m(\beta)=s}K_\beta\rho K_\beta^\dagger
=\Pi_s\rho\Pi_s
\quad\text{for every input }\rho.
\label{eq:adaptive-fine-instrument}
\end{equation}
Equality on all density matrices is equality of linear maps, so it remains true after tensoring with the identity on arbitrary references.
Taking Choi matrices yields
\begin{equation}
\sum_{\beta:m(\beta)=s}
|K_\beta\rangle\!\rangle\langle\!\langle K_\beta|
=
|\Pi_s\rangle\!\rangle\langle\!\langle\Pi_s|.
\label{eq:adaptive-pure-choi}
\end{equation}
The right side has rank one and every summand is positive.
Each nonzero $K_\beta$ is thus a scalar multiple of $\Pi_s$.
This conclusion relies on both the outcome probabilities and the conditional quantum output states of exact extraction, and syndrome statistics alone do not imply it.

\paragraph{Rank growth along a trajectory.}
Prepare one local Bell pair between each input data qubit and a private reference on the same side of the hardware cut.
Together with the locally purified product ancillas, the initial pure state has Schmidt rank one across the enlarged partition.
The references and purifiers introduce no crossing interactions and leave the original hardware boundary unchanged.

Let a pure vector on this partition have Schmidt rank $R$.
A product operator acting separately on the two sides maps its $R$ product terms to at most $R$ product terms, even if the operator is nonunitary.
A two-qubit Kraus operator across the cut, in contrast, has an operator Schmidt decomposition
\begin{equation}
M=\sum_{j=1}^{q}A_j\otimes B_j,
\qquad q\le4,
\end{equation}
because the operator space of one qubit has dimension four.
Such an operator therefore maps a rank-$R$ vector to a sum of at most $4R$ product terms.
The bound applies to each Kraus operator separately, which avoids any need for monotonicity under averaging or conditioning.
Classical communication only selects the operators on a fixed trajectory and does not itself change the quantum vector.

After $g_L(\beta)$ crossing operations and the final local contractions, the data--reference vector is proportional to $|K_\beta\rangle\!\rangle$ and has Schmidt rank at most $4^{g_L(\beta)}$.
For every nonzero trajectory, Eq.~\eqref{eq:adaptive-pure-choi} and Lemma~\ref{lem:projector-choi-schmidt} therefore give
\begin{equation}
2^{\chi_{Q\cap L}(\cS)}
\le4^{g_L(\beta)},
\qquad
\chi_{Q\cap L}(\cS)\le2g_L(\beta).
\end{equation}
Each trajectory contains at most $\depth(C)|\partial_{G_{\hw}}L|$ crossing operations, because the depth is a uniform bound and each layer has disjoint supports on hardware edges.
For a nonzero boundary, division and maximization over cuts prove
\begin{equation}
\depth(C)\ge\frac12
\max_{L:\,|\partial L|>0}
\frac{\chi_{Q\cap L}(\cS)}{|\partial L|}.
\label{eq:intro-full-model}
\end{equation}
If the boundary is empty, the same rank inequality makes a positive intrinsic cut cost impossible.
When the crossing primitives consist of unitaries and two-qubit Pauli measurements, the latter have rank at most two in their operator Schmidt decomposition.
The same multiplication argument then gives Eq.~\eqref{eq:weighted-crossing-primitives}.

\begin{proposition}[Sectorwise instrument rigidity]
\label{prop:sectorwise-instrument-rigidity}
Let $\{\Pi_s\}_s$ be the mutually orthogonal, nonzero syndrome projectors on the fixed input and output data space, with $\sum_s\Pi_s=I$.
Let $\{\mathcal E_o\}_o$ be a finite instrument whose sum is trace preserving, let $s(o)$ be its reported syndrome, and let $P_o$ be a physical Pauli operator specified by the observed record $o$.
Assume that for every syndrome $t$ and every unit vector $|\psi\rangle\in\operatorname{ran}\Pi_t$,
\begin{equation}
\begin{aligned}
\sum_{o:s(o)\ne t}\operatorname{Tr}\mathcal E_o(|\psi\rangle\langle\psi|)&=0,\\
\sum_{o:s(o)=t}P_o^\dagger\mathcal E_o(|\psi\rangle\langle\psi|)P_o&=|\psi\rangle\langle\psi|.
\end{aligned}
\label{eq:sectorwise-correctness}
\end{equation}
For any Kraus refinement $\mathcal E_o(\rho)=\sum_\alpha K_{o,\alpha}\rho K_{o,\alpha}^\dagger$, there are scalars $c_{o,\alpha}$ such that
\begin{equation}
\begin{aligned}
K_{o,\alpha}&=c_{o,\alpha}P_o\Pi_{s(o)},\\
\sum_{o:s(o)=s}\sum_\alpha|c_{o,\alpha}|^2&=1\qquad\text{for every }s.
\end{aligned}
\label{eq:sectorwise-kraus-rigidity}
\end{equation}
Consequently, for every data operator $X$,
\begin{equation}
\sum_{o:s(o)=s}P_o^\dagger\mathcal E_o(X)P_o=\Pi_sX\Pi_s.
\label{eq:sectorwise-full-instrument}
\end{equation}
This equality also holds after tensoring the maps with the identity on any reference system.
Under the physical resource assumptions of Definition~\ref{def:adaptive-local-extraction}, the circuit obeys Theorem~\ref{thm:adaptive-instrument-cut-bound} and the primitive-dependent refinements of Remark~\ref{rem:adaptive-sharp-constant}, even when $P_o$ is tracked classically rather than physically removed.
\end{proposition}

\begin{proof}
Fix an observed outcome $o$ with $s(o)=s$ and a Kraus index $\alpha$.
For $t\ne s$, the first condition in Eq.~\eqref{eq:sectorwise-correctness} is a sum of nonnegative squared norms, so $K_{o,\alpha}\Pi_t=0$.
For a unit vector $|\psi\rangle\in\operatorname{ran}\Pi_s$, the second condition is a sum of positive rank-one operators equal to $|\psi\rangle\langle\psi|$.
Every vector $P_o^\dagger K_{o,\alpha}|\psi\rangle$ is therefore proportional to $|\psi\rangle$.
A linear operator for which every vector of a subspace is an eigenvector is scalar on that subspace, because applying it to two basis vectors and their sum makes their eigenvalues equal.
The same scalar-on-a-code-subspace observation underlies the exact recovery criterion of Knill and Laflamme~\cite[Theorem~III.1]{knillTheoryQuantumErrorcorrecting1997}.
Applied here, it gives $P_o^\dagger K_{o,\alpha}\Pi_s=c_{o,\alpha}\Pi_s$, including in the rank-one-sector case, where the assertion is immediate.
Summing $K_{o,\alpha}\Pi_t$ over the complete family of sectors proves the first identity in Eq.~\eqref{eq:sectorwise-kraus-rigidity}.
Taking the trace of the correct-sector output gives the second identity.
Substitution proves Eq.~\eqref{eq:sectorwise-full-instrument} on every operator, including off-diagonal blocks between different sectors, and hence on arbitrary mixed inputs and arbitrary extensions by a reference.

To apply this conclusion to the hardware circuit, purify each mixed product ancilla locally.
Refine all Kraus indices and final ancillary traces as in the complete-trajectory proof.
Each complete trajectory $\beta$ then has $K_\beta=c_\beta P_{o(\beta)}\Pi_{s(o(\beta))}$.
For $U=Q\cap L$, a physical Pauli has the factorization
\begin{equation}
\begin{aligned}
P_o&=e^{i\theta_o}\bigotimes_{q\in Q}P_{o,q}\\
&=e^{i\theta_o}P_{o,U}\otimes P_{o,U^c},
\end{aligned}
\end{equation}
where $P_{o,q}\in\{I,X,Y,Z\}$.
Multiplication by these local unitaries preserves the operator Schmidt rank and the Schmidt coefficients of the normalized vectorization~\cite[Sec.~II]{nielsenQuantumDynamicsPhysical2003}.
Every nonzero complete trajectory therefore has operator Schmidt rank $2^{\chi_U(\cS)}$ by Lemma~\ref{lem:projector-choi-schmidt}.
On the product of local data--reference Bell pairs, crossing Kraus operators of operator Schmidt ranks $r_1,\ldots,r_g$ can produce Schmidt rank at most $\prod_jr_j$.
Comparing the two ranks gives
\begin{equation}
\chi_U(\cS)\le\sum_{j=1}^{g}\log_2r_j\le2g_L(\beta),
\end{equation}
with the stronger bound $\chi_U(\cS)\le g_L(\beta)$ when every crossing Kraus operator has operator Schmidt rank at most two.
The same layer-capacity argument then gives the depth bounds directly for the framed circuit.
\end{proof}

\paragraph{Meaning of the frame convention.}
The frame must be determined by the actual observed record, not by hidden Kraus indices.
The reported syndrome refers to the input sector even if $P_o$ changes the output sector.
For $P_o=I$, Proposition~\ref{prop:sectorwise-instrument-rigidity} is equivalent to the original all-input L\"uders requirement.
Forgetting the frame need not leave a L\"uders instrument.
Permitting a recorded Pauli frame is therefore an explicit extension of the output convention in Definition~\ref{def:adaptive-local-extraction}.
If a physical unframed output is required, all frame corrections can be applied in at most one additional layer of one-qubit gates.
That layer requires no additional crossing operation.
An arbitrary entangling change of output encoding is not a physical Pauli frame and need not preserve the cut bound.

\subsection{Approximate extraction on the reference input}
\label{subsec:approximate-branch-proof}

We now prove Corollary~\ref{cor:approximate-instrument-cut-bound}.
The circuit obeys the resource assumptions of Definition~\ref{def:adaptive-local-extraction}, namely the one-qubit sites, edge-supported two-qubit primitives, product ancillas, fixed data interface, finite classical records, and worst-case depth convention.
The circuit need not implement the ideal instrument exactly, however.
Its output is a trace-preserving instrument on the same quantum data and classical syndrome spaces as the ideal channel.
We first compare the fidelity with an average overlap and then control that overlap through the Schmidt rank of each trajectory.

\paragraph{Fidelity and conditional overlap.}
Let $r=\dim\cS$ and write $|\phi_s\rangle=|\widehat\Pi_s\rangle$ for the normalized vectorized syndrome projector.
On the normalized reference input $\Phi$ of Eq.~\eqref{eq:reference-bell-input}, the ideal and actual outputs have the form
\begin{equation}
\begin{aligned}
\sigma&=\sum_s p_s|s\rangle\langle s|\otimes|\phi_s\rangle\langle\phi_s|,\\
\rho&=\sum_s q_s|s\rangle\langle s|\otimes\rho_s.
\end{aligned}
\label{eq:approximate-cq-outputs}
\end{equation}
The ideal probability is $p_s=2^{-r}$, since $\operatorname{Tr}\Pi_s=2^{n-r}$ and the data marginal of $\Phi$ is maximally mixed.
Terms with $q_s=0$ can be omitted, and we put $f_s=\langle\phi_s|\rho_s|\phi_s\rangle$ and $W=\sum_s q_sf_s$.
The block-diagonal form and the purity of each ideal conditional state give
\begin{equation}
\begin{aligned}
F(\rho,\sigma)
&=\left(\sum_s\sqrt{p_sq_sf_s}\right)^2\\
&\le\left(\sum_s p_s\right)\left(\sum_s q_sf_s\right)
=W.
\end{aligned}
\label{eq:fidelity-overlap-comparison}
\end{equation}
The inequality is the Cauchy--Schwarz inequality, and $F$ is the squared fidelity of Eq.~\eqref{eq:approximate-test-fidelity}.

\paragraph{Overlap on a refined trajectory.}
Refine the circuit and its final local ancillary traces as in Appendix~\ref{subsec:adaptive-branch-proof}.
Local purifications handle mixed product ancillas without adding crossing operations.
For each nonzero complete trajectory $\beta$, let $q_\beta$ be its probability on $\Phi$, let $s(\beta)$ be its reported syndrome, and let $|v_\beta\rangle$ be its normalized data--reference output.
These weights sum to one, and
\begin{equation}
q_s\rho_s
=\sum_{\beta:s(\beta)=s}q_\beta|v_\beta\rangle\langle v_\beta|.
\end{equation}
The rank-growth argument did not use exactness, so it still gives $\operatorname{SR}(v_\beta)\le4^{g_L(\beta)}$.

Let $D=2^{\chi_U}$ be the Schmidt rank of $|\phi_s\rangle$ at the chosen cut.
By Lemma~\ref{lem:projector-choi-schmidt}, its reduced state on either side is a rank-$D$ projector divided by $D$.
If a normalized vector $|v\rangle$ has Schmidt rank $R$ and $P_v$ projects onto its left Schmidt support, then
\begin{equation}
\begin{aligned}
|\langle\phi_s|v\rangle|^2
&\le\|(P_v\otimes I)|\phi_s\rangle\|^2\\
&\le\min\{1,R/D\}.
\end{aligned}
\label{eq:flat-target-overlap}
\end{equation}
The first step is the Cauchy--Schwarz inequality, and the second uses the flat reduced spectrum, which gives the usual Schmidt-rank overlap bound \cite{terhalSchmidtNumberDensity2000}.
Applying this bound to the trajectory vectors and averaging proves the central estimate
\begin{equation}
F\le W
\le\sum_\beta q_\beta
\min\!\left\{1,2^{2g_L(\beta)-\chi_U}\right\}.
\label{eq:approximate-crossing-moment}
\end{equation}

\paragraph{Depth and probability.}
Write $t=\depth(C)$ and $b_L=|\partial_{G_{\hw}}L|$, so that each trajectory satisfies $g_L(\beta)\le t b_L$.
For $F>0$ and $b_L>0$, Eq.~\eqref{eq:approximate-crossing-moment} therefore gives $F\le2^{2tb_L-\chi_U}$ and hence Eq.~\eqref{eq:approximate-cut-depth}.
For $b_L=0$ it gives $F\le2^{-\chi_U}$.

To convert the overlap estimate into a probability bound, suppose more generally that $W\ge\alpha>0$ and choose $0<\theta<\alpha$.
Let $H_\theta$ be the event $g_L(\beta)\ge(\chi_U+\log_2\theta)/2$, with probability $p_\theta$ under the test input.
On its strict complement, the trajectory overlap is at most $\theta$, while on $H_\theta$ it is at most one.
The two cases together give
\begin{equation}
\alpha\le W\le\theta+(1-\theta)p_\theta,
\qquad
p_\theta\ge\frac{\alpha-\theta}{1-\theta}.
\label{eq:approximate-general-tail}
\end{equation}
Taking $\alpha=F$ and $\theta=F/2$ yields
\begin{equation}
\Pr_\Phi[g_L(\beta)\ge a_F]
\ge\frac{F}{2-F}\ge\frac F2.
\end{equation}
This proves Eq.~\eqref{eq:approximate-cut-tail}, including its threshold convention.
The crossing count depends only on the actual circuit trajectory, so summing the hidden Kraus refinements does not change the physical probability of this event.
This probability is specific to $\Phi$, or equivalently to its maximally mixed data marginal, and it is not asserted for arbitrary data inputs.

\paragraph{Diamond-distance guarantee.}
A diamond-distance guarantee also controls the same overlap, as we see by defining the projector
\begin{equation}
E=\sum_s|s\rangle\langle s|\otimes|\phi_s\rangle\langle\phi_s|.
\end{equation}
It satisfies $0\le E\le I$, $\operatorname{Tr}(E\sigma)=1$, and $\operatorname{Tr}(E\rho)=W$.
For the normalized outputs, the variational bound for trace distance gives
\begin{equation}
W\ge1-\frac12\|\rho-\sigma\|_1.
\end{equation}
A guarantee $\tfrac12\|\mathcal E-\mathcal M_{\cS}\|_\diamond\le\varepsilon<1$ therefore implies $W\ge1-\varepsilon$ on the test input.
Using $\alpha=1-\varepsilon$ in the same moment and probability arguments yields
\begin{equation}
\depth(C)\ge
\frac{\chi_U-\log_2(1/(1-\varepsilon))}{2b_L}
\label{eq:approximate-diamond-depth}
\end{equation}
for $b_L>0$, together with the corresponding probability bound with $F$ replaced by $1-\varepsilon$.
If the error convention instead bounds the full diamond norm by $\varepsilon$, the parameter is $1-\varepsilon/2$.
This derivation uses the overlap witness directly and does not require a lower bound of $1-\varepsilon$ on the squared fidelity.

The same global reference input serves for all cuts, and only the grouping of its data--reference pairs changes from one cut to another.
Both the test fidelity and a diamond-distance guarantee can therefore be used when the depth lower bound is optimized over hardware cuts.
These conclusions concern the implemented instrument under the stated circuit resources.

\section{Attaining protocols and tree resources}
\label{app:single-cut-protocols}

The attaining protocol rests on the fact that one shared Bell pair measures a crossing Pauli parity without revealing either local factor.
We prove that fact on every measurement branch and use it to attain the single-cut costs of Theorem~\ref{thm:single-cut-resource}.
We also explain the relation to earlier entanglement-cost results.
We then give the sparse gross-code basis and the projector-state reduction that makes the Bell-pair bounds of all tree edges attainable at once.
Throughout the single-cut arguments, Alice holds $U$, Bob holds $U^c$, and $\chi=\chi_U(\cS)$.
The local processing and storage of both parties are unrestricted.

\subsection{The remote-parity protocol}
\label{subsec:single-cut-all-branches}

\paragraph{A cut-adapted basis and output labels.}
Fix independent commuting Hermitian Pauli generators $g_1,\ldots,g_r$ of the signed stabilizer group, which does not contain $-I$.
Lemma~\ref{lem:minimum-crossing-generators} gives a basis $b_1,\ldots,b_r$ of $\cS$ in which $b_1,\ldots,b_\chi$ cross the cut.
Every other element of this basis is supported on $U$ or on $U^c$.
Write $b_k=\sum_jN_{kj}g_j$ in binary coordinates, with $N$ invertible.
Lift $b_k$ to the signed operator $h_k=\prod_jg_j^{N_{kj}}$.
Each $h_k$ is a Hermitian Pauli operator because the $g_j$ commute, and the $h_k$ generate the same signed group.
The common eigenspace with syndrome $s$ for the $g_j$ is the common eigenspace with syndrome $Ns$ for the $h_k$.
Measuring the $h_k$ and applying $N^{-1}$ to their outcomes therefore implements the same instrument.
Suppose the output register is labeled by a possibly dependent list $f_a=(-1)^{d_a}\prod_jg_j^{M_{aj}}$ of commuting Hermitian Pauli observables whose binary representatives lie in $\cS$.
The protocol then writes the bits $m_a=d_a+\sum_jM_{aj}s_j$ modulo two.
These labels are affine functions of the independent syndrome.
Combinations of labels that are inconsistent with the group relations therefore occur with probability zero and carry no additional information.

\paragraph{The remote-parity gadget.}
Let $P$ be a crossing Hermitian Pauli operator.
Write $P=P_A\otimes P_B$ with Hermitian Pauli factors, absorbing a sign into $P_A$.
Alice and Bob share $|\Phi\rangle_{ab}=(|00\rangle+|11\rangle)/\sqrt2$ on an ancilla $a$ of Alice and an ancilla $b$ of Bob.
Alice applies controlled-$P_A$ from $a$ to her data, and Bob applies controlled-$P_B$ from $b$ to his data.
A controlled Pauli operator on several qubits of one party is a product of controlled single-qubit Pauli operators.
Controlled-$(-P_A)$ differs from controlled-$P_A$ by a Pauli $Z$ on the control.
Both operations are therefore local Clifford unitaries.
To verify the quantum output, apply these operations to any vector $|\psi\rangle$ of the data and any references.
The joint state then becomes
\begin{equation}
\frac{1}{\sqrt2}\bigl(|00\rangle_{ab}\otimes I+|11\rangle_{ab}\otimes P\bigr)|\psi\rangle.
\end{equation}
Reading $a$ and $b$ in the $X$ basis with outcomes $u$ and $v$ contracts the ancillas with $\langle u|_X\langle v|_X$, where $\langle u|_X|x\rangle=(-1)^{ux}/\sqrt2$.
The resulting Kraus operator on the data and references is
\begin{equation}
\begin{gathered}
K_{uv}=\frac{I+(-1)^{u+v}P}{2\sqrt2}=\frac{1}{\sqrt2}P_{u\oplus v},\\
P_t=\frac{I+(-1)^tP}{2}.
\end{gathered}
\label{eq:remote-parity-kraus}
\end{equation}
Each parity $t=u\oplus v$ arises from two raw pairs.
Hence $\sum_{u\oplus v=t}K_{uv}\rho K_{uv}^\dagger=P_t\rho P_t$ for every input $\rho$, and no correction of the data is needed.
The raw bit $u$ is uniformly distributed given $t$ and carries no further information about the data.
The measured ancillas are left in $X$ eigenstates, which local Pauli corrections return to $|+\rangle$, and the Bell pair is consumed.
Beckman et al.\ already measure $XX$ and $ZZ$ nondestructively in this way, with CNOT gates and one shared pair per parity.
They attribute the extension to CSS codes to Steane \cite{beckmanCausalLocalizableQuantum2001}.
Lim, Hhan, and Kwon use the controlled-Pauli form for nondestructive discrimination of stabilizer states \cite{limTradeoffInformationGain2025}.
A local element $h$ supported on one side is measured by the party holding that side, with an ancilla in $|+\rangle$, a controlled-$h$ operation, and an $X$ readout.
The resulting local measurement applies $(I\pm h)/2$ without the factor $1/\sqrt2$.

\paragraph{All branches of the protocol.}
Measure the $\chi$ crossing elements $h_1,\ldots,h_\chi$ with the gadget, and the $r-\chi$ local elements with local ancillas.
Carry out these measurements in one fixed order.
The restrictions of two crossing elements to one side may anticommute.
The elements themselves nevertheless commute, so their projectors commute.
A raw branch assigns an outcome bit $s'_k$ to each basis element, together with a second, uniformly distributed raw bit for each gadget.
By Eq.~\eqref{eq:remote-parity-kraus}, the Kraus operator of the branch on the data is
\begin{equation}
2^{-\chi/2}\prod_{k=1}^r\frac{I+(-1)^{s'_k}h_k}{2}=2^{-\chi/2}\Pi_{s'},
\end{equation}
where $\Pi_{s'}$ is the joint projector for the $h_k$.
The $2^\chi$ raw branches with the same $s'$ therefore sum to $\Pi_{s'}\rho\Pi_{s'}$ on every input, including inputs entangled with arbitrary references.
The relabeling by $N^{-1}$ followed by the affine output map described above then gives the required instrument.
No quantum operation depends on a remote outcome.
Classical communication is needed only to combine the raw bits into the output register.
The protocol thus establishes sufficiency in part~\ref{it:single-cut-bell} of Theorem~\ref{thm:single-cut-resource} without any normal-form theorem.

\subsection{Shared pairs, crossing layers, and local storage}
\label{subsec:preparing-shared-pairs}
\label{subsec:crossing-layers-storage}

\paragraph{Preparing the shared pairs.}
A single crossing CNOT gate from Alice's $|+\rangle$ to Bob's $|0\rangle$ prepares $|\Phi\rangle$.
For two pairs, Alice prepares $|\Phi\rangle_{aa'}$ and Bob prepares $|\Phi\rangle_{bb'}$ locally.
A crossing SWAP gate on $a$ and $b$ then produces
\begin{equation}
|\Phi\rangle_{ba'}\otimes|\Phi\rangle_{ab'},
\end{equation}
two Bell pairs shared across the cut \cite{eisertOptimalLocalImplementation2000}.
The data do not move, since $a,a',b,b'$ are ancillas.
Using $\lfloor\chi/2\rfloor$ SWAP gates and, when $\chi$ is odd, one CNOT gate therefore prepares the $\chi$ pairs that Appendix~\ref{subsec:single-cut-all-branches} consumes.
This preparation proves the upper bounds in parts~\ref{it:single-cut-cnot} and~\ref{it:single-cut-general} of Theorem~\ref{thm:single-cut-resource}, including $\chi=1$.
The same Bell state is prepared by a crossing CZ gate followed by a Hadamard on Bob's qubit, since $(I\otimes H)\operatorname{CZ}|++\rangle=|\Phi\rangle$.
A crossing $ZZ$ measurement on $|++\rangle$ followed by an outcome-dependent local Pauli correction also prepares $|\Phi\rangle$.

To prove the matching lower bounds, apply the rank argument of Appendix~\ref{subsec:adaptive-branch-proof}, which allows arbitrary local instruments, with Alice's and Bob's laboratories as the two sides.
On the reference input, every nonzero complete Kraus branch must end in a vector proportional to the flat state of Lemma~\ref{lem:projector-choi-schmidt}, whose Schmidt rank is $2^\chi$.
Local Kraus operators cannot increase the Schmidt rank, whereas shared Bell pairs contribute an initial rank $2^m$.
A crossing CNOT gate has operator Schmidt rank two, and a crossing two-qubit Kraus operator has operator Schmidt rank at most four.
Hence $m\ge\chi$, at least $\chi$ crossing CNOT gates, and at least $\lceil\chi/2\rceil$ arbitrary crossing instruments are needed on each nonzero branch, which completes the proof of parts~\ref{it:single-cut-bell}--\ref{it:single-cut-general}.

\paragraph{Crossing layers and local storage.}
Suppose that the crossing links have a maximum matching $\mathcal M$ of size $\nu\ge1$.
In one layer, the crossing operations act on disjoint qubits and hence on at most $\nu$ links.
A branch with $B$ layers containing crossing operations therefore applies at most $\nu B$ crossing operations.
Parts~\ref{it:single-cut-cnot} and~\ref{it:single-cut-general} then give $B\ge\chi/\nu$ for CNOT gates and $B\ge\chi/(2\nu)$ for arbitrary instruments.
For the upper bound, the links of $\mathcal M$ suffice.
In each crossing layer, apply a CNOT gate or a SWAP gate on every link of $\mathcal M$, as described above.
Between crossing layers, move the prepared halves from the link qubits into fresh local storage by local operations.
With CNOT gates there are $\chi$ gate slots, which fit into $\lceil\chi/\nu\rceil$ layers.
With SWAP gates, together with one CNOT gate when $\chi$ is odd, there are $\lceil\chi/2\rceil$ slots, which fit into $\lceil\lceil\chi/2\rceil/\nu\rceil=\lceil\chi/(2\nu)\rceil$ layers.
After all pairs are stored, the parties run the protocol of Appendix~\ref{subsec:single-cut-all-branches}, which proves part~\ref{it:single-cut-layers}.

\subsection{A normal form for one cut}
\label{subsec:bipartite-normal-form}

The following stabilizer-group formulation of the normal form of Audenaert and Plenio \cite{audenaertEntanglementMixedStabilizer2005} separates Bell correlations from classical parity correlations.
It also connects the cut optimum to earlier nondestructive-discrimination bounds.

\begin{lemma}[Bipartite normal form]
\label{lem:bipartite-normal-form}
Let $p$ be half the rank of the form $\omega(v,w)=\langle v_U,w_U\rangle$ on $\cS$ and put $q=\chi-2p$.
There are local Clifford unitaries $C_U\otimes C_{U^c}$ that map the stabilizer group, up to signs, to the group generated by $X_{a_i}X_{b_i}$ and $Z_{a_i}Z_{b_i}$ for $i\le p$, by $Z_{c_j}Z_{d_j}$ for $j\le q$, and by single-qubit $Z$ operators on further qubits of each side, where all listed qubits are distinct, the $a_i$ and $c_j$ lie in $U$, and the $b_i$ and $d_j$ lie in $U^c$.
\end{lemma}

\begin{proof}
The form $\omega$ is alternating, so its rank is even.
To construct the unitaries, first choose bases of $\cS_U$ and $\cS_{U^c}$.
These subspaces are isotropic, so local Clifford unitaries map them to single-qubit $Z$ operators on distinct qubits of the respective sides.
Every element of $\cS$ commutes with these operators, so it acts on their qubits only by $I$ or $Z$.
Multiplying by local basis elements removes these factors.
The remaining crossing elements span a $\chi$-dimensional space $V$ that complements $\cS_U\oplus\cS_{U^c}$.
Restriction to either side is injective on $V$, because an element of $V$ with zero restriction to one side would be a local stabilizer on the other side.
The form $\omega$ vanishes whenever one argument is local, because $\cS$ is isotropic.
Its rank on $\cS$ therefore equals its rank on $V$.
Global isotropy also gives $\omega(v,w)=\langle v_{U^c},w_{U^c}\rangle$.
Symplectic Gram--Schmidt reduction yields a basis of $V$ consisting of $p$ pairs $(e_i,f_i)$ with $\omega(e_i,f_i)=1$ and all other values zero, together with $q$ vectors in the radical of $\omega$.
On Alice's side, the restrictions of this basis are independent with this Gram matrix.
Adjoining dual partners for the radical vectors extends them to part of a symplectic basis of her remaining qubits.
Such partners exist because a subspace whose form has rank $2p$ and a $q$-dimensional radical lies in a nondegenerate subspace of dimension $2(p+q)$.
Hence $p+q$ does not exceed the number of remaining qubits.
A local Clifford unitary then maps the restrictions of $e_i$, $f_i$, and the radical vectors to $X_{a_i}$, $Z_{a_i}$, and $Z_{c_j}$ on distinct qubits.
The restrictions to Bob's side have the same Gram matrix, and a local Clifford unitary maps them to $X_{b_i}$, $Z_{b_i}$, and $Z_{d_j}$.
The basis elements then become $X_{a_i}X_{b_i}$, $Z_{a_i}Z_{b_i}$, and $Z_{c_j}Z_{d_j}$.
Signs can be absorbed by local Pauli operators or recorded as an affine relabeling of outcomes.
\end{proof}

By the lemma, the maximally mixed code state is locally equivalent to $p$ Bell pairs, $q$ classically correlated pairs, and a product of local states.
Its mutual information across the cut is therefore $2p+q=\chi$.
The lemma also recovers Lemma~\ref{lem:projector-choi-schmidt}, because the vectorized projectors of the three kinds of factor have $4$, $2$, and $1$ equal Schmidt coefficients, respectively, and $4^p2^q=2^\chi$.
The normal form explains these correlations, while the direct projector and remote-parity proofs establish the resource bounds without it.

\subsection{Reductions to prior entanglement-cost bounds}
\label{subsec:single-cut-prior-work}

The Bell-pair lower bound also follows from Theorem~2 of Lim, Hhan, and Kwon.
Their theorem requires at least $\log_2 K$ ebits for perfect nondestructive local discrimination of $K$ equally likely orthogonal maximally entangled states \cite{limTradeoffInformationGain2025}.
To apply their theorem, use the normal form of Lemma~\ref{lem:bipartite-normal-form}.
Choose all four Bell states on each of the $p$ pairs constrained by $XX$ and $ZZ$.
Choose $|\Phi^+\rangle$ or $|\Psi^+\rangle$ on each of the $q$ pairs constrained by $ZZ$ alone.
The resulting $K=4^p2^q=2^\chi$ tensor-product states are orthogonal and maximally entangled on the active registers, of local dimension $2^{p+q}$.
Prepare the remaining constrained and unconstrained registers in fixed local pure states.
Then apply the inverse local Clifford transformations.
The $K$ inputs have distinct stabilizer syndromes, and each input is a simultaneous eigenstate of all checks.
Exact syndrome extraction then identifies each input with certainty while preserving it.
After local decoding and the discarding of the fixed local registers, exact extraction becomes a perfect nondestructive discrimination protocol on the original maximally entangled ensemble.
The cited theorem therefore gives the full lower bound of $\chi$ ebits, even though its task imposes only a promise on the input ensemble.
Together with the established remote-parity protocol and the cut-adapted basis, this bound already yields the Bell-pair optimum in Theorem~\ref{thm:single-cut-resource}.

A second route follows from the general instrument theorem of Akibue, Miyazaki, and Osaka \cite[Theorem~3(i)]{akibueOptimizingEntanglementManipulation2026}.
Suppose a bipartite instrument is deterministically implementable by separable instruments assisted by a pure resource $|\tau\rangle$.
That theorem then bounds the Schmidt rank of $|\tau\rangle$ below by the Schmidt number of each branch Choi operator.
LOCC is a subclass of separable operations.
The syndrome branch $\rho\mapsto\Pi_s\rho\Pi_s$ has Choi operator $J_s=|\Pi_s\rangle\!\rangle\langle\!\langle\Pi_s|$, of Schmidt number $2^\chi$ by Lemma~\ref{lem:projector-choi-schmidt}.
Their result thus gives $\operatorname{SR}(|\tau\rangle)\ge2^\chi$, and $m\ge\chi$ when the supplied resource consists of $m$ Bell pairs.
Part~(ii) of their theorem, which is not invoked here, additionally assumes equality of resource and Choi Schmidt ranks and a branch represented by a co-isometry.

Both routes show that the Bell-pair lower bound can also be deduced from prior entanglement-cost results.
The direct proof in Appendix~\ref{subsec:adaptive-branch-proof} additionally keeps track of the permitted crossing primitive on every nonzero refined trajectory.
The matching-layer optima then follow by batching with free local processing and sufficient storage.
Neither the underlying rank monotonicity nor the parity and SWAP gadgets are new ingredients.

\subsection{The gross-code matrix and its sparse certificate}
\label{subsec:finite-code-calculations}

We specify the matrix ordering and a subset of the original checks that attains the optimum in Proposition~\ref{prop:gross-sparse-single-cut-optimum}.
The distance $12$ is the value that Bravyi et al.\ report for the gross code \cite{bravyiHighthresholdLowoverheadFaulttolerant2024}.
The calculation here verifies its stabilizer ranks and the selected cut, not its distance.
For CSS matrices $H_X,H_Z$, the two Pauli types contribute separately, so that
\begin{equation}
\begin{aligned}
\chi_U={}&\rank(H_{X,U})+\rank(H_{X,U^c})-\rank H_X\\
&+\rank(H_{Z,U})+\rank(H_{Z,U^c})-\rank H_Z.
\end{aligned}
\label{eq:css-numerical-cut}
\end{equation}
All ranks are over $\mathbb F_2$.
The orthogonality condition $H_XH_Z^{\mathsf T}=0$ ensures that the $X$ and $Z$ checks commute.

Label the 72 cells of a $12\times6$ torus by $(i,j)$, with $0\le i<12$, $0\le j<6$, and row index $6i+j$.
The two data species have column offsets zero and 72.
Let $x$ and $y$ shift the first and second torus coordinates by one, respectively, and set
\begin{equation}
\begin{aligned}
A&=x^3+y+y^2,& B&=y^3+x+x^2,\\
H_X&=(A\mid B),& H_Z&=(B^{\mathsf T}\mid A^{\mathsf T}).
\end{aligned}
\end{equation}
An $X$ row therefore has the indicated positive shifts and a $Z$ row the negative shifts, with both coordinates reduced modulo their periods.
Each row has weight six, and each data qubit belongs to six of the 144 original checks.
The cut $U$ contains both data species at $i<6$.
Each full matrix has rank 66, and its restriction to either half has rank 48.
The full stabilizer dimension is therefore 132.
Each restricted stabilizer dimension is 96, and each supported stabilizer dimension is $132-96=36$.
Equation~\eqref{eq:css-numerical-cut} then gives $\chi_U=60$.

\begin{proof}[Proof of Proposition~\ref{prop:gross-sparse-single-cut-optimum}]
An $X$ row at $(i,j)$ touches the four first coordinates $i,i+1,i+2,i+3$, while a $Z$ row touches $i-3,i-2,i-1,i$.
The local $X$ rows have $i=0,1,2$ in $U$ and $i=6,7,8$ in $U^c$.
The local $Z$ rows have $i=3,4,5$ in $U$ and $i=9,10,11$ in $U^c$.
There are therefore 18 local rows of each type in each half.
For the $X$ rows with $i=0,1,2$, restriction to the first data species at $i=3,4,5$ retains only their distinct $x^3$ terms, which prove independence.
The same argument applies in the other half, and with $x^{-3}$ on the second species it applies to the local $Z$ rows.
Together these cases prove independence of every listed local set.
Rank-nullity gives 18 supported stabilizers of each type in each half.
These listed local checks therefore span the full supported spaces.

Retain every original row except those with the following zero-based row indices
\begin{equation}
\begin{aligned}
X:\;&64,65,68,69,70,71,\\
Z:\;&46,47,50,51,52,53.
\end{aligned}
\end{equation}
The remaining 66 rows in each Pauli sector are independent, as binary elimination on the displayed matrices can verify.
The exact certificate records all 132 distinct echelon pivots and reconstructs all 144 original rows from the retained rows.
All twelve omitted rows cross the cut.
The retained basis therefore contains 72 local checks and 60 crossing checks.
Deleting rows preserves weight six and incidence at most six.

Measure the local basis checks with local ancillas, and the crossing basis checks with the one-pair gadget of Appendix~\ref{subsec:single-cut-all-branches}.
For any independent syndrome $s$, each compatible raw branch has Kraus operator $2^{-30}\Pi_s$.
There are $2^{60}$ such branches, so their sum is exactly $\Pi_s\rho\Pi_s$.
Every omitted row is a product of retained rows of its own Pauli type.
Its reported bit is therefore the corresponding parity, with no sign offset.
In particular, this within-type reconstruction verifies the signed group relations as well as the binary row relations.
The full 144-bit syndrome and the correct quantum output are thus recovered on every input and reference.
The lower bound of Theorem~\ref{thm:single-cut-resource} proves optimality of the 60 supplied Bell pairs.
\end{proof}

For the comparison in Table~\ref{tab:gross-resource-comparison}, assign the ancilla of each original row to the side selected by its first coordinate $i$.
Enumerating the six support incidences of every row gives 144 crossing check--data interactions.
Of the 144 checks themselves, 72 cross the data cut.
The first count charges a Bell pair for every remote CNOT in the original circuit.
The second charges one pair for each remote parity measured directly.
Both circuits, and the sparse 60-pair circuit, report all original check outcomes and preserve the same syndrome-sector state.
The certificate is preserved with the numerical audit.
These finite matrix identities are independent of any search over other cuts.

\subsection{Matrix conventions for the further codes}
\label{subsec:further-code-details}

This subsection records the matrices behind Tables~\ref{tab:mitten-thumb} and~\ref{tab:three-circuits-four-codes}.
All ranks are over $\mathbb F_2$.
Each rank was computed with two independent elimination routines, one on packed integer rows with highest-bit pivots and one on dense binary arrays with leftmost-column pivots.
For each CSS code, Equation~\eqref{eq:css-numerical-cut} evaluates $\chi_U$ from the restricted ranks of $H_X$ and $H_Z$.

\paragraph{Rotated surface code.}
The data qubits form a $12\times12$ array with $q(r,c)=12r+c$.
Check centers sit at the positions $(i,j)$ with $0\le i,j\le12$ and act on the data qubits at the adjacent array positions.
A center carries an $X$ check if $i+j$ is even and a $Z$ check otherwise.
Boundary $X$ checks appear only on the top and bottom edges, and boundary $Z$ checks only on the left and right edges.
This gives $71$ $X$ checks and $72$ $Z$ checks, all independent, so the stabilizer dimension is $143$.
The distance $12$ of the constructed matrices was certified by a shortest-path and boundary-parity argument on the syndrome graph of each sector.
The cut $U$ contains the rows $r<6$.
Both restrictions have rank $77$, and $\chi_U=77+77-143=11$.
The eleven crossing checks are the five $X$ checks and six $Z$ checks centered on the line $i=6$, and each has two qubits on each side.

\paragraph{Two-gross code.}
With the conventions of Appendix~\ref{subsec:finite-code-calculations}, the two-gross code lives on a $12\times12$ torus with
\begin{equation}
A=x^3+y^2+y^7,\qquad B=y^3+x+x^2,\qquad H_X=(A\mid B),\qquad H_Z=(B^{\mathsf T}\mid A^{\mathsf T}).
\end{equation}
Each matrix has rank $138$, so the stabilizer dimension is $276$.
The cut $U$ contains both species at $i<6$.
Each restriction has rank $204$, hence $\chi_U=204+204-276=132$.
Of the $288$ listed checks, $144$ cross the cut.
Placing each ancilla on the side of its row coordinate, which is also the side holding the larger part of its support, gives $288$ crossing check--data incidences.

\paragraph{Mitten codes.}
The matrices are the exact files published by the authors for the eight instances, read at a pinned revision of their repository \cite{bhardwajHighrateQLDPCProcessors2026}.
Each $H_X$ and $H_Z$ has $2g$ rows of weight nine, and the consecutive column blocks of $g$ qubits are the blocks $D_1,\ldots,D_5$ of their construction.
In every instance both matrices have full row rank $2g$, and every row has exactly three entries in $D_5$.
The restriction of each matrix to $D_5$ has rank $g$, and its restriction to the fingers has rank $2g$.
Equation~\eqref{eq:css-numerical-cut} gives $\chi_{D_5}=(g+2g-2g)+(g+2g-2g)=2g$.
For the $[[150,30,10]]$ instance, the ten cuts that keep two blocks together have the following costs.
$D_1D_2$ costs $84$, $D_3D_4$ costs $88$, and $D_1D_3$ and $D_2D_4$ cost $90$.
$D_3D_5$ and $D_4D_5$ cost $116$, and $D_1D_4$, $D_1D_5$, $D_2D_3$, and $D_2D_5$ cost $118$.
The listed crossing count is $90$ for the first four of these cuts and $120$ for the other six.
The weight statement in Section~\ref{subsec:mitten-thumb} comes from an exhaustive enumeration of the $2^{30}-1$ nonzero elements of each $30$-dimensional finger-supported sector, which found minimum weight $28$ in both sectors.
The support of a CSS stabilizer is the union of the supports of its $X$ and $Z$ parts, so no mixed stabilizer supported on the fingers is lighter.

\paragraph{Balanced cuts.}
The costs above are for specified cuts.
A simulated-annealing search over balanced cuts found $\chi_U=11$ for the surface code, $\chi_U=56$ for a $72/72$ cut of the gross code, and $\chi_U=76$ for a $75/75$ cut of the mitten code.
For the two-gross code it found no $144/144$ cut below $132$.
These values are attained upper bounds on the minimum over balanced cuts and are not certified minima.

\subsection{The projector resource on a tree}
\label{app:network-tree-resource}

The tree theorem follows by turning exact extraction into preparation of one multipartite pure state.
Each network vertex owns both copies of every data qubit assigned to it.
All Bell measurements and Pauli corrections in the following proof are therefore local to vertices.
For the reverse reduction, from extraction to resource preparation, both copies must be corrected, because correcting only the output copy would leave a syndrome-dependent state.

\begin{proof}[Proof of Proposition~\ref{prop:network-tree-resource}]
Put $d=2^n$, $q=2^k$, and use $|A\rangle\!\rangle=\sum_{ij}A_{ij}|i\rangle_O|j\rangle_R$.
Given $|C\rangle$, Bell-measure the actual input against the reference-copy qubits locally.
With the bra convention $\langle B_P|=\langle\Phi_d|(P\otimes I)$ for a tensor-product Pauli $P$, the resulting Kraus operator is $\Pi_0P/\sqrt{dq}$.
Correcting the output by $P^\dagger$ gives $\Pi_{s(P)}/\sqrt{dq}$, where $s(P)$ is the commutation syndrome of $P$.
The symplectic syndrome map has rank $n-k$, so every syndrome has $2^{2n-(n-k)}=dq$ Pauli outcomes.
Coarse-graining therefore gives exactly $\rho\mapsto\Pi_s\rho\Pi_s$, including arbitrary references.
All corrections are local to the modules.
Local resets and SWAPs return the output to the original data sites.

Conversely, prepare input-reference Bell pairs within each module and then run the extraction protocol.
On outcome $s$, the normalized output is $|\Pi_s\rangle\!\rangle/\sqrt q$.
Choose a tensor Pauli $Q_s$ of syndrome $s$.
Acting on both copies by $Q_s\otimes Q_s^*$ gives $|Q_s\Pi_sQ_s^\dagger\rangle\!\rangle/\sqrt q=|C\rangle$ deterministically.
The reference copies used in this reduction are local work registers.

By Lemma~\ref{lem:projector-choi-schmidt}, the Schmidt rank of $|C\rangle$ across tree edge $e$ is $2^{\chi_{U_e}}$.
The exact pure-state tree-construction theorem of Yamasaki, Soeda, and Murao supplies all these edge dimensions simultaneously \cite[Theorem~5]{yamasakiGraphAssociatedEntanglement2017}.
Equivalently, prepare the state at a root and then recursively teleport each child subtree's compressed Schmidt-support register, using $\chi_{U_e}$ Bell pairs on its edge.
Combining this preparation with the local Bell-measurement protocol proves sufficiency.
For necessity, group each side of an edge into one party and apply Theorem~\ref{thm:single-cut-resource} to that edge.
\end{proof}

\section{Annular constructions}
\label{app:annular-construction}

The annular examples separate the expansion of a chosen check list from the intrinsic correlations of its stabilizer space.
This appendix proves Theorem~\ref{thm:annular-counterexample}, Corollary~\ref{prop:many-annuli}, and their consequence for the stated local-expander bound of DBT.
We construct a bounded-degree expander and transfer its expansion to the qubits by grouping them into small disjoint supports.
We then apply an invertible recombination of surface-code checks, and the extraction circuit still measures the original local checks after this recombination.
The last subsection gives elementary versions of the same obstruction and identifies the relevant proof steps in the cited DBT version.

\subsection{Expanders and clique blow-ups}
\label{subsec:expander-supply}
\label{subsec:clique-blowup}

\paragraph{A bounded-degree expander on every vertex set.}
For a graph $E$ and a vertex set $S$, write $N_E(S)$ for the set of vertices outside $S$ that have a neighbor in $S$.

\begin{lemma}[Expander supply]
\label{lem:expander-supply}
For every integer $m\ge2$ there is a simple graph $E_m$ on $m$ vertices with maximum degree at most $32$ such that $|N_{E_m}(S)|\ge|S|/4$ for every vertex set $S$ with $0<|S|\le m/2$.
In particular, at least $|S|/4$ edges of $E_m$ leave every such $S$.
\end{lemma}

\begin{proof}
Let $\pi_1,\ldots,\pi_{16}$ be independent uniformly random permutations of $[m]$.
Let $G$ be the simple graph with an edge $\{v,\pi_j(v)\}$ for every $v$ and $j$ with $\pi_j(v)\ne v$.
Each permutation contributes at most two edges at every vertex, so the maximum degree of $G$ is at most $32$.
To bound the probability of failure, suppose that a set $S$ of size $s\le m/2$ has $|N_G(S)|<s/4$.
Then $S\cup N_G(S)$ has at most $t=\lfloor5s/4\rfloor$ elements.
Some set $T$ of size $t$ therefore contains $\pi_j(S)$ for every $j$.
For fixed $S$ and $T$, a single permutation maps $S$ into $T$ with probability
\begin{equation}
\frac{t(t-1)\cdots(t-s+1)}{m(m-1)\cdots(m-s+1)}\le\left(\frac tm\right)^s.
\end{equation}
Since $s\le t\le5s/4$ and $em/s>1$, the binomial bounds $\binom ms\le(em/s)^s$ and $\binom mt\le(em/s)^{5s/4}$ hold.
A union bound over $S$ and $T$ then bounds the probability that some such $S$ exists by
\begin{equation}
\begin{aligned}
&\sum_{1\le s\le m/2}\binom ms\binom mt\left(\frac tm\right)^{16s}\\
&\quad\le\sum_{1\le s\le m/2}\left[e^{9/4}\left(\frac54\right)^{16}\left(\frac sm\right)^{55/4}\right]^s
\le\sum_{s\ge1}c^s,
\end{aligned}
\end{equation}
with $c=(2e)^{9/4}(5/8)^{16}$, where the last step uses $s/m\le1/2$.
Because $e<3$, we have $(2e)^{9/4}<36\sqrt6<90$.
The inequality $1440\cdot5^{16}<8^{16}$ then gives $c<1/16$, so the sum is less than $1/15$.
Hence some choice of permutations yields a graph with the stated property.
For the edge count, every vertex of $N_G(S)$ has an edge to $S$, and distinct outside vertices give distinct edges.
\end{proof}

For definiteness, let $E_m$ be the graph obtained from the first $16$-tuple of permutations, in lexicographic order, for which the neighbor bound holds.
The lemma guarantees that this finite search terminates.
The search serves only as a definition, not as an efficient algorithm.
Any explicit bounded-degree family with the same expansion property would serve equally well.

\paragraph{Expansion of clique blow-ups.}
We next transfer expansion from the graph of blocks to arbitrary subsets of vertices within those blocks.
\begin{lemma}[Clique blow-up]
\label{lem:clique-blowup}
Let $E$ be a graph with $|\partial_EW|\ge h\min\{|W|,|V(E)\setminus W|\}$ for every $W\subseteq V(E)$, where $h>0$.
Replace each vertex $v$ by a nonempty block $B_v$ of at most four new vertices, and let $H$ be any graph on $\bigsqcup_vB_v$ that contains a clique on each block and every edge between $B_v$ and $B_w$ whenever $vw$ is an edge of $E$.
Then every $A\subseteq V(H)$ with $0<|A|\le|V(H)|/2$ satisfies
\begin{equation}
|\partial_HA|\ge\frac{h}{4+2h}|A|.
\label{eq:clique-blowup}
\end{equation}
\end{lemma}

\begin{proof}
Let $\beta=|\partial_HA|$, and let $t$ be the number of blocks that meet both $A$ and its complement.
The clique on each such block contains an edge leaving $A$, so $\beta\ge t$.
Now let $W$ be the set of vertices $v$ with $|A\cap B_v|>|B_v|/2$, together with those with $|A\cap B_v|=|B_v|/2$ according to any fixed rule.
The union $B_W$ of the blocks in $W$ differs from $A$ only within the $t$ split blocks.
There it differs by the minority part of at most two vertices, so $|A\triangle B_W|\le2t$.
Both $A$ and its complement contain at least $|A|$ vertices, so $B_W$ and its complement both contain at least $|A|-2t$ vertices.
Since blocks have at most four vertices,
\begin{equation}
\min\{|W|,|V(E)\setminus W|\}\ge\frac{|A|-2t}{4}.
\end{equation}
For every edge $vw$ of $E$ with $v\in W$ and $w\notin W$, the block $B_v$ contains a vertex of $A$ and the block $B_w$ contains a vertex outside $A$.
The graph $H$ joins these two vertices, and distinct edges of $E$ give distinct edges of $H$ because they join different pairs of blocks.
Hence
\begin{equation}
\beta\ge h\,\frac{|A|-2t}{4}\ge h\,\frac{|A|-2\beta}{4},
\end{equation}
which rearranges to Eq.~\eqref{eq:clique-blowup}.
\end{proof}

The lemma applies to every vertex set $A$ of $H$, not only to unions of blocks, and it uses only edges that $H$ actually contains.
For $h=1/4$, the constant it gives is $1/18$.

\subsection{The annulus and its expanding presentation}
\label{subsec:annulus-code}

\paragraph{The surface code on a square annulus.}
Fix $L\ge1$ and take the unit square faces with lower-left corners
\begin{equation}
(i,j)\in\{0,\ldots,3L-1\}^2\setminus\{L,\ldots,2L-1\}^2,
\end{equation}
together with exactly the edges and vertices incident to them.
This square annulus has outer side $3L$, inner side $L$, and width $L$.
Direct counting gives $8L^2$ faces and $(3L+1)^2-(L-1)^2=8L^2+8L$ vertices.
The annulus has Euler characteristic zero, so it has $n_L=16L^2+8L$ edges, which carry the data qubits.
Let $A_v$ be the product of $X$ on the edges at a vertex $v$.
Let $B_f$ be the product of $Z$ on the boundary of a face $f$.
These operators commute, since a vertex meets a face in zero or two edges.
The graph is connected, so a set of vertex checks multiplies to the identity only if it is empty or contains every vertex.
The vertex checks therefore satisfy exactly one relation.
The face checks are independent, because a nonempty set of faces of the annulus has a nonempty boundary.
Consider the check $A_{w_0}$ at the white vertex $w_0=(1,0)$, where white means that $i+j$ is odd.
Deleting it therefore leaves an independent list of $n_L-1$ commuting checks with positive signs, so the code encodes one logical qubit.

To verify that the distance is $L+1$, consider first a $Z$-type operator that commutes with every vertex check.
It has even degree at every vertex, so its support is an edge-disjoint union of simple cycles.
A simple cycle that does not surround the removed block bounds a union of retained faces and is therefore a product of face checks.
Hence a nontrivial logical operator contains a simple cycle that surrounds the removed block.
The bounding box of such a cycle contains the removed block, so its width and height are at least $L$.
A closed lattice curve is at least twice as long as the sum of the width and height of its bounding box.
The length of the cycle is therefore at least $4L$, which the inner boundary attains.

For the other Pauli type, an $X$-type operator that commutes with every face check has even overlap with every face.
Form the dual graph whose vertices are the retained faces together with one vertex for the removed block and one for the exterior.
Such an operator is then an even-degree subgraph of this dual graph at the face vertices.
The inner boundary cycle is a nontrivial $Z$-type logical operator.
A straight radial path of $L+1$ edges from the removed block to the exterior commutes with every face check and overlaps that cycle once.
The parity of the overlap with the inner boundary cycle vanishes on vertex checks and therefore detects the single nontrivial $X$-type class.
The overlap with the inner boundary equals the degree of the dual vertex of the removed block.
A nontrivial class then has odd degree at both extra vertices, and so contains a dual path between them.
To bound the length of such a path, assign to each retained face its ring index.
The ring index is the smallest $\rho$ such that the face lies in the square of side $L+2\rho$ centered on the removed block.
Adjacent faces have ring indices that differ by at most one.
Only faces of index one are adjacent to the removed block, and only faces of index $L$ are adjacent to the exterior.
A dual path between the two extra vertices therefore visits at least $L$ faces and crosses at least $L+1$ primal edges.
The radial path attains $L+1$.
A general logical operator contains the support of a nontrivial $X$-type or $Z$-type part, so the distance is $\min\{4L,L+1\}=L+1$.

\paragraph{An expanding presentation.}
Color a vertex black when $i+j$ is even, so that every edge joins vertices of opposite colors.
The black checks therefore have pairwise disjoint supports that together contain every edge exactly once.
Because every vertex of the annulus has degree between two and four, every black support $B_v$ is nonempty with at most four edges.
The full square of $(3L+1)^2$ vertices and the removed interior square of $(L-1)^2$ vertices have the same difference between black and white vertices.
There are therefore $m=4L(L+1)\ge8$ black vertices.

Order the black vertices lexicographically and take the expander $E_m$ of Lemma~\ref{lem:expander-supply} on them.
Replace every black check by
\begin{equation}
A'_j=A_j\prod_{i<j:\,ij\in E_m}A_i,
\label{eq:annulus-transformed-check}
\end{equation}
keep the white checks other than $A_{w_0}$, and keep all face checks.
On the black coordinates, this change of basis is a unit lower-triangular binary matrix $T=I+N_T$.
The identity $N_T^m=0$ then gives $T^{-1}=I+N_T+\cdots+N_T^{m-1}$ over $\mathbb F_2$.
The new list $\mathcal G_L$ is therefore an independent generating list of the same signed stabilizer group.

No support cancellation occurs in Eq.~\eqref{eq:annulus-transformed-check}, because $A'_j$ is a product of distinct black checks with disjoint supports.
Its support is therefore the disjoint union of theirs.
Each $A'_j$ is a product of at most $33$ black checks, which gives weight at most $4\cdot33=132$.
An edge in the support of $A_i$ lies in $A'_i$ and in at most $32$ further products $A'_j$.
It also lies in at most one retained white check and two face checks.
Every qubit therefore lies in at most $36$ generators.
This proves part~\ref{it:annulus-ldpc} of Theorem~\ref{thm:annular-counterexample}.

Turning to the expansion, let $H$ be the contracted Tanner graph of $\mathcal G_L$ on all $n_L$ edges.
Every block $B_j$ is a clique of $H$, because $A'_j$ contains it.
For an edge $ij$ of $E_m$ with $i<j$, the generator $A'_j$ contains both $B_i$ and $B_j$, so $H$ contains every edge between them.
The expander satisfies $|\partial_{E_m}W|\ge\min\{|W|,m-|W|\}/4$ for every $W$, because $\partial W=\partial(V\setminus W)$ and Lemma~\ref{lem:expander-supply} applies to the smaller side.
Lemma~\ref{lem:clique-blowup} with $h=1/4$ therefore gives $|\partial_HA|\ge|A|/18$ for every $A$ with $0<|A|\le n_L/2$.
This proves part~\ref{it:annulus-expansion}.

\subsection{Nearest-neighbor extraction}
\label{subsec:annulus-circuit}

\paragraph{The circuit.}
Use the square grid $\{0,\ldots,6L\}^2$ with $N_L=(6L+1)^2$ sites.
A vertex $(i,j)$ of the annulus has its $X$-check ancilla at $(2i,2j)$.
A horizontal edge from $(i,j)$ to $(i+1,j)$ has its data qubit at $(2i+1,2j)$.
A vertical edge from $(i,j)$ to $(i,j+1)$ has its data qubit at $(2i,2j+1)$.
A face with lower-left corner $(i,j)$ has its $Z$-check ancilla at $(2i+1,2j+1)$.
The remaining sites hold idle ancillas in $|0\rangle$, and no ancilla is placed for the deleted check $A_{w_0}$.
Every data qubit keeps its site throughout the circuit.

To implement the new presentation, the circuit measures the original checks.
It begins with one layer that prepares the vertex ancillas in $|+\rangle$ and the face ancillas in $|0\rangle$.
Four CNOT layers, one for each compass direction, apply CNOT gates from each vertex ancilla to the data qubit in that direction.
In each layer the direction defines an injective map from participating ancillas to data sites, so the gates act on disjoint pairs of nearest neighbors.
The vertex ancillas are then measured in the $X$ basis.
A classically controlled $Z$ returns each ancilla with outcome one to $|+\rangle$.
Four further CNOT layers apply CNOT gates from the data qubits to the face ancillas in the four directions.
The face ancillas are measured in the $Z$ basis.
A classically controlled $X$ returns each ancilla with outcome one to $|0\rangle$.
In total, the circuit has one preparation layer, eight CNOT layers, two measurement layers, and two layers of classically controlled Pauli operations, so its quantum depth is at most $13$.
The depth is at most $14$ if the computation of the classical output is counted as a separate layer.
The circuit uses no measurement of individual data qubits, no reset of data, no long-range gate, and no pre-shared entanglement.

To verify the instrument, observe that the CNOT gates with a common vertex ancilla as control apply the controlled operator $|0\rangle\langle0|\otimes I+|1\rangle\langle1|\otimes A_v$.
All these gates commute because they apply $X$ to their targets.
Measuring the ancilla in the $X$ basis applies the Kraus operator $(I+(-1)^{t_v}A_v)/2$ to the data.
The face ancillas similarly apply $(I+(-1)^{t_f}B_f)/2$.
All these projectors commute, so a run with raw outcome vector $t$ applies the joint projector $\Pi^{\rm old}_t$ of the original independent list.
In the projector identity below, $T$ is extended by the identity on the unchanged white-vertex and face coordinates.
With this extension, the circuit reports the white and face bits unchanged and the black bits $s=Tt$, which are parities of raw outcomes.
Since $A'_j$ is the positive product of the checks in row $j$ of $T$,
\begin{equation}
\Pi^{\rm new}_s=\Pi^{\rm old}_{T^{-1}s},
\end{equation}
so every reported branch applies $\rho\mapsto\Pi^{\rm new}_s\rho\Pi^{\rm new}_s$, also on inputs entangled with references.
The ancillas end in $|0\rangle$ or $|+\rangle$, as required by Sections~2.1--2.2 of DBT.
The reported parities of raw outcomes may overlap, which those sections allow \cite{delfosseBoundsStabilizerMeasurement2021}.
This proves part~\ref{it:annulus-circuit} and completes the proof of Theorem~\ref{thm:annular-counterexample}.

\subsection{Several annuli}
\label{subsec:many-annuli-proof}

\begin{proof}[Proof of Corollary~\ref{prop:many-annuli}]
Take the tensor product of $m$ copies of the annular code.
Omit the white star at $(1,0)$ separately in each copy.
The stabilizer rank is $m(16L^2+8L-1)$, so the code encodes $m$ logical qubits.
A nontrivial logical Pauli has a nontrivial restriction to at least one component.
A logical Pauli on one component attains its component distance, so the tensor-product distance is $L+1$.
The $B=4mL(L+1)$ black-star supports partition all data qubits into blocks of sizes at most four.
Choose the graph $E_B$ of Lemma~\ref{lem:expander-supply} and order all black vertices globally.
Replace their checks by
\begin{equation}
A'_j=A_j\prod_{i<j:\,ij\in E_B}A_i.
\end{equation}
The matrix $T=I+R$ of this transformation is unit lower triangular, so $R^B=0$ and $T^{-1}=\sum_{a=0}^{B-1}R^a$ over $\mathbb F_2$, without any parity condition on $B$.
The untouched white and face checks, together with these transformed black checks, therefore remain an independent generating list of the same signed stabilizer group.
Disjointness of the black supports gives the unchanged weight and incidence bounds $132$ and $36$.
The contracted Tanner graph contains the clique blow-up of the single global expander $E_B$, including its edges between different annuli.
Lemma~\ref{lem:clique-blowup} therefore gives the same expansion constant $1/18$ on every data subset of size at most $n/2$.

For the hardware layout, set $a=6L+1$.
Place copy $j$, for $0\le j<m$, in the tile with offset $a(j\bmod q,\lfloor j/q\rfloor)$.
Within each tile, use the data and ancilla coordinates of Appendix~\ref{subsec:annulus-circuit}.
Fill the remaining sites of $\{0,\ldots,qa-1\}^2$ with idle product ancillas.
The active sites then consist of $n$ data qubits and $n-m$ check ancillas.
All $m$ local-check circuits run simultaneously in the same thirteen quantum layers because their tiles are disjoint.
If $t$ is the complete raw syndrome and $\widehat T$ extends $T$ by the identity on white and face coordinates, report $s=\widehat Tt$.
The identity $\Pi^{\mathrm{new}}_s=\Pi^{\mathrm{old}}_{\widehat T^{-1}s}$ proves exactness of every conditional data map, including on inputs entangled with a reference.
The reported parities may overlap, as allowed in Sections~2.1--2.2 of DBT.
A separately counted final output layer increases the depth bound to fourteen \cite{delfosseBoundsStabilizerMeasurement2021}.
Finally, $q^2\le2m$ and $(6L+1)^2<\tfrac94(16L^2+8L)$ give Eq.~\eqref{eq:many-annuli-hardware}.
Direct substitution gives Eq.~\eqref{eq:many-annuli-rate-distance}.
\end{proof}

\paragraph{Consequences for the local-expander bounds of DBT.}
In DBT, a family of quantum LDPC codes has fixed generating lists whose generator weights and qubit incidences are bounded by a constant.
A local-expander family additionally has contracted Tanner graphs with $h_\varepsilon\ge\alpha$ for fixed $\alpha,\varepsilon>0$.
Here $h_\varepsilon(G)$ minimizes $|\partial A|/|A|$ over nonempty sets with $|A|\le\varepsilon|V(G)|/2$ \cite{delfosseBoundsStabilizerMeasurement2021}.
Parts~\ref{it:annulus-ldpc} and~\ref{it:annulus-expansion} of Theorem~\ref{thm:annular-counterexample} give these properties with constant $132$ for weight, constant $36$ for incidence, $\alpha=1/18$, and $\varepsilon=1$.
A two-dimensional $b$-local circuit in DBT acts on qubits of a subset of $\mathbb Z^2$ with gates between qubits at Chebyshev distance at most $b$.
Corollary~1 concerns circuits on a square patch of $N$ qubits.
The circuit of part~\ref{it:annulus-circuit} is $1$-local on a square patch of $N_L=(6L+1)^2$ qubits.
It uses only operations of the Clifford circuits of their Section~2.1, and it is a Pauli measurement circuit for $\mathcal G_L$ in the sense of their Section~2.2.
Equation~\eqref{eq:annulus-ratio} shows that $n_L/\sqrt{N_L}$ grows linearly in $L$ while the depth stays at most $13$.
This proves Corollary~\ref{cor:dbt-local-expander}.
The circuit acts on $n_L$ data qubits, $8L^2+8L-1$ vertex ancillas, and $8L^2$ face ancillas.
Counting only these qubits gives $N_L=32L^2+16L-1$, a constant fraction of the patch.
The same conclusion follows, in line with the remark of DBT that Corollary~1 also applies to circuits on a constant fraction of a patch.

Proposition~1 of DBT bounds the number $n_{\rm cut}$ of listed generators that cross a vertical slice $Y$ of the patch below by $2h_\varepsilon(\bar T)|\mathbf D\cap Y|/(w(w-1))$.
In this bound, $\bar T$ is the contracted Tanner graph, $\mathbf D$ the set of data qubits, and $w$ the maximum weight.
The counting estimate itself is valid, because every boundary edge of the contracted Tanner graph comes from a crossing generator and each generator induces at most $\binom w2$ edges.
For the annular family, a slice containing between $n_L/2-\sqrt{N_L}$ and $n_L/2$ data qubits therefore has $n_{\rm cut}=\Theta(n_L)$, whereas $O(\sqrt{N_L})$ edges of the circuit cross it.
Theorem~1 of DBT would then give depth $\Omega(L)$.
For sufficiently large $L$, the theorem therefore fails for these independent lists of bounded weight and incidence.

\subsection{Generator changes and the DBT cut bound}
\label{subsec:basis-change-counterexample}

The four-check example of Fig.~\ref{fig:basis-paradox} extends to any number of checks.
The extension shows directly why an arbitrary independent list cannot be counted as independent information across a cut.
Consider
\begin{equation}
S_i=Z_{\ell_i}Z_r,
\qquad i=1,\ldots,t,
\end{equation}
and place all $\ell_i$ on the left of the cut and $r$ on the right, so that each listed check crosses.
Multiplying neighboring checks cancels the factor on the right and gives another basis of the same stabilizer space,
\begin{equation}
T_1=S_1,
\qquad
T_i=S_{i-1}S_i=Z_{\ell_{i-1}}Z_{\ell_i}\quad(i\ge2).
\end{equation}
In this basis only $T_1$ crosses the cut.
The original outcomes are recovered by
\begin{equation}
m(S_i)=m(T_1)\oplus\cdots\oplus m(T_i)
\end{equation}
which is an invertible relabeling.
Since the two bases have the same joint projectors, this relabeling preserves the conditional quantum output as well as the syndrome.

Our circuit uses the fixed hardware path $r-\ell_1-\cdots-\ell_t$ with all $\ell_i$ on the left, and its only crossing hardware edge is $r\ell_1$.
We measure the $T_i$ in the two matchings of the path, using the allowed two-qubit Pauli measurements.
These commuting measurements implement the joint instrument in two quantum layers, although $S_1,\ldots,S_t$ all cross the cut.
Theorem~1 of DBT states $\depth(C)\ge n_{\rm cut}/(64|\partial L|)$, where $n_{\rm cut}$ counts independent listed checks acting on both sides \cite{delfosseBoundsStabilizerMeasurement2021}.
The $S_i$ are independent, so the numerator of this bound is $t$, whereas the connectivity graph of the circuit has $|\partial L|=1$.
The claimed lower bound is therefore $t/64$, which exceeds the actual depth two for $t>128$.

The comparison uses the operational task and the primitives specified in Sections~2.1--2.2 of DBT.
Those sections allow two-qubit Pauli measurements, unrestricted classical communication at no time cost, and syndrome bits formed from possibly overlapping parities of raw outcomes.
They require, for every input, the joint Born probabilities and the conditional state obtained by projecting onto the common eigenspace of the listed checks.
The path construction satisfies this requirement on the state, because the two generator bases have exactly the same joint projectors under the invertible parity relabeling.
It uses no ancillas, so the requirement of DBT to return ancillary qubits to their prescribed product states is vacuous.
Replacing the generators, finally, changes neither the hardware graph nor the data locations.

The gap can be located in the proof of Theorem~1 of DBT, in the lower-information argument of their Lemma~4.
In that argument, the independence of whole crossing operators is used, following their Eq.~(42), to infer the surjectivity of a one-sided commutation map.
In the path construction, the $t$ full operators are independent, but all restrictions to the right are $Z_r$, so the right-side commutation map has rank one.
In contrast, $\chi_{Q\cap L}=t-(t-1)-0=1$ correctly records the single shared constraint.

Independence of the restrictions does not suffice, however, because a second step of the same proof can still fail.
Consider the independent list $(Z_A,Z_AZ_B)$ with $A$ on the left and $B$ on the right.
Measure $Z_A$ and $Z_B$ with separate local ancillas, reporting the raw outcomes $u$ and $u\oplus v$.
The result is an exact Pauli measurement circuit in the sense of DBT, and its reported parities overlap.
The crossing check $Z_AZ_B$ has the nonzero restriction $Z_B$, so the surjectivity step holds.
In the double-measurement circuit of DBT, the second-round outcomes are $u\oplus x_A$ and $v\oplus x_B$, where $x_A$ and $x_B$ are the bit-flip components of the random Pauli error.
Given the first-round record, the left outcomes and errors are therefore independent of the right outcomes and errors.
Hence the conditional mutual information that their Lemma~4 bounds below by $n_{\rm cut}/2=1/2$ vanishes.
The failing step is the removal of conditioning between their Eqs.~(41) and~(42).
That step fails because the right partial parity $v\oplus x_B$ of the second round, together with the first-round record, reveals the right-side error component $x_B$.
In this example $\chi_{Q\cap L}=0$, which is consistent with an exact circuit that uses no crossing edge.

The same example also gives a counterexample to Theorem~1 of DBT with bounded weight and incidence.
For this, take $t$ disjoint copies $(Z_{A_i},Z_{A_i}Z_{B_i})$, measured in this way.
Add one further crossing check $Z_{A_0}Z_{B_0}$, measured by a single two-qubit Pauli measurement on a hardware edge.
The resulting independent list has weight and incidence at most two, $n_{\rm cut}=t+1$, and $|\partial L|=1$.
The circuit has constant depth, whereas Theorem~1 of DBT requires depth at least $(t+1)/64$.

The star list has unbounded incidence, and the bounded-incidence copies of $(Z_A,Z_AZ_B)$ do not expand.
These examples therefore isolate failures of Theorem~1 without, by themselves, contradicting its local-expander corollary.
The annular family supplies the additional bounded-weight, bounded-incidence, and expansion hypotheses, and its constant-depth circuit contradicts Corollary~1 of DBT v1 as stated.
Section~\ref{sec:static-2d-barrier} instead derives geometric bounds from an intrinsic condition and distinguishes them from bounds based on code parameters.

The adaptive proof in Appendix~\ref{subsec:adaptive-branch-proof} avoids the invalid surjectivity inference by evaluating the operator Schmidt rank of each syndrome projector.
It uses the exact instrument to constrain every complete Kraus trajectory and requires neither a Clifford normalization nor a comparison of conditional mutual informations.

\section{The uniform Tanner theorem and geometric bounds}
\label{app:tanner-geometry}

The first part proves that every sufficiently small region of the fixed quantum Tanner family has a positive density of shared stabilizer constraints.
The remaining parts turn such intrinsic estimates into geometric depth bounds and distinguish them from the aggregate consequences of rate and distance.
Throughout, construction graphs describe the code, whereas hardware graphs describe the allowed circuit operations.

\subsection{Local representations in quantum Tanner codes}
\label{subsec:quantum-tanner-target-window}

All numbered results of Leverrier and Z\'emor used in this appendix refer to \href{https://arxiv.org/abs/2202.13641v3}{arXiv:2202.13641v3}, the full version of Ref.~\cite{leverrierQuantumTannerCodes2022}.

This subsection proves Theorem~\ref{thm:quantum-tanner-target-window}.
By the correctable-region identity, it suffices to bound $\dim\cS_U$ from above.
Bounding it requires control of every stabilizer confined to $U$, including products of generators that cancel outside it.

The argument follows the outline in Section~\ref{sec:separator-barriers}.
The first step compares dimensions through coefficient representations.
For sufficiently small $U$, we show that the representatives with few nonzero vertex components form a linear space of the same dimension as the supported stabilizers.
Local distance then bounds the $X$ and $Z$ dimensions separately.
Controlling the overlap of their boundary square sets finally gives the bound on $\dim\cS_U$ without choosing a data region through the hardware geometry.

\paragraph{Construction.}
The construction graphs describe cancellation among $X$ stabilizers and among $Z$ stabilizers.
They are defined on the left-right Cayley complex of the construction.
This complex has edge-label sets $A$ and $B$ of size $\Delta$, and its squares $Q$ are the qubits.
Each square joins two opposite corners in $V_0$ and two in $V_1$.
Within the class $V_i$, we regard the square as an edge of a diagonal multigraph $G_i^\square$, keeping its label and any parallel-edge multiplicity.
Each vertex is then incident to $\Delta^2$ such square edges.
The union graph $G^\cup$ instead uses the $A$- and $B$-edges along the sides of the squares, connecting the two classes.
Both construction graphs describe the overlap of generators and do not specify hardware interactions.

The quadripartite variant of the construction is included through the aggregate classes $V_0=V_{00}\cup V_{11}$ and $V_1=V_{10}\cup V_{01}$.
Their square-incidence counts are unchanged, and their diagonal graphs are bipartite.
Distinct square corners are needed later in both variants, when an overlapping boundary square is associated with a side edge.
In the bipartite construction, total no-conjugacy ensures this distinctness by requiring $ag\ne gb$ for all $a\in A$, $b\in B$, and group elements $g$.
In the quadripartite variant, different part labels ensure that the corners are distinct \cite{leverrierQuantumTannerCodes2022}.

Each vertex carries a local tensor-code pattern on its incident squares, built from binary component codes $C_A,C_B$ of length $\Delta$ and dimensions $r,\Delta-r$.
We fix $\Delta$ and write $D=\Delta^2$.
For $n=|Q|$, counting the two diagonal endpoints of each square gives $M=2n/D$ vertices in either class.
The local coefficient spaces and their common dimension are
\begin{equation}
\begin{aligned}
L_0&=C_A\otimes C_B,
&L_1&=C_A^\perp\otimes C_B^\perp,\\
q&=\dim L_i=r(\Delta-r).&&
\end{aligned}
\end{equation}
These local patterns embed into the square coordinates through
\begin{equation}
\Phi_i:\bigoplus_{v\in V_i}L_i\longrightarrow\mathbb F_2^Q.
\end{equation}
Each map extends a local tensor word by zero off its incident squares and adds the words modulo two.
A tuple of local words is a coefficient representation, and its image is the resulting stabilizer.
A tuple in $\ker\Phi_i$ is a relation, since its contributions cancel on every square.
Throughout, we count the vertices at which a tuple has a nonzero component, separately from the weight of its image on qubits.

In symplectic coordinates, the CSS space is the direct sum of $\operatorname{im}\Phi_0$ for the $X$ components and $\operatorname{im}\Phi_1$ for the $Z$ components.
This identification interchanges the $X/Z$ convention of Leverrier and Z\'emor globally, which preserves supports and ranks \cite{leverrierQuantumTannerCodes2022}.
A stabilizer therefore lies inside $U$ exactly when both its $X$ and $Z$ components do.
We can thus bound their dimensions separately before adding them.

\paragraph{Linear representatives.}
The coefficient maps may have nonzero kernels, so arbitrary representatives need not respect addition.
We transfer the supported stabilizer dimension to coefficient space by comparing two supports.
A nonzero relation has a large support, whereas each chosen representation has a small one.
A large enough gap between the two forces the chosen representation to be unique and linear.

We write $\ell_i=d(L_i)$.
Let $\bar\lambda_i$ bound the diagonal mixing parameters uniformly over the family, and let $\bar\mu$ do the same for the union graph.
These constants are independent of $n$ and are absolute eigenvalue bounds for adjacency matrices, excluding the trivial eigenvalues $\pm d$ of a $d$-regular bipartite graph.
The diagonal degree is $D=\Delta^2$, the union degree is $2\Delta$, and $M=|V_i|$ is the size of an aggregate class.

We use the edge-multiplicity and mixing convention of Lemma~3 of Leverrier and Z\'emor.
Their Lemma~4 bounds the eigenvalues of the diagonal and union graphs.
For a bipartite diagonal graph, we first apply mixing between its two parts and then bound internal edges.
This order fixes the density coefficients below and in Eq.~\eqref{eq:qt-two-boundary-overlap}.
The gap between local distance and diagonal mixing is measured by
\begin{equation}
\sigma_i=\frac{\ell_i-\bar\lambda_i}{D}.
\label{eq:qt-sigma}
\end{equation}
To bound the support of a nonzero relation from below, consider such a relation $a\in\ker\Phi_i$.
Let $R$ be the set of vertices at which $a$ has a nonzero component.
Every nonzero square coordinate at a vertex in $R$ must be canceled at the other endpoint of the square in the same class.
All such coordinates therefore lie on diagonal edges internal to $R$.
Let $e_i(R)$ count these edges with multiplicity.
Local distance gives at least $\ell_i$ coordinates at each vertex in $R$, while mixing bounds the available internal edges, so
\begin{equation}
\ell_i|R|
\le2e_i(R)
\le\frac{D}{M}|R|^2+\bar\lambda_i|R|.
\end{equation}
For a bipartite diagonal graph, let $x,y$ be the sizes of $R$ in its two parts, each of which has $M/2$ vertices.
Bipartite mixing and the AM--GM inequality give
\begin{equation}
\begin{aligned}
2e_i(R)&\le\frac{4D}{M}xy+2\bar\lambda_i\sqrt{xy}\\
&\le\frac{D}{M}(x+y)^2+\bar\lambda_i(x+y),
\end{aligned}
\end{equation}
which proves the same internal-edge inequality with the stated normalization.
Canceling $|R|>0$ then shows that a nonzero relation must have at least $\sigma_iM$ nonzero vertex components.

We next bound the number of local components needed to represent a supported stabilizer.
For this we use Theorem~1 of Leverrier and Z\'emor, which gives the reduction range.
We also use their Remark~15, which gives a weight decrease proportional to $D$ in each step \cite{leverrierQuantumTannerCodes2022}.
To check that reduction applies in both stabilizer sectors, let $\mathcal C_i=T(G_i^\square,L_i^\perp)$ be the classical Tanner code whose restriction at each vertex lies in $L_i^\perp$.
The construction gives
\begin{equation}
\begin{aligned}
\operatorname{im}\Phi_0&=\mathcal C_0^\perp\subseteq\mathcal C_1,\\
\operatorname{im}\Phi_1&=\mathcal C_1^\perp\subseteq\mathcal C_0.
\end{aligned}
\end{equation}
Theorem~1 reduces words of the first image using local $L_0$ words.
Applying the same theorem with the dual component codes and interchanged vertex classes reduces words of the second image using local $L_1$ words.
The second application proceeds as in their proof of Theorem~17.
The full hypotheses on the component codes provide the distance and robustness assumptions for both applications.
For $\operatorname{im}\Phi_i$, there are consequently uniform constants $a_i,\kappa_i>0$ with the following property.
A nonzero image word of weight below $\kappa_in$ can be reduced by more than $a_iD$ using one local codeword at a vertex.
Iterating this reduction represents each $x\in(\operatorname{im}\Phi_i)_U$ by a tuple with at most
\begin{equation}
T_i=\left\lfloor\frac{|U|}{a_iD}\right\rfloor
\end{equation}
nonzero vertex components.

The two estimates combine to give uniqueness and linearity once $3T_i<\sigma_iM$, which we now suppose.
The difference of two such representations of one word has at most $2T_i$ nonzero components and therefore cannot be a nonzero relation.
The representation is thus unique, and we denote it by $s_i(x)$.
The relation $s_i(x)+s_i(y)+s_i(x+y)\in\ker\Phi_i$ is a sum of three tuples, each with at most $T_i$ nonzero vertex components.
It must therefore also vanish, and $s_i$ is linear.
Let $E_i$ be its image.
Because $\Phi_i(s_i(x))=x$, we have $E_i\cong(\operatorname{im}\Phi_i)_U$, so these two spaces have the same dimension.

Linearity also permits a bound on the union of the supports of all elements of $E_i$.
An individual tuple bound would not control this union for an arbitrary set of tuples.
Let $R_i$ be the set of vertices at which some element of $E_i$ has a nonzero component.
At any $v\in R_i$, projection is a nonzero linear map, so at least half of the elements of $E_i$ have a nonzero projection there.
Each element, on the other hand, has at most $T_i$ nonzero vertex components.
Comparing the two expressions for the expected number of these components gives
\begin{equation}
|R_i|\le2T_i.
\label{eq:qt-root-support}
\end{equation}

\paragraph{Local dimension bounds.}
The bounded vertex set $R_i$ restricts the coordinates available to each local component.
The supported stabilizer dimension can therefore be bounded by adding local dimensions.
The requirement that the resulting stabilizer lie in $U$ restricts these coordinates further.
A square incident to only one vertex of $R_i$ cannot have its coefficient canceled, so its coefficient must vanish if that square lies outside $U$.
A component at a vertex of $R_i$ can thus use only internal squares or boundary squares belonging to $U$.

Let $\delta_{G_i^\square}(R_i)$ be the set of labeled diagonal edges with exactly one endpoint in $R_i$.
The allowed boundary squares are then $B_i=U\cap\delta_{G_i^\square}(R_i)$.
We write $m_i=\dim(\operatorname{im}\Phi_i)_U$.
For both values of $i$, the boundary square sets lie in the same qubit set $U$.
If $P_v$ is the projection onto the component at $v$, all these projections together determine the whole representative, which gives
\begin{equation}
E_i\hookrightarrow\bigoplus_{v\in R_i}P_v(E_i).
\label{eq:qt-aggregate-injection}
\end{equation}
Only the collection of projections needs to be injective here, not an individual projection.
Each nonzero projected space is a linear code contained in $L_i$ and supported on its allowed coordinates.
Deleting the coordinates forced to vanish preserves the weights of nonzero words, so its distance is at least $\ell_i$.
The Singleton bound then limits its dimension to the number of available coordinates minus $\ell_i-1$.
Summing over the vertices in $R_i$ counts each internal square twice and each allowed boundary square once, which yields
\begin{equation}
m_i
\le
2e_i(R_i)+|B_i|-(\ell_i-1)|R_i|.
\label{eq:qt-boundary-singleton}
\end{equation}

\paragraph{Boundary overlap.}
For a nonempty data region, the sum of the two supported CSS dimensions must come out smaller than $|U|$ by a fixed fraction.
The local bounds alone do not give this deficit, because their boundary square sets may overlap.
It therefore remains to control the number of squares counted in both $B_0$ and $B_1$.
A square in $B_0\cap B_1$ has one endpoint in each of $R_0$ and $R_1$, and it selects the labeled union-graph edge joining those endpoints.
Total no-conjugacy, or the part labels in the quadripartite construction, ensures that at most $\Delta$ squares select one edge.
Writing $E_{G^\cup}(R_0,R_1)$ for the number of labeled edges between $R_0$ and $R_1$, union-graph mixing gives
\begin{align}
|B_0\cap B_1|
&\le\Delta E_{G^\cup}(R_0,R_1)\nonumber\\
&\le\frac{2D}{M}|R_0||R_1|
+\Delta\bar\mu\sqrt{|R_0||R_1|}.
\label{eq:qt-two-boundary-overlap}
\end{align}
This overlap bound controls the extra contribution that arises when the two Singleton estimates are added.
The union of the boundary square sets lies in $U$, and Eq.~\eqref{eq:qt-two-boundary-overlap} bounds their intersection.
Applying diagonal mixing to the edges internal to each vertex support then bounds the remaining terms in the sum.
With $R=|R_0|+|R_1|$, $m=m_0+m_1$, $\ell_*=\min_i\ell_i$, and $\bar\lambda_*=\max_i\bar\lambda_i$, the combined estimate is
\begin{equation}
m
\le
|U|-
\left[
\ell_*-1-\bar\lambda_*-\frac{\Delta\bar\mu}{2}
-\frac{D}{M}R
\right]R.
\label{eq:qt-closing-deficit}
\end{equation}

A strict bound below $|U|$ requires the expression in square brackets to stay positive.
Its only term that grows with $R$ is $(D/M)R$, which we control by restricting the size of the data region.
To state the restriction, we collect the fixed constants as
\begin{equation}
A_a=a_0^{-1}+a_1^{-1},
\qquad
h=\ell_*-1-\bar\lambda_*-\frac{\Delta\bar\mu}{2}.
\end{equation}
We assume $\sigma_i>0$ and $h>0$.
We then choose the size fraction $\eta$ small enough for the reduction range, the linearity of $s_i$, and the positivity of that expression,
\begin{equation}
0<\eta<
\min\left\{
2\min_i\kappa_i,
\frac43\min_i(a_i\sigma_i),
\frac{2h}{DA_a}
\right\}.
\label{eq:qt-general-eta}
\end{equation}
To obtain the corresponding bound on $\chi_U(\cS)$, we also shrink $\eta$ until $\eta/2$ lies strictly below the relative distance of the fixed family, which ensures correctability.
For every $|U|\le\eta n/2$, the chosen constants then give $3T_i<\sigma_iM$.
Equation~\eqref{eq:qt-root-support} bounds the support-dependent loss by
\begin{equation}
\frac{D}{M}R\le\frac{\eta DA_a}{2}.
\end{equation}
The expression in square brackets in Eq.~\eqref{eq:qt-closing-deficit} is therefore at least $g_\eta$, which gives $m\le|U|-g_\eta R$ with
\begin{equation}
g_\eta=h-\frac{\eta DA_a}{2}>0.
\end{equation}
To finish the dimension bound, we remove its dependence on the number of vertices used by the representatives.
Each component space has dimension at most $q$, so $m\le qR$, and hence $R\ge m/q$ when $m>0$.
Substituting this lower bound for $R$ gives
\begin{equation}
\dim\cS_U=m
\le\frac{q}{q+g_\eta}|U|.
\label{eq:qt-shortening-deficit}
\end{equation}
When $m=0$, the same bound is immediate, including for $U=\varnothing$.
The supported stabilizer dimension is therefore bounded for every $0\le|U|\le\eta n/2$, with no positive lower-size threshold, so the code estimate is available before a hardware partition is chosen.
The interval $\eta n/8<|U|\le\eta n/2$ is needed only for that later choice.

\paragraph{Uniform constants.}
It remains to choose the constants independently of block length within the hypotheses of the construction, retaining a positive gap between distance and mixing.
Let $\delta>0$ lower-bound the relative distances of both component codes and both duals.
Let $\varepsilon\in(0,1/2)$ be the parameter in Theorem~17 of Leverrier and Z\'emor that sets the robustness and distance scales, with its puncturing parameter in the required range.
The constants that specify the size of the data region in the separator argument are defined separately.

Condition~2 of Theorem~17 bounds the component and dual distances, and hence also their tensor-product distances.
Lemma~4 bounds the eigenvalues in the stated adjacency convention, and Theorem~1 gives the reduction radius.
Remark~15 bounds the weight decrease in each reduction step, which yields the constants $a_i$.
These bounds are uniform over the fixed family and give
\begin{equation}
\begin{aligned}
\ell_i&\ge\delta^2D,
&\bar\lambda_i&\le4\Delta,
&\bar\mu&\le4\sqrt\Delta,\\
\kappa_i&=\frac{\delta}{4\Delta^{3/2+\varepsilon}}.&&&&
\end{aligned}
\end{equation}
To bound the remaining gap between distance and mixing, set
\begin{equation}
H=\delta^2\Delta^2-1-4\Delta-2\Delta^{3/2},
\qquad
a_{\min}=\min(a_0,a_1).
\end{equation}
We choose a fixed degree that is admissible for the construction and large enough for $H>0$.
The following size fraction then satisfies the preceding strict inequalities,
\begin{equation}
\begin{split}
\eta_0=\frac12\min\biggl\{&
1,\frac{\delta}{2\Delta^{3/2+\varepsilon}},\\
&\frac{4a_{\min}}3\left(\delta^2-\frac4\Delta\right),
\frac{2H}{D(a_0^{-1}+a_1^{-1})}
\biggr\}.
\end{split}
\label{eq:qt-safe-eta}
\end{equation}
This choice leaves $g_{\eta_0}\ge H/2$, so the supported dimension falls below the size of the data region by a uniform fraction, which we can choose as
\begin{equation}
\tau_0=\frac{H}{2q+H}>0.
\label{eq:qt-safe-tau}
\end{equation}
These formulas depend on the uniform reduction constants $a_i$ of the construction, and numerical values would require quantitative choices of those constants.
The degree is fixed large enough for this positive gap and for all admissible-degree, component-existence, robustness, and reduction hypotheses of the construction.
Floors strengthen the estimates, so no $1/n$ correction is needed.
Theorem~17 of Leverrier and Z\'emor gives $\eta_0 n/2<d$, which allows Lemma~\ref{lem:correctable-cut-entropy} to turn Eq.~\eqref{eq:qt-shortening-deficit} into $\chi_U(\cS)\ge2\tau_0|U|$.
We set $\eta=\eta_0$ and $\tau=\tau_0$.
The entropy identity of Lemma~\ref{lem:correctable-cut-entropy} gives the pure-code-state conclusion as well, which completes the proof of Theorem~\ref{thm:quantum-tanner-target-window}.
For the geometric application, we impose in addition $n\ge\lceil8/\eta_0\rceil$, so that the separator interval is nonempty.
We also impose any eventual threshold of the family.

\subsection{Proof of the intrinsic bound on square patches}
\label{subsec:intrinsic-patch-proof}

To prove Corollary~\ref{cor:intrinsic-patch}, give the patch coordinates $(x,y)$ with $0\le x,y<\sqrt N$.
Let $L_u$ contain the sites with $x\le u$.
Moving from $L_u$ to $L_{u+1}$ adds one column of $\sqrt N$ sites.
The data count rises from zero to $n>\eta n/2$, so some $u$ satisfies
\begin{equation}
\frac{\eta n}{2}-\sqrt N\le|Q\cap L_u|\le\frac{\eta n}{2}.
\end{equation}
Only sites in the last $b$ columns of $L_u$ have neighbors outside it, and each has at most $b(2b+1)$ of them, so $|\partial L_u|\le b^2(2b+1)\sqrt N$.
This count rests on an upper bound on the interaction range and does not require the patch graph to be connected.

Suppose first that $\sqrt N\le\eta n/4$.
Then $U=Q\cap L_u$ satisfies $\eta n/4\le|U|\le\eta n/2$, so $\chi_U\ge\gamma\eta n/4>0$.
If $|\partial L_u|=0$, Theorem~\ref{thm:adaptive-instrument-cut-bound} rules out exact extraction.
Otherwise the same theorem gives
\begin{equation}
\depth(C)\ge\frac{\chi_U}{2|\partial L_u|}\ge\frac{\gamma\eta}{8b^2(2b+1)}\frac{n}{\sqrt N}.
\end{equation}
Suppose instead that $\sqrt N>\eta n/4$, so that $(\eta/4)n/\sqrt N<1$.
Since $\eta n\ge2$, the hypothesis applies to a single data qubit $q$, and $\chi_{\{q\}}\ge\gamma>0$.
The cut around the site of $q$ therefore requires a crossing operation by Theorem~\ref{thm:adaptive-instrument-cut-bound}, so $\depth(C)\ge1>(\eta/4)n/\sqrt N$, and both cases give Eq.~\eqref{eq:intrinsic-patch-depth}.

\subsection{Weighted separators}
\label{subsec:weighted-separators}

Separators locate bottlenecks, namely hardware regions with many data qubits and few edges leaving them.
The argument below applies to any code with the required bound on $\chi_U$, not only to the quantum Tanner family.

The code estimate is indexed by data qubits, whereas a separator acts on all hardware vertices.
We therefore assign weight one to data vertices and weight zero to ancillas.
A separator removes a small set of vertices and leaves two outside sides with no edge between them, and each side may consist of several connected components.
We write this partition as $V(H)\setminus X=A\sqcup B$.

A single small separator, however, does not ensure that the selected data subset lies in the range of the code estimate.
If an outside side still contains too much data, we continue inside that side until its data size enters the required interval.
The separator guarantee must therefore hold on every induced subgraph, so that it remains available at each step.

\begin{definition}[Separator size]
\label{def:weighted-separator-profile}
For a hardware family, let $s_{\hw}(N)$ be a function such that every induced subgraph $H$ with at most $N$ vertices and every nonnegative vertex weight admit a separator $X$ with $|X|\le s_{\hw}(N)$ such that every outside side has at most two thirds of the total weight.
\end{definition}

The recursion will select a data size between $\eta n/8$ and $\eta n/2$.
The upper endpoint allows the code estimate to be applied, and the lower endpoint retains enough data for a bound proportional to $n$.
The hardware argument therefore needs a code estimate only on this interval.
We state this weaker condition separately from the quantum Tanner theorem.

\begin{definition}[Lower bound on $\chi_U(\cS)$]
\label{def:intrinsic-target-window-expansion}
A stabilizer space $\cS$ on $n$ data qubits has parameters $(\alpha,\eta)$, where $\alpha>0$ and $0<\eta\le1$, if
\begin{equation}
\chi_U(\cS)\ge\alpha|U|
\quad\text{whenever}\quad
\frac{\eta n}{8}<|U|\le\frac{\eta n}{2}.
\label{eq:intrinsic-target-window-expansion}
\end{equation}
\end{definition}

A lower bound for all data subsets of size at most $\varepsilon n/2$ implies this condition with $\eta=\min\{\varepsilon,1\}$.
The lower endpoint is needed only by the separator recursion.
For a family with $d\ge\delta n$ and $0<\eta\le\min\{1,\delta\}$, every region in the interval is correctable, so Lemma~\ref{lem:correctable-cut-entropy} expresses the condition as
\begin{equation}
\dim\cS_U
\le
\left(1-\frac{\alpha}{2}\right)|U|
\label{eq:target-window-shortening-rank}
\end{equation}
on that interval.
This formulation separates the role of distance from the additional dimension bound proved by the family theorem.

The recursion selects a region in this interval and controls its boundary through the removed separators.
Comparing the resulting code requirement with the available crossing operations gives the following depth bound.

\begin{theorem}[Depth bound from separators]
\label{thm:weighted-separator-depth}
Let $C$ exactly measure a stabilizer space $\cS$ satisfying Definition~\ref{def:intrinsic-target-window-expansion} with parameters $(\alpha,\eta)$ in the model of Definition~\ref{def:adaptive-local-extraction}.
Suppose its $N$-vertex hardware graph has maximum degree $0<\Delta_{\hw}$ and separator size bound $s_{\hw}(N)$ as in Definition~\ref{def:weighted-separator-profile}.
Put
\begin{equation}
K_\eta=\left\lceil\log_{3/2}\frac{2}{\eta}\right\rceil.
\label{eq:eta-k-definitions}
\end{equation}
If $0<s_{\hw}(N)\le\eta n/4$, then
\begin{equation}
\depth(C)
\ge
\frac{\alpha\eta}{16K_\eta}
\frac{n}{\Delta_{\hw}s_{\hw}(N)}.
\label{eq:weighted-separator-depth-explicit}
\end{equation}
For fixed $\alpha,\eta$ and $s_{\hw}(N)=o(n)$ this is $\Omega(n/[\Delta_{\hw}s_{\hw}(N)])$.
\end{theorem}

\begin{proof}
It suffices to select a region in which the code estimate applies and to bound the edges leaving it.
To select it, assign data weight one and ancilla weight zero, and set $t=\eta n/2$.
If the current induced subgraph contains more than $t$ data qubits, remove a weighted separator.
Then continue into an outside side with more than $t$ data qubits if one exists.
If no such side exists, stop and take the outside side with larger data weight as $L$.

We first verify that the data size of $L$ lies in the required interval.
At the stopping step, neither side exceeds $t$.
Before that step the subgraph still contained more than $t$ data qubits, and the separator removes at most $s_{\hw}(N)$ of them.
The side with larger data weight therefore contains at least half of the remainder.
The two bounds together give
\begin{equation}
\frac{\eta n}{8}<|Q\cap L|\le\frac{\eta n}{2}
\label{eq:recursive-expansion-window}
\end{equation}
where the lower endpoint follows from $(t-s_{\hw}(N))/2\ge\eta n/8$.

To control the boundary of $L$, we first bound the number of removed separators.
Each continuation reduces the data weight by a factor of at most $2/3$, so at most $K_\eta$ disjoint separators are removed.
No edge joins a followed outside side to its discarded counterpart.
Every edge in the original hardware graph that leaves the final side must therefore meet the union of the removed separators.
The degree bound then gives
\begin{equation}
|\partial_{G_{\hw}}L|
\le
\Delta_{\hw}K_\eta s_{\hw}(N).
\label{eq:weighted-separator-boundary}
\end{equation}
The selected data set lies in the interval of Definition~\ref{def:intrinsic-target-window-expansion}, so $\chi_{Q\cap L}(\cS)>\alpha\eta n/8$.
For a positive boundary, dividing this lower bound by the boundary estimate and applying Theorem~\ref{thm:adaptive-instrument-cut-bound} proves Eq.~\eqref{eq:weighted-separator-depth-explicit}.
If the boundary is zero, Theorem~\ref{thm:adaptive-instrument-cut-bound} instead shows that exact measurement is infeasible, because $\chi_{Q\cap L}(\cS)$ is positive.
\end{proof}

Hardware that excludes a fixed minor, including planar hardware, has small separators.
Together with a bound on the vertex degree, these separators restrict the number of crossing hardware edges.
Theorem~\ref{thm:weighted-separator-depth} then gives an explicit depth bound.

\begin{corollary}[Conditional fixed-minor-free barrier]
\label{cor:fixed-minor-free-depth}
Under the hypotheses of Theorem~\ref{thm:weighted-separator-depth}, suppose $G_{\hw}$ excludes an $h$-vertex minor and $h^{3/2}\sqrt N\le\eta n/4$.
Then
\begin{equation}
\depth(C)
\ge
\frac{\alpha\eta}{16K_\eta}
\frac{n}{\Delta_{\hw}h^{3/2}\sqrt N}.
\label{eq:fixed-minor-free-depth-explicit}
\end{equation}
For fixed $h,\alpha,\eta$, bounded degree, and $N=\Theta(n)$, this is $\Omega(\sqrt n)$.
\end{corollary}

\begin{proof}
Theorem~\ref{thm:weighted-separator-depth} requires the weighted separator bound on every induced subgraph.
Since each such subgraph still excludes the same minor, the weighted theorem of Alon, Seymour, and Thomas~\cite{alonSeparatorTheoremGraphs1990} gives $s_{\hw}(N)=h^{3/2}\sqrt N$ throughout the recursion.
Theorem~\ref{thm:weighted-separator-depth} then yields the claim.
\end{proof}

For the sufficiently large members of the family in Theorem~\ref{thm:quantum-tanner-target-window}, the code condition holds with $\alpha=2\tau$, which gives the following consequence.

\begin{corollary}[Quantum Tanner fixed-minor-free hardware barrier]
\label{cor:quantum-tanner-fixed-minor-free-depth}
For the family in Theorem~\ref{thm:quantum-tanner-target-window}, every exact hardware-local adaptive measurement on an $N$-vertex hardware graph excluding an $h$-vertex minor obeys
\begin{equation}
\depth(C_j)
\ge
\frac{\tau\eta}{8K_\eta}
\frac{n_j}{\Delta_{\hw}h^{3/2}\sqrt N},
\label{eq:quantum-tanner-fixed-minor-free-depth}
\end{equation}
provided $h^{3/2}\sqrt N\le\eta n_j/4$.
A selected zero-boundary cut makes exact measurement infeasible instead.
For fixed $h,\eta,\tau$, bounded hardware degree, and $N=\Theta(n_j)$, this is $\Omega(\sqrt{n_j})$.
\end{corollary}

\begin{proof}
Theorem~\ref{thm:quantum-tanner-target-window} gives the code estimate with $\alpha=2\tau$ throughout the selection interval.
Corollary~\ref{cor:fixed-minor-free-depth} therefore applies to every region that the recursion may select.
\end{proof}

\subsection{The existing partition method with data-weighted separators}
\label{subsec:bfs-weighted-partitions}

We derive here the depth--space scaling directly from the partition framework of BFS, rather than comparing only with their stated Euclidean formula \cite{baspinLowerBoundOverhead2023}.
Let $C$ implement exact syndrome extraction of an $[[n,k,d]]$ stabilizer code in their adaptive local-channel model, with product ancillas.
Partition all hardware sites into sets $\Gamma_i$ with $|\Lambda_i|<d$, where $\Lambda_i=\Gamma_i\cap Q$.
For this argument, $\partial_v\Gamma_i$ includes the vertices on both sides that are incident to an edge crossing $\Gamma_i$.
We write $B_v=\sum_i|\partial_v\Gamma_i|$.
For $B_v>0$, the partition argument gives the conservative bound
\begin{equation}
\depth(C)\ge\frac{k}{3B_v}.
\label{eq:bfs-partition-consequence}
\end{equation}

We give a self-contained derivation for the product input needed here.
Its rank argument in fact improves the denominator from $3B_v$ to $B_v$.
We retain the weaker display, however, for comparison with Lemmas~14 and~22 of BFS as stated.
The derivation first bounds the entanglement created from a product input.
It then shows that the corrected output carries total entanglement at least $k$ across the correctable parts.
Finally, it constructs a suitable partition from separators.

\paragraph{Growth from a product input.}
We measure bipartite entanglement by the relative entropy of entanglement $E_R$, the minimum quantum relative entropy to a separable state.
That is, $E_R(A:B)_\rho=\inf_{\zeta\,\mathrm{separable}}D(\rho\Vert\zeta)$, where $D(\rho\Vert\zeta)=\operatorname{Tr}\rho(\log_2\rho-\log_2\zeta)$ and separable states are mixtures of product states on $A$ and $B$.
Fix a hardware cut $L$, and include any private references or purifiers on their respective sides.
Suppose the initial state is a product across this cut.
The layers in Definition~12 of BFS \cite{baspinLowerBoundOverhead2023} admit Kraus decompositions into tensor factors on disjoint hardware cliques and a separate factor on the workspace.
We assign the workspace to the complementary side.
On a fixed Kraus trajectory, a clique factor acting on $a$ qubits in $L$ and $b$ outside it has operator Schmidt rank at most $2^{2\min(a,b)}\le2^{a+b}$.
A factor confined to one side has rank one.
Every vertex of a crossing clique belongs to the vertex boundary, and distinct clique factors have disjoint supports.
With $v_L=|\partial_v L|$, the product of their ranks in one layer is therefore at most $2^{v_L}$.
After $t$ layers, each pure refined trajectory has Schmidt rank at most $2^{t v_L}$ and entanglement entropy at most $t v_L$.
The same estimate holds for mixed product inputs by local purification.
Convexity of the relative entropy of entanglement, followed by its monotonicity under separable corrections and local discards, gives
\begin{equation}
E_R(L:L^c)_{\rm out}\le t v_L.
\label{eq:product-input-boundary-growth}
\end{equation}
The estimate also applies to Definition~\ref{def:adaptive-local-extraction}, where each crossing clique is a two-qubit edge.
It is a bound for a product input, so no assertion about the entanglement increase of an arbitrary incoming state is needed.

\paragraph{Correctable partitions.}
Use a product data input.
Append a syndrome-dependent Pauli correction that maps each syndrome eigenspace into the code space.
For any partition part, this correction is separable across $\Gamma_i:(\Gamma_i^cX)$, with the syndrome register $X$ assigned to the complementary side, because its Kraus operators have the form $P_{s,\Gamma_i}\otimes(P_{s,\Gamma_i^c}\otimes|0\rangle\!\langle s|_X)$.
The correction therefore cannot increase the relative entropy of entanglement across that cut.
After tracing out ancillary registers, let $\rho$ be the corrected data state and $\omega$ the maximally mixed code state.
Correctability gives $\rho_{\Lambda_i}=\omega_{\Lambda_i}$ and decouples $\Lambda_i$ from a purifying reference of $\rho$, so $S(\rho_{\Lambda_i^c})-S(\rho)=S(\omega_{\Lambda_i})$.
The coherent-information lower bound on the relative entropy of entanglement then gives $E_R(\Lambda_i:\Lambda_i^c)_\rho\ge S(\omega_{\Lambda_i})$.
For completeness, this lower bound follows from $\zeta_{AB}\le I_A\otimes\zeta_B$ for separable $\zeta$.
Operator monotonicity of the logarithm implies $D(\rho\Vert\zeta)\ge S(\rho_B)-S(\rho)+D(\rho_B\Vert\zeta_B)$, after which we minimize over $\zeta$.
Singular states are handled by a full-rank limit.
Summing and using entropy subadditivity for $\omega$ gives
\begin{equation}
\sum_i E_R(\Lambda_i:\Lambda_i^c)_\rho
\ge\sum_i S(\omega_{\Lambda_i})
\ge S(\omega)=k.
\end{equation}
This inequality is the correctable-partition step of Lemma~22 of BFS.
Equation~\eqref{eq:product-input-boundary-growth} bounds each hardware term by $\depth(C)|\partial_v\Gamma_i|$, so $k\le\depth(C)B_v$, which implies Eq.~\eqref{eq:bfs-partition-consequence}.
The appended correction costs no quantum depth in this entanglement comparison.

It remains to construct a suitable partition from separators.
Suppose the hardware has maximum degree $\Delta_{\hw}$ and the hereditary weighted-separator bound $s_{\hw}(N)$ of Definition~\ref{def:weighted-separator-profile}.
For $d>1$, recursively separate every part containing more than $d-1$ data qubits.
Data sites receive weight one and all other sites weight zero.
Every child contains at most two thirds of the data of its parent, so there are at most
\begin{equation}
J=\left\lceil\log_{3/2}\frac{n}{d-1}\right\rceil
\end{equation}
levels and at most $2^J-1$ separator sets.
Let $X_{\rm sep}$ be their union.
Take the final outside parts together with each individual separator vertex as the hardware partition, so that each part contains fewer than $d$ data qubits.
Every edge between parts meets $X_{\rm sep}$, and each of its two endpoint vertices can be counted in the boundaries of both parts, so
\begin{equation}
B_v\le4\Delta_{\hw}|X_{\rm sep}|
\le4\Delta_{\hw}(2^J-1)s_{\hw}(N).
\end{equation}
Before dividing by a possibly zero boundary, the conservative bound reads $k\le3\depth(C)B_v$, so a zero boundary with $k>0$ is incompatible with extraction in this partition.
Otherwise,
\begin{equation}
\depth(C)\ge
\frac{k}{12\Delta_{\hw}(2^J-1)s_{\hw}(N)}.
\label{eq:bfs-weighted-separator-consequence}
\end{equation}
For $k,d=\Theta(n)$, $J$ is bounded independently of $n$ and $N$.
Bounded-degree hardware excluding a fixed minor has $s_{\hw}(N)=O(\sqrt N)$, which gives $\Omega(n/\sqrt N)$ directly from the existing framework.
This derived consequence applies to good stabilizer families without a uniform lower bound on $\chi_U$ for every small subset.

\section{Grid compilation}
\label{app:grid-width-compiler}

This appendix proves Theorem~\ref{thm:adaptive-grid-width-upper} by separating the exact measurement task from its geometric implementation.
We first express the instrument as a constant number of unrestricted CNOT matchings.
We then implement each matching either with prepared Bell connections or with a mesh permutation.
The tiled route counts every resource-preparation step, whereas the mesh route covers the small-workspace case, including $N=n$.
Both routes return every data label to its input site.

\subsection{Measuring disjoint checks without check ancillas}

We first construct a constant-depth circuit with unrestricted pair interactions and no check ancillas.
The graph joining checks with overlapping supports has degree at most $w(\delta-1)$.
Its checks can therefore be colored with $c_{\rm chk}=w(\delta-1)+1$ colors, each consisting of disjoint supports.
For one check $S=\epsilon\bigotimes_{j\in A}P_j$, choose a pivot $p\in A$ and a product of one-qubit Clifford gates $B$ that maps every nonidentity factor to $Z$.
With
\begin{equation}
\begin{gathered}
U=\left(\prod_{j\in A\setminus\{p\}}\operatorname{CNOT}_{j\to p}\right)B,\\
U^\dagger Z_pU=\epsilon S,
\end{gathered}
\end{equation}
apply $U$.
Then measure $Z_p$ and apply $U^\dagger$.
If the raw outcome is $t$, the reported bit $s=t\oplus[\epsilon=-1]$ has Kraus operator
\begin{equation}
U^\dagger\frac{I+(-1)^tZ_p}{2}U
=\frac{I+(-1)^sS}{2}.
\end{equation}

For each color, the forward and inverse circuits contain at most $2(w-1)$ CNOT matchings.
The one-qubit basis changes and pivot measurement take at most five additional layers.
To combine the measurement circuits, complete one color before beginning the next.
The commuting check projectors multiply to the joint projector $\Pi_s$, including zero projectors for inconsistent redundant outcomes.
The combined circuit is therefore the exact joint L\"uders instrument on every input and on inputs entangled with an untouched reference.

\subsection{Fixed terminals and batches of routes}

It remains to implement each CNOT matching on the grid.
Set $q=\lfloor L/4\rfloor$ and first suppose $q^2\ge n$.
Partition a $4q\times4q$ subgrid into $q^2$ tiles of size $4\times4$.
Place at most one data qubit per tile, with at most $h=\lceil n/q\rceil$ data terminals in each tile row.
For example, terminal $i$, indexed from zero, can occupy tile $(i\bmod q,\lfloor i/q\rfloor)$.
This placement is fixed for the entire extraction.

For a logical matching, construct the multigraph whose vertices are tile rows and whose edges are matched terminal pairs.
A pair in the same row is allowed as a loop.
Each row has at most $h$ incident endpoints, so a pair conflicts with at most $2h-2$ other pairs by sharing an endpoint row.
Greedy coloring therefore partitions the matching into at most $2h-1$ batches in which different pairs use disjoint sets of rows.

In one batch, assign every pair in different rows a distinct column, there being at most $\lfloor q/2\rfloor$ such pairs.
Route each pair horizontally from its first terminal to its assigned column, vertically to the other terminal's row, and horizontally to the second terminal.
A same-row pair instead uses the direct horizontal path.
There is at most one horizontal segment in each row and one vertical segment in each assigned column, so no tile-grid edge is shared.
At a tile, at most two paths meet, and their local directional ports are disjoint.
The only two-path intersections are a horizontal passage or terminal attachment together with a vertical passage.
The routes are thus edge-disjoint, with a constant amount of switching inside each tile, rather than vertex-disjoint paths on the unexpanded grid.

\subsection{Local Bell preparation and simultaneous fusion}

To implement these routes with local gates, we specify the switching circuit inside each tile.
In each $4\times4$ tile reserve $(0,0)$ for data, and use the five ancilla ports
\begin{equation}
\begin{gathered}
N=(0,2),\quad E=(2,3),\quad S=(3,2),\\
W=(2,0),\quad T=(0,1).
\end{gathered}
\end{equation}
The port $T$ is adjacent to the data, and facing directional ports of adjacent tiles are nearest neighbors.
The ancilla subgraph obtained by deleting $(0,0)$ is connected, and any two named ports can be joined within it by a path of at most four edges.
The local routing pattern specifies a matching of at most two pairs of these ports.

Start the ancillas in $|0\rangle$.
For a prescribed port pair $(a,b)$, apply $H_a$.
Swap the state at $a$ along a local path until it is next to $b$.
Apply CNOT to $b$, and reverse the swaps.
The sequence realizes $\operatorname{CNOT}_{a\to b}H_a$ exactly and acts as the identity on the other tile registers.
The identity action holds even when a temporary SWAP path traverses a register entangled with another tile or with a previously prepared port pair.
The second preparation thus preserves the first pair and every spectator state.
Preparing two disjoint port pairs serially then creates the required product of Bell states, including the crossing pairing, while leaving the data untouched.
All tiles run in parallel, and at most forty layers suffice if each SWAP is decomposed into three CNOTs.
Related constant-depth Bell switching appears in the quadratic-space construction of DBT~\cite{delfosseBoundsStabilizerMeasurement2021}.

To join the local resources into endpoint pairs, Bell-measure the facing ports on every used tile-grid edge.
These measurements have disjoint physical supports, even when the tile-grid edges share a tile.
Three elementary layers therefore suffice for their CNOT, Hadamard, and local readout implementations.
For a route visiting $k$ tiles, contracting its $k$ local Bell pairs with $k-1$ Bell measurements leaves a Bell pair between its two terminal ports.
The Pauli on this pair is determined by the two parities of the Bell outcomes.
Every outcome sequence has probability $4^{-(k-1)}$, and its corrected unnormalized endpoint state is $2^{-(k-1)}|\Phi^+\rangle$ up to a global phase.
The computation is the usual Bell-contraction identity applied simultaneously to disjoint measured registers.
The parity corrections take one additional one-qubit layer, since classical processing is free.

\subsection{Remote gates, the mesh fallback, and total depth}

For a resulting Bell pair $(a,b)$ beside control $c$ and target $t$, apply $\operatorname{CNOT}_{c\to a}$.
Measure $Z_a$ with outcome $m$, and apply $X_b^m$.
Apply $\operatorname{CNOT}_{b\to t}$.
Measure $X_b$ with outcome $k$, and apply $Z_c^k$.
The Kraus operator on the data for each $(m,k)$ is $\tfrac12\operatorname{CNOT}_{c\to t}$.
The sequence implements the desired CNOT in six layers and leaves both data labels at their original sites.

Every pair in the batch runs in parallel, so one batch, including an initial ancilla-reset layer, has depth at most $1+40+3+1+6=51$.
Ancillas are reset and reused between batches, and no Bell state is supplied for free.
The complete matching therefore costs at most $51(2h-1)$ layers.
Every refined remote-gate branch is a scalar multiple of the intended CNOT, so composition preserves the exact joint instrument.
All classical records are finite for finite $n,N$.

The narrow-grid regime $q^2<n$ remains, and there we implement the same ancilla-free logical circuit by ordinary mesh permutations.
Assign the pairs of one logical matching to different adjacent edges of a snake-path matching.
Complete the assignment to a permutation of all $L^2$ tokens.
The row--column--row routing construction implements this permutation in at most $3L$ parallel SWAP layers.
The construction edge-colors the regular bipartite graph of source and destination rows to select intermediate columns, and then uses odd-even transpositions along the rows, columns, and destination rows.
Apply the adjacent CNOT layer and undo the permutation, for at most $18L+1$ layers in the CNOT alphabet.
The mesh route returns all data to their fixed sites and needs no additional check register.
It is the same mesh-permutation ingredient as the known linear-space compiler~\cite{baconSparseQuantumCodes2017b}, applied here to a circuit with intermediate pivot measurements.

To combine the two regimes, observe that $q^2<n$ implies $L<4(\sqrt n+1)\le8\sqrt n$, so $L\le64n/L$.
In the tiled regime, $L<4(q+1)\le8q$, so $h\le8n/L+1$.
Together with the constant number of logical CNOT matchings, these estimates prove Eq.~\eqref{eq:adaptive-grid-width-upper} over the entire range $N\ge n$.
For an $A\times B$ rectangle with $A\le B\le\kappa A$, use $q=\lfloor A/4\rfloor$ in the tiled construction and the full-rectangle row--column--row permutation in the remaining regime.
The latter costs at most $2B+A$ SWAP layers per permutation, and this regime has $AB=O_\kappa(n)$.
The two estimates give the claimed bounded-aspect-ratio extension.

\section{Weighted-cut certificates and verified extraction circuits}
\label{app:certified-circuits}

Theorem~\ref{thm:adaptive-instrument-cut-bound} bounds the depth of a circuit by the intrinsic cut cost of one hardware region divided by the size of its boundary.
A hardware graph with several narrow places may not be governed by any single region.
This appendix combines several regions into one certificate, which is checkable in exact arithmetic, and applies it to extraction circuits for published codes.

\subsection{Weighted-cut certificates}
\label{subsec:weighted-cut-certificates}

Let $G_{\hw}=(V_{\hw},E_{\hw})$ be the hardware graph and let $Q\subseteq V_{\hw}$ hold the data.
A certificate consists of hardware regions $W_1,\ldots,W_m\subseteq V_{\hw}$ with weights $w_1,\ldots,w_m\ge0$, together with a potential $y_v\ge0$ on every vertex $v$.
It is feasible if every hardware edge $uv$ satisfies
\begin{equation}
\sum_{i\,:\,uv\in\partial W_i}w_i\le y_u+y_v .
\label{eq:certificate-edge-inequality}
\end{equation}
For a circuit $C$, let $T_2(C)$ be the largest number, over all nonzero branches, of layers that contain a two-qubit operation.

\begin{proposition}[Weighted-cut certificate]
\label{prop:weighted-cut-certificate}
Let $(W_i,w_i,y)$ be a feasible certificate with $\sum_vy_v>0$.
Every exact circuit $C$ in Definition~\ref{def:adaptive-local-extraction} satisfies
\begin{equation}
\depth(C)\ge T_2(C)\ge\frac{\sum_iw_i\,\chi_{Q\cap W_i}(\cS)}{2\sum_vy_v}.
\label{eq:weighted-cut-general}
\end{equation}
If every two-qubit operation of $C$ is a CNOT, a CZ, or a projective measurement of one two-qubit Pauli operator, the factor $2$ in the denominator can be removed.
\end{proposition}

\begin{proof}
Fix a nonzero complete Kraus trajectory $\beta$ of $C$.
For each region $W_i$, Theorem~\ref{thm:adaptive-instrument-cut-bound} gives $\chi_{Q\cap W_i}(\cS)\le2g_{W_i}(\beta)$, where $g_{W_i}(\beta)$ counts the two-qubit operations of $\beta$ on edges of $\partial W_i$.
Multiply the $i$th inequality by $w_i$ and sum over $i$.
The left side becomes $\sum_iw_i\chi_{Q\cap W_i}(\cS)$.
On the right side, a two-qubit operation on the edge $uv$ is counted with weight $2\sum_{i:uv\in\partial W_i}w_i$, which is at most $2(y_u+y_v)$ by \eqref{eq:certificate-edge-inequality}.
Within one layer the two-qubit operations act on disjoint pairs of sites, so the potentials of all their endpoints sum to at most $\sum_vy_v$.
Only layers that contain a two-qubit operation contribute, and the trajectory passes through at most $T_2(C)$ of them.
Hence $\sum_iw_i\chi_{Q\cap W_i}(\cS)\le2T_2(C)\sum_vy_v$, which is \eqref{eq:weighted-cut-general}.
When every two-qubit operation has operator Schmidt rank at most two, Remark~\ref{rem:adaptive-sharp-constant} gives $\chi_{Q\cap W_i}(\cS)\le g_{W_i}(\beta)$ instead, and the same summation removes the factor $2$.
\end{proof}

A single region $W$ with weight one is the special case in which $y$ is the indicator of a vertex cover of the boundary edges.
The boundary edges form a bipartite graph between $W$ and its complement, so by K\H{o}nig's theorem the smallest cover has $\nu_W$ vertices, and the proposition returns the matching bound $\chi_{Q\cap W}/(2\nu_W)$ of Section~\ref{sec:submodular-cd}.
Feasible certificates form a linear program, and a numerical solver can propose weights and potentials.
Any feasible pair with rational entries is a complete proof of its bound, however it was found.
We therefore verify every certificate below with exact rational arithmetic on every hardware edge and recompute every $\chi_{Q\cap W_i}$ with an independent elimination over $\mathbb F_2$.
The search over candidate regions is finite, so a certificate below does not show that a stronger certificate is impossible.

\subsection{Verified circuits on fixed hardware graphs}
\label{subsec:verified-circuits}

We applied the certificates to syndrome-extraction circuits for twelve published codes on fifteen hardware graphs.
Each case fixes the published check list of a code, a hardware graph, the data sites, and an explicit circuit on that graph.
Table~\ref{tab:verified-circuits} lists the certified lower bounds beside the layer counts of the circuits.

The codes are the $[[72,12,6]]$ and $[[108,8,10]]$ bivariate bicycle codes of Bravyi et al.\ \cite{bravyiHighthresholdLowoverheadFaulttolerant2024}.
They include the hypergraph product codes with $45$ and $58$ data qubits and the lifted product code with $40$ data qubits of Nguyen, Rimbach-Russ, and Bosco \cite{nguyenDisassemblingQLDPCCodes2026}.
They include the quantum Tanner codes with $60$ and $144$ data qubits of Leverrier, Rozendaal, and Z\'emor \cite{leverrierSmallQuantumTanner2025}.
The remaining codes are the Steane code and the triangular color codes with $19$ and $37$ data qubits in the presentation of Zhang et al.\ \cite{zhangOptimalCompilationSyndrome2026}.
The rotated surface codes of distances $3$ and $5$ serve as controls.
The bivariate bicycle, hypergraph product, lifted product, $60$-qubit quantum Tanner, and Steane circuits on Tanner graphs reproduce the schedules of their sources.
The color-code circuits are our alternating baselines rather than the six-layer schedules reported by Zhang et al., and the $144$-qubit quantum Tanner circuit comes from our own bounded search.
The bus-grid and proxy circuits are our implementations of the cited constructions.
Four kinds of hardware graph appear.
The Tanner incidence graph has one ancilla site for each listed check, adjacent to the data qubits of that check.
The local grid of the rotated surface code places each check ancilla between its data qubits.
The bus grid is the two-sublattice construction of DBT for hypergraph product codes, with bounded-length bus edges and adaptive parity gadgets \cite[Proposition~6 and Algorithm~3]{delfosseBoundsStabilizerMeasurement2021}.
The proxy graph gives every site of the $[[72,12,6]]$ Tanner graph two proxy qubits that group its six incidences into threes.
It follows the proxy idea of Vittal et al.\ \cite{vittalFlagProxyNetworks2024} and lowers the hardware degree from six to four.

Every circuit was verified to implement the exact syndrome instrument on arbitrary inputs, including inputs entangled with an untouched reference.
The verification applies the circuit to one half of $n$ Bell pairs with coherent measurement records and compares the resulting signed stabilizer state with an independently constructed controlled-Pauli reference instrument.
Equality of the two Choi isometries fixes the data channel, the outcome labels, and the normalization of every branch, including the branches of redundant checks.
For the adaptive bus gadgets, every gadget was dilated and compared with the Choi state of its intended unitary, after which the assembled circuit passed the same instrument check.
The stabilizer arithmetic uses the Stim simulator \cite{gidneyStim2021}.

The counts in Table~\ref{tab:verified-circuits} follow the model of Section~\ref{sec:models}.
Ancillas start in product states, classical processing is free, and no pre-shared entanglement or free permutation is used.
The CNOT column counts layers containing CNOTs.
The two-qubit column also counts layers whose only two-qubit operations are individual two-qubit Pauli measurements, which the bus gadgets use.
The total column adds the layers of preparation, readout, and conditional Pauli corrections that the circuit requires.
The general bound is \eqref{eq:weighted-cut-general}, and the rank-two bound is the version without the factor $2$, which applies because every two-qubit operation in these circuits is a CNOT or an individual two-qubit Pauli measurement.
Both bounds count layers that contain a two-qubit operation, so both apply to the two-qubit column and to the total column.

\begin{table}[t]
\centering
\caption{Certified lower bounds and verified layer counts for extraction circuits of published codes.
Within a row, the bounds and the circuit use the same data placement, hardware graph, and ancilla budget, and $N$ counts all hardware sites.
The general bound holds for arbitrary two-qubit instruments and bounds the last two columns.
The rank-two bound holds when every two-qubit operation is a CNOT, a CZ, or an individual two-qubit Pauli measurement, and bounds the last two columns.
The surface-code totals use native $X$-basis preparation and readout, and rise to eight with explicit Hadamard layers.}
\label{tab:verified-circuits}
\small
\begin{tabular}{@{}llrrrrrr@{}}
\toprule
Code $[[n,k]]$ & Hardware & $N$ & General & Rank-two & CNOT & Two-qubit & Total \\
\midrule
BB $[[72,12]]$ & Tanner & 144 & 1 & 2 & 7 & 7 & 9 \\
BB $[[108,8]]$ & Tanner & 216 & 1 & 2 & 7 & 7 & 9 \\
BB $[[72,12]]$ & proxy, degree four & 432 & 1 & 2 & 35 & 35 & 37 \\
HGP $[[45,9]]$ & Tanner & 81 & 2 & 3 & 7 & 7 & 9 \\
HGP $[[45,9]]$ & bus grid & 162 & 4 & 8 & 18 & 54 & 148 \\
HGP $[[58,16]]$ & Tanner & 100 & 2 & 3 & 8 & 8 & 10 \\
HGP $[[58,16]]$ & bus grid & 200 & 4 & 8 & 20 & 60 & 164 \\
LP $[[40,2]]$ & Tanner & 80 & 1 & 2 & 4 & 4 & 6 \\
QT $[[60,2]]$ & Tanner & 120 & 1 & 2 & 12 & 12 & 14 \\
QT $[[144,12]]$ & Tanner & 288 & 1 & 2 & 15 & 15 & 17 \\
Color $[[19,1]]$ & Tanner & 37 & 2 & 3 & 12 & 12 & 14 \\
Color $[[37,1]]$ & Tanner & 73 & 2 & 3 & 12 & 12 & 14 \\
Steane $[[7,1]]$ & Tanner & 13 & 2 & 3 & 6 & 6 & 8 \\
Surface $[[9,1]]$ & local grid & 17 & 2 & 3 & 4 & 4 & 6 \\
Surface $[[25,1]]$ & local grid & 49 & 2 & 3 & 4 & 4 & 6 \\
\bottomrule
\end{tabular}
\end{table}

\subsection{Reading the table}
\label{subsec:reading-verified-circuits}

No verified circuit violates a certified bound, and the table shows where the bound binds and where it does not.
On a Tanner incidence graph, every listed check is adjacent to all of its data qubits, so every hardware region has a wide boundary.
The certified general bounds there are one or two layers and the rank-two bounds two or three, while the circuits use between four and fifteen CNOT layers.
Those depths are instead explained by the incidence of the listed checks, since a schedule that applies one CNOT per check--data incidence needs at least as many layers as the largest incidence.
The lifted product code, the $[[58,16]]$ hypergraph product code, and the Steane code attain that incidence bound, as do the surface codes on their grids.
Theorem~\ref{thm:adaptive-instrument-cut-bound} does not exclude shallower circuits on these graphs, because it allows changing generators and measuring during the circuit.
Whether any such circuit beats the incidence bound of a listed check list on its own Tanner graph is open.

On the local grids, the bounds come closer to the circuits.
The surface codes have a rank-two bound of three against four CNOT layers.
The bus grids separate the two hardware graphs of the same code most clearly.
Moving the $[[58,16]]$ code from its Tanner graph to the bus grid raises the certified rank-two bound from three to eight and the verified two-qubit depth from eight to sixty.
The remaining gaps on the bus grids, by factors of about seven in two-qubit layers and twenty in total layers, are open.
Theorem~\ref{thm:adaptive-grid-width-upper} concerns the scaling of the depth with $n/\sqrt N$ for a fixed family.
At these sizes the constants are set by the gadgets that route parities along the buses rather than by any cut.

Two qualifications apply to every row.
The circuits and bounds are noiseless and exact, so the table says nothing about fault tolerance.
The region search behind each certificate is finite, so a low certified bound does not show that a stronger certificate is impossible.

\clearpage
\phantomsection
\addcontentsline{toc}{section}{References}
\begingroup
\interlinepenalty=10000
\bibliographystyle{manuscript-alpha}
\bibliography{manuscript-references}

@article{akibueOptimizingEntanglementManipulation2026,
  title = {Optimizing Entanglement Manipulation via Algebraic--Geometric Decompositions and Semi-Definite Programming Hierarchies},
  author = {Akibue, Seiseki and Miyazaki, Jisho and Osaka, Hiroyuki},
  year = 2026,
  month = aug,
  journal = {Letters in Mathematical Physics},
  volume = {116},
  number = {5},
  pages = {110},
  doi = {10.1007/s11005-026-02138-9},
  url = {https://doi.org/10.1007/s11005-026-02138-9},
}

@inproceedings{alonSeparatorTheoremGraphs1990,
  title = {A Separator Theorem for Graphs with an Excluded Minor and Its Applications},
  booktitle = {Proceedings of the Twenty-Second Annual {{ACM}} Symposium on {{Theory}} of Computing  - {{STOC}} '90},
  author = {Alon, N. and Seymour, P. and Thomas, R.},
  year = 1990,
  pages = {293--299},
  publisher = {ACM Press},
  address = {Baltimore, Maryland, United States},
  doi = {10.1145/100216.100254},
  url = {http://portal.acm.org/citation.cfm?doid=100216.100254},
}

@article{audenaertEntanglementMixedStabilizer2005,
  title = {Entanglement on Mixed Stabilizer States: Normal Forms and Reduction Procedures},
  author = {Audenaert, Koenraad M R and Plenio, Martin B},
  year = 2005,
  month = aug,
  journal = {New Journal of Physics},
  volume = {7},
  number = {1},
  pages = {170},
  doi = {10.1088/1367-2630/7/1/170},
  url = {https://doi.org/10.1088/1367-2630/7/1/170},
}

@article{baconSparseQuantumCodes2017b,
  title = {Sparse Quantum Codes from Quantum Circuits},
  author = {Bacon, Dave and Flammia, Steven T. and Harrow, Aram W. and Shi, Jonathan},
  year = 2017,
  month = apr,
  journal = {IEEE Transactions on Information Theory},
  volume = {63},
  number = {4},
  pages = {2464--2479},
  doi = {10.1109/TIT.2017.2663199},
  url = {https://ieeexplore.ieee.org/document/7839955},
}

@article{bandyopadhyayEntanglementCostNonlocal2009,
  author = {Bandyopadhyay, Somshubhro and Brassard, Gilles and Kimmel, Shelby and Wootters, William K.},
  title = {Entanglement Cost of Nonlocal Measurements},
  journal = {Physical Review A},
  volume = {80},
  number = {1},
  pages = {012313},
  year = {2009},
  doi = {10.1103/PhysRevA.80.012313},
  eprint = {0809.2264},
  archivePrefix = {arXiv},
  url = {https://doi.org/10.1103/PhysRevA.80.012313},
}

@misc{baspinLowerBoundOverhead2023,
  title = {A Lower Bound on the Overhead of Quantum Error Correction in Low Dimensions},
  author = {Baspin, Nou{\'e}dyn and Fawzi, Omar and Shayeghi, Ala},
  year = 2023,
  month = feb,
  number = {arXiv:2302.04317},
  eprint = {2302.04317},
  publisher = {arXiv},
  doi = {10.48550/arXiv.2302.04317},
  url = {http://arxiv.org/abs/2302.04317},
  archiveprefix = {arXiv},
}

@article{beckmanCausalLocalizableQuantum2001,
  title = {Causal and Localizable Quantum Operations},
  author = {Beckman, David and Gottesman, Daniel and Nielsen, M. A. and Preskill, John},
  year = 2001,
  month = oct,
  journal = {Physical Review A},
  volume = {64},
  number = {5},
  pages = {052309},
  publisher = {American Physical Society},
  doi = {10.1103/PhysRevA.64.052309},
  url = {https://link.aps.org/doi/10.1103/PhysRevA.64.052309},
}

@misc{bhardwajHighrateQLDPCProcessors2026,
  title = {High-Rate {{qLDPC}} Processors},
  year = 2026,
  month = jul,
  number = {arXiv:2607.28795},
  eprint = {2607.28795},
  publisher = {arXiv},
  doi = {10.48550/arXiv.2607.28795},
  url = {http://arxiv.org/abs/2607.28795},
  archiveprefix = {arXiv},
  author = {Bhardwaj, Aditya and Ma, Muzhou and Meister, Nadine and King, Robbie and Bluvstein, Dolev and Preskill, John and Cain, Madelyn and Xu, Qian and Huang, Hsin-Yuan},
}

@article{bombinSingleshotFaulttolerantQuantum2015,
  title = {Single-Shot Fault-Tolerant Quantum Error Correction},
  author = {Bomb{\'i}n, H{\'e}ctor},
  year = 2015,
  month = sep,
  journal = {Physical Review X},
  volume = {5},
  number = {3},
  pages = {031043},
  publisher = {American Physical Society},
  doi = {10.1103/PhysRevX.5.031043},
  url = {https://link.aps.org/doi/10.1103/PhysRevX.5.031043},
}

@article{bravyiHighthresholdLowoverheadFaulttolerant2024,
  title = {High-Threshold and Low-Overhead Fault-Tolerant Quantum Memory},
  year = 2024,
  month = mar,
  journal = {Nature},
  volume = {627},
  number = {8005},
  pages = {778--782},
  publisher = {Nature Publishing Group},
  doi = {10.1038/s41586-024-07107-7},
  url = {https://www.nature.com/articles/s41586-024-07107-7},
  author = {Bravyi, Sergey and Cross, Andrew W. and Gambetta, Jay M. and Maslov, Dmitri and Rall, Patrick and Yoder, Theodore J.},
}

@article{bravyiTradeoffsReliableQuantum2010,
  title = {Tradeoffs for Reliable Quantum Information Storage in {{2D}} Systems},
  author = {Bravyi, Sergey and Poulin, David and Terhal, Barbara},
  year = 2010,
  month = feb,
  journal = {Physical Review Letters},
  volume = {104},
  number = {5},
  pages = {050503},
  publisher = {American Physical Society},
  doi = {10.1103/PhysRevLett.104.050503},
  url = {https://link.aps.org/doi/10.1103/PhysRevLett.104.050503},
}

@article{campbellTheorySingleshotError2019,
  title = {A Theory of Single-Shot Error Correction for Adversarial Noise},
  author = {Campbell, Earl T},
  year = 2019,
  month = feb,
  journal = {Quantum Science and Technology},
  volume = {4},
  number = {2},
  pages = {025006},
  publisher = {IOP Publishing},
  doi = {10.1088/2058-9565/aafc8f},
  url = {https://doi.org/10.1088/2058-9565/aafc8f},
}

@misc{chandraDistributedQuantumError2026,
  title = {Distributed Quantum Error Correction with Bivariate Bicycle Codes in a Modular Architecture},
  author = {Chandra, Nitish Kumar and Kaur, Eneet and Nejabati, Reza and Seshadreesan, Kaushik P.},
  year = 2026,
  month = may,
  number = {arXiv:2605.04663},
  eprint = {2605.04663},
  publisher = {arXiv},
  doi = {10.48550/arXiv.2605.04663},
  url = {http://arxiv.org/abs/2605.04663},
  archiveprefix = {arXiv},
}

@misc{delfosseBoundsStabilizerMeasurement2021,
  title = {Bounds on Stabilizer Measurement Circuits and Obstructions to Local Implementations of Quantum {{LDPC}} Codes},
  author = {Delfosse, Nicolas and Beverland, Michael E. and Tremblay, Maxime A.},
  year = 2021,
  month = sep,
  number = {arXiv:2109.14599v1},
  eprint = {2109.14599v1},
  publisher = {arXiv},
  doi = {10.48550/arXiv.2109.14599},
  url = {https://arxiv.org/abs/2109.14599v1},
  archiveprefix = {arXiv},
}

@article{eisertOptimalLocalImplementation2000,
  title = {Optimal Local Implementation of Nonlocal Quantum Gates},
  author = {Eisert, J. and Jacobs, K. and Papadopoulos, P. and Plenio, M. B.},
  year = 2000,
  month = oct,
  journal = {Physical Review A},
  volume = {62},
  number = {5},
  pages = {052317},
  publisher = {American Physical Society},
  doi = {10.1103/PhysRevA.62.052317},
  url = {https://link.aps.org/doi/10.1103/PhysRevA.62.052317},
}

@misc{fattalEntanglementStabilizerFormalism2004,
  title = {Entanglement in the Stabilizer Formalism},
  author = {Fattal, David and Cubitt, Toby S. and Yamamoto, Yoshihisa and Bravyi, Sergey and Chuang, Isaac L.},
  year = 2004,
  month = jun,
  number = {arXiv:quant-ph/0406168},
  eprint = {quant-ph/0406168},
  publisher = {arXiv},
  doi = {10.48550/arXiv.quant-ph/0406168},
  url = {http://arxiv.org/abs/quant-ph/0406168},
  archiveprefix = {arXiv},
}

@inproceedings{fawziConstantOverheadQuantum2018,
  title = {Constant Overhead Quantum Fault-Tolerance with Quantum Expander Codes},
  booktitle = {2018 {{IEEE}} 59th {{Annual Symposium}} on {{Foundations}} of {{Computer Science}} ({{FOCS}})},
  author = {Fawzi, Omar and Grospellier, Antoine and Leverrier, Anthony},
  year = 2018,
  month = oct,
  pages = {743--754},
  doi = {10.1109/FOCS.2018.00076},
  url = {https://ieeexplore.ieee.org/document/8555154},
}

@article{gidneyStim2021,
  title = {Stim: A Fast Stabilizer Circuit Simulator},
  author = {Gidney, Craig},
  journal = {Quantum},
  volume = 5,
  pages = 497,
  year = 2021,
  doi = {10.22331/q-2021-07-06-497},
}

@article{gottesmanFaulttolerantQuantumComputation2014,
  title = {Fault-Tolerant Quantum Computation with Constant Overhead},
  author = {Gottesman, Daniel},
  year = 2014,
  month = nov,
  journal = {Quantum Information and Computation},
  volume = {14},
  number = {15\&16},
  pages = {1339--1371},
  doi = {10.26421/QIC14.15-16-5},
  url = {http://www.rintonpress.com/journals/doi/QIC14.15-16-5.html},
}

@article{guSingleshotDecodingGood2024,
  title = {Single-Shot Decoding of Good Quantum {{LDPC}} Codes},
  year = 2024,
  month = mar,
  journal = {Communications in Mathematical Physics},
  volume = {405},
  number = {3},
  pages = {85},
  doi = {10.1007/s00220-024-04951-6},
  url = {https://doi.org/10.1007/s00220-024-04951-6},
  author = {Gu, Shouzhen and Tang, Eugene and Caha, Libor and Choe, Shin Ho and He, Zhiyang and others},
}

@phdthesis{kimConditionalIndependenceQuantum2013,
  title = {Conditional Independence in Quantum Many-Body Systems},
  author = {Kim, Isaac Hyun},
  year = 2013,
  month = may,
  doi = {10.7907/PZJN-A841},
  url = {https://resolver.caltech.edu/CaltechTHESIS:05102013-172241867},
  school = {California Institute of Technology},
}

@article{knillTheoryQuantumErrorcorrecting1997,
  title = {Theory of Quantum Error-Correcting Codes},
  author = {Knill, Emanuel and Laflamme, Raymond},
  year = 1997,
  month = feb,
  journal = {Physical Review A},
  volume = {55},
  number = {2},
  pages = {900--911},
  publisher = {American Physical Society},
  doi = {10.1103/PhysRevA.55.900},
  url = {https://link.aps.org/doi/10.1103/PhysRevA.55.900},
}

@inproceedings{leverrierQuantumTannerCodes2022,
  title = {Quantum {Tanner} codes},
  booktitle = {2022 {{IEEE}} 63rd {{Annual Symposium}} on {{Foundations}} of {{Computer Science}} ({{FOCS}})},
  author = {Leverrier, Anthony and Z{\'e}mor, Gilles},
  year = 2022,
  month = oct,
  pages = {872--883},
  publisher = {IEEE},
  address = {Denver, CO, USA},
  doi = {10.1109/FOCS54457.2022.00117},
  url = {https://ieeexplore.ieee.org/document/9996782/},
}

@misc{leverrierSmallQuantumTanner2025,
  title = {Small Quantum {{Tanner}} Codes from Left-Right {{Cayley}} Complexes},
  author = {Leverrier, Anthony and Rozendaal, Wouter and Z{\'e}mor, Gilles},
  year = 2025,
  month = dec,
  number = {arXiv:2512.20532},
  eprint = {2512.20532},
  publisher = {arXiv},
  doi = {10.48550/arXiv.2512.20532},
  url = {http://arxiv.org/abs/2512.20532},
  archiveprefix = {arXiv},
}

@article{limTradeoffInformationGain2025,
  title = {Trade-off between Information Gain and Disturbance in Local Discrimination of Entangled Quantum States},
  author = {Lim, Youngrong and Hhan, Minki and Kwon, Hyukjoon},
  year = 2025,
  month = mar,
  journal = {Quantum Science and Technology},
  volume = {10},
  number = {2},
  pages = {025048},
  publisher = {IOP Publishing},
  doi = {10.1088/2058-9565/adc034},
  url = {https://doi.org/10.1088/2058-9565/adc034},
}

@misc{nguyenDisassemblingQLDPCCodes2026,
  title = {Disassembling {{qLDPC}} Codes for Depth-Optimal Parity-Check Circuits},
  author = {Nguyen, Minh T. P. and {Rimbach-Russ}, Maximilian and Bosco, Stefano},
  year = 2026,
  month = aug,
  number = {arXiv:2608.19917},
  eprint = {2608.19917},
  publisher = {arXiv},
  doi = {10.48550/arXiv.2608.19917},
  url = {http://arxiv.org/abs/2608.19917},
  archiveprefix = {arXiv},
}

@article{nielsenQuantumDynamicsPhysical2003,
  title = {Quantum Dynamics as a Physical Resource},
  year = 2003,
  month = may,
  journal = {Physical Review A},
  volume = {67},
  number = {5},
  pages = {052301},
  publisher = {American Physical Society},
  doi = {10.1103/PhysRevA.67.052301},
  url = {https://link.aps.org/doi/10.1103/PhysRevA.67.052301},
  author = {Nielsen, Michael A. and Dawson, Christopher M. and Dodd, Jennifer L. and Gilchrist, Alexei and Mortimer, Duncan and others},
}

@inproceedings{panteleevAsymptoticallyGoodQuantum2022,
  title = {Asymptotically Good Quantum and Locally Testable Classical {{LDPC}} Codes},
  booktitle = {Proceedings of the 54th {{Annual ACM SIGACT Symposium}} on {{Theory}} of {{Computing}}},
  author = {Panteleev, Pavel and Kalachev, Gleb},
  year = 2022,
  month = jun,
  series = {{{STOC}} 2022},
  pages = {375--388},
  publisher = {Association for Computing Machinery},
  address = {New York, NY, USA},
  doi = {10.1145/3519935.3520017},
  url = {https://dl.acm.org/doi/10.1145/3519935.3520017},
}

@misc{shawNetworkedRealizationQuantum2026,
  title = {Networked Realization of Quantum {{LDPC}} Codes},
  author = {Shaw, Swayangprabha and Rengaswamy, Narayanan},
  year = 2026,
  month = apr,
  number = {arXiv:2604.25026},
  eprint = {2604.25026},
  publisher = {arXiv},
  doi = {10.48550/arXiv.2604.25026},
  url = {http://arxiv.org/abs/2604.25026},
  archiveprefix = {arXiv},
}

@article{terhalSchmidtNumberDensity2000,
  author = {Terhal, Barbara M. and Horodecki, Pawe{\l}},
  title = {{Schmidt} number for density matrices},
  journal = {Physical Review A},
  volume = {61},
  number = {4},
  pages = {040301(R)},
  year = {2000},
  month = mar,
  doi = {10.1103/PhysRevA.61.040301},
  eprint = {quant-ph/9911117v4},
  archivePrefix = {arXiv},
  url = {https://doi.org/10.1103/PhysRevA.61.040301},
}

@inproceedings{vittalFlagProxyNetworks2024,
  title = {Flag-Proxy Networks: Overcoming the Architectural, Scheduling and Decoding Obstacles of Quantum {{LDPC}} Codes},
  author = {Vittal, Suhas and {Javadi-Abhari}, Ali and Cross, Andrew W. and Bishop, Lev S. and Qureshi, Moinuddin},
  booktitle = {2024 57th {{IEEE}}/{{ACM International Symposium}} on {{Microarchitecture}} ({{MICRO}})},
  year = 2024,
  eprint = {2409.14283},
  archiveprefix = {arXiv},
  url = {http://arxiv.org/abs/2409.14283},
}

@article{yamasakiGraphAssociatedEntanglement2017,
  author = {Yamasaki, Hayata and Soeda, Akihito and Murao, Mio},
  title = {Graph-associated entanglement cost of a multipartite state in exact and finite-block-length approximate constructions},
  journal = {Physical Review A},
  volume = {96},
  pages = {032330},
  year = {2017},
  doi = {10.1103/PhysRevA.96.032330},
  eprint = {1705.00006},
  archivePrefix = {arXiv},
  url = {https://arxiv.org/abs/1705.00006v2},
}

@misc{zhangOptimalCompilationSyndrome2026,
  title = {Optimal Compilation of Syndrome Extraction Circuits for General Quantum {{LDPC}} Codes},
  year = 2026,
  month = mar,
  journal = {arXiv.org},
  url = {https://arxiv.org/abs/2603.21499v1},
  author = {Zhang, Kai and Gao, Dingchao and Yang, Zhaohui and Zhou, Runshi and Liu, Fangming and others},
}

@article{zhaoGraphbasedApproachEntanglement2026,
  title = {Graph-Based Approach to Entanglement Entropy of Quantum Error-Correcting Codes},
  author = {Zhao, Wuxu and Fang, Menglong and Su, Daiqin},
  year = 2026,
  month = may,
  journal = {Physical Review A},
  volume = {113},
  number = {5},
  pages = {052412},
  publisher = {American Physical Society},
  doi = {10.1103/jc6b-txy9},
  url = {https://link.aps.org/doi/10.1103/jc6b-txy9},
}
\endgroup
\end{document}